\documentclass[a4paper, reqno, 10pt]{amsart}
\usepackage[left=3cm, right=3cm, top=4cm, bottom=3cm]{geometry}

\usepackage{helvet}
\usepackage[T1]{fontenc}
\usepackage[utf8]{inputenc}
\usepackage{amsmath,amssymb,amsthm,mathtools}
\usepackage{graphicx}
\usepackage{float}
\usepackage{tikz}
\usetikzlibrary{arrows.meta,calc,decorations.pathreplacing}
\usepackage{microtype}
\numberwithin{equation}{section}
\usepackage{hyperref}
\hypersetup{colorlinks=true,linkcolor=blue,citecolor=blue,urlcolor=blue}

\newtheorem{theorem}{Theorem}[section]
\newtheorem{proposition}[theorem]{Proposition}
\newtheorem{lemma}[theorem]{Lemma}
\newtheorem{corollary}[theorem]{Corollary}
\theoremstyle{definition}
\newtheorem{definition}[theorem]{Definition}
\newtheorem{remark}[theorem]{Remark}

\newcommand{\N}{\mathbb{N}}
\newcommand{\R}{\mathbb{R}}
\newcommand{\C}{\mathbb{C}}

\newcommand{\Dom}{\operatorname{dom}}
\newcommand{\Ran}{\operatorname{ran}}
\newcommand{\Ker}{\operatorname{ker}}
\newcommand{\dd}{\mathrm{d}}
\newcommand{\pa}{\partial}
\newcommand{\la}{\lambda}
\newcommand{\kap}{\kappa}
\newcommand{\tr}{\operatorname{tr}}
\newcommand{\supp}{\operatorname{supp}}
\newcommand{\dist}{\operatorname{dist}}
\newcommand{\Lip}{\operatorname{Lip}}

\begin{document}

\title[Point interactions on unbounded domains]{The Dirichlet Laplacian with a point interaction\\ 
on unbounded Lipschitz domains}

\author[D.~Noja]{Diego Noja}
\address[D.~Noja]{Department of Mathematics and Applications, University of Milano Bicocca}\address{Via Roberto Cozzi 55, 20125 Milano, Italy}
\email{diego.noja@unimib.it}

\author[F.~Raso Stoia]{Francesco Raso Stoia}
\address[F.~Raso Stoia]{Independent researcher}
\email{f.rasostoiawork@gmail.com}

\date{\today}

\subjclass[2020]{Primary 35J10, 81Q10; Secondary 47B25, 35P05, 35B34, 81U24}

\keywords{Point interactions; self-adjoint extensions; Dirichlet Laplacian; spectral theory; threshold resonances; unbounded domains}


\begin{abstract}
We study one-centre point interactions for the Dirichlet Laplacian on
unbounded domains in dimensions two and three, with emphasis on exterior
domains and special Lipschitz domains. These operators are singular perturbations constructed as
self-adjoint extensions of the Dirichlet Laplacian restricted to functions
vanishing at the interaction centre, and their resolvents are given by an explicit Kre\u{\i}n formula with a single extension parameter $\alpha$. 
The negative spectrum is completely characterized by
a scalar equation and the critical coupling $\alpha$ separating binding from non-binding
is the threshold limit of the Weyl function appearing in the Kre\u{\i}n formula. 
We establish domain monotonicity of the
Weyl function, of the critical coupling, and of the unique negative eigenvalue when existing, and
we derive sharp near-boundary asymptotics of the critical coupling in
uniformly $C^{1,1}$ geometries. These estimates imply that, for every fixed
coupling, nonpositive spectrum disappears when the interaction centre
approaches the Dirichlet boundary. We also prove limiting absorption
principles and purely absolutely continuous positive spectrum for a point-interaction in exterior 
domains case and in classes of special Lipschitz domains. We also 
analyze in depth several threshold phenomena. We show that the critical coupling is
governed by the far-field behavior of the zero-energy Green function:
exterior domains give threshold resonances, domains contained in a
three-dimensional half-space give threshold eigenvalues, the half-plane gives
a $p$-wave resonance, and planar wedges exhibit types of threshold states that are aperture-dependent.
Finally, low-energy resolvent expansions are computed in the model cases and persistence or disappearance of eigenvalues at threshold are studied.
The present paper seems to be the first systematic work on the subject.
\end{abstract}

\maketitle

\begingroup
\small
\tableofcontents
\endgroup

\section*{Introduction}

Point interactions are among the best-known solvable models in mathematical physics. They are singular perturbations of the Laplacian, describing interactions concentrated around one or more points and not representable by ordinary Schr\"odinger operators of the form $-\Delta+V$. They arise both in classical and quantum models. At the classical level they appear in the interaction of continua and fields with point sources; see, for example, \cite{ACCS19, Andronov, NP98,NP99, MPS22, MP24} and references therein. In quantum mechanics they provide effective descriptions of particles subjected to short-range interactions. At low energy, or equivalently when the relevant wavelength is much larger than the range of the interaction, the microscopic shape of the potential becomes irrelevant and the effective model is a zero-range Hamiltonian. In this regime the scattering length is the only parameter that survives in the leading description. See the monograph \cite{AGHH}, whose bibliography covers the literature up to 2005, and the recent work \cite{DerezinskiPoint}. In dimensions $d\leq3$, point interactions are rigorously realized as self-adjoint extensions of the Laplacian restricted away from the interaction support, a finite or infinite set of points, and their resolvents are described by Kre\u{\i}n-type formulas. The mentioned references cover the whole-space theory in $\R^d$, $d=1,2,3$, which is classical.

The corresponding theory on proper subdomains of $\R^d$ is considerably less developed. Point interactions on bounded domains and the dependence of the principal eigenvalue on the interaction centre were studied in \cite{BlanchardFigariMantile,ExnerMantile2008, Posilicano2013, LotoMiche21}. For the approximation of point interactions on bounded domains see \cite{NojaScandone25} and references therein. The case of compact manifolds without boundary occurs in the study of quantum ergodicity and \v{S}eba billiards; see, for instance, \cite{KurlbergUebershaer14,KurlbergUebershaer23}. For hyperbolic manifolds of finite volume, see the influential works \cite{YCdV82,YCdV83}. On unbounded domains with boundary, the available results concern waveguides, treated in \cite[Chapter~5]{EK2015} and the few papers cited there, and the half-space; see \cite{ACSZ08,NojaRasoStoia25}. The unbounded setting appears nevertheless natural, since the effect of the singular perturbation depends both on the boundary and on the geometry of the domain at infinity.

The purpose of this paper is to begin and develop a systematic theory for point interactions on unbounded domains in dimensions two and three. Here we restrict to the case of a single centre. We focus on two complementary classes: exterior domains, which have compact boundary, and special Lipschitz domains, whose boundary is a global noncompact graph. The only one-dimensional counterpart is the half-line problem with a delta potential, usually called the Winter model, and it is not considered here; see, for example, \cite{DM:SemiclassicalResonanceAsymptoticDeltaHalfLine,Sacchetti23}. From the point of view of the classification of unbounded domains, exterior and special Lipschitz domains are quasi-conical domains, that is, they contain a sequence of disjoint balls with radii diverging to infinity. They are two large classes of importance in applications and represent the cases of bounded and unbounded boundary.

Let $\Omega\subset\R^d$, $d\in\{2,3\}$, be a domain in one of the above classes, let $A_D=-\Delta_D^\Omega$ be the Dirichlet Laplacian in $L^2(\Omega)$, and let $y\in\Omega$ be the interaction centre. We restrict $A_D$ to functions vanishing at $y$. This gives a closed symmetric operator $S_y$ with deficiency indices $(1,1)$ with self-adjoint extensions  parametrized by $\alpha\in \mathbb R^*$. The parameter $\alpha$ will be called the coupling constant, or briefly, coupling. The extension with $\alpha\in\R$ is denoted by $A_{\alpha,y}^\Omega$ in dimension three and by $A_{\alpha,y}^{(2),\Omega}$ in dimension two, while $\alpha=\infty$ corresponds to the unperturbed Dirichlet Laplacian. For $d>3$ the symmetric operator $S_y$ is essentially self-adjoint, so that non trivial point interactions do not exist (but see the already quoted \cite{DerezinskiPoint} for a study of the Laplacian with a point potential in higher dimension and out of the Hilbert space framework). In both dimensions $d=2,3$ the resolvent is a rank-one perturbation of the background Dirichlet resolvent and is determined by the Dirichlet Green function of the domain and a scalar function $M_y^\Omega(z)$, the Weyl function. A generic element of the operator domain of any non decoupled extension is the sum of a regular part in the domain of the Dirichlet Laplacian and a part singular at $y$. The singularity is Coulomb in dimension three and logarithmic in dimension two. Moreover, a precise relation exists between the value of the regular part at the singularity and the coefficient of the singular part. This boundary condition at the singularity, involving the parameter $\alpha$, is characteristic of the point interaction. 
After constructing and characterizing the family of self-adjoint extensions in both dimensions, we discuss the spectral properties that are the main subject of the paper. The spectral theory of the single-centre point interaction in $\R^d$, $d=2,3$, is elementary. By contrast, the presence of a boundary, already for one centre, requires a careful analysis and reveals several distinct and entirely new behaviors. The main results of the present analysis identify how the local boundary geometry and the geometry at infinity enter the spectral theory and can be synthetically collected as follows.

\begin{enumerate}
\item Singularities of the resolvent are defined by the denominator appearing in the Kre\u{\i}n formula, $\alpha-M_y^\Omega(z)$; see \eqref{eq:krein-resolvent-3d} and \eqref{eq:krein-resolvent-2d}. In particular the study of the negative spectrum is reduced exactly to the solution of the characteristic equation
\begin{equation*}
\alpha=M_y^\Omega(-\lambda),
\qquad \lambda>0.
\end{equation*}
There is at most one negative eigenvalue, and it is simple. For both exterior and special Lipschitz domains, the zero-energy limit ($\lambda\downarrow 0$) of the Weyl function exists and is finite. It therefore defines the critical coupling $\alpha_c^\Omega(y)$ separating binding from non-binding. In dimension two this entails a nontrivial cancellation between the logarithmic divergence of the free Green function and that of the harmonic correction; see Theorems~\ref{thm:critical-coupling-3d} and \ref{thm:critical-coupling-2d}.

\item Domain inclusion yields monotonicity of the Green function, the Weyl function, the critical coupling $\alpha_c^\Omega(y)$, and the unique negative eigenvalue; see Theorem~\ref{thm:domain_monotonicity_eigenvalue}. When the boundary regularity is $C^{1,1}$, uniformly in the special Lipschitz case, we also obtain the sharp asymptotics
\begin{align*}
\alpha_c^\Omega(y)
&=-\frac{1}{8\pi\,\dist(y,\partial\Omega)}+O(1),
\qquad \qquad \qquad \qquad d=3,\\
\alpha_c^{(2),\Omega}(y)
&=\frac{1}{2\pi}\log\!\bigl(2\dist(y,\partial\Omega)\bigr)
+O\!\bigl(\dist(y,\partial\Omega)\bigr),
\ \qquad d=2.
\end{align*}
These estimates imply that, for every fixed coupling and in both dimensions, the point interaction has no nonpositive eigenvalue when its centre is sufficiently close to the Dirichlet boundary and that the critical coupling can be any negative real number for $d=3$ and any real number for $d=2$; see Theorem~\ref{thm:near-boundary-alpha} and Corollary~\ref{cor:no_eigenvalue_near_boundary}. This is a sharp and quantitative example of the interaction between boundary and point interaction. We recall that in the boundary-less cases there is a negative eigenvalue if and only if $\alpha<0$ in $\mathbb R^3$ (critical coupling is zero) and for any $\alpha\in \mathbb R$ in $\mathbb R^2$ (critical coupling is $+\infty$).

\item A limiting absorption principle for the background Dirichlet Laplacian can be transferred, through the Kre\u{\i}n formula, to the point interaction. This proves that the positive spectrum is purely absolutely continuous for exterior domains and for special Lipschitz domains in several geometries, including for example domains conical at infinity; see Theorem~\ref{thm:abstract_transfer_ac}, Corollary~\ref{cor:exterior_transfer_ac}, and Theorem~\ref{thm:special_lipschitz_transfer_ac}. The strict positivity
$
\operatorname{Im}M_y^\Omega(\lambda+i0)>0
$ for $\lambda>0$ is the main point of the analysis.

\item At the critical coupling, the nature of the threshold of the continuous spectrum is governed by the far-field behavior of the zero-energy Green function. We show that three-dimensional exterior domains have a monopole threshold resonance, with leading decay $|x|^{-1}$, while planar exterior domains have an $s$-wave resonance; see Definition~\ref{def:critical-zero-energy-state}. By contrast, every three-dimensional Dirichlet domain contained in a half-space (or more generally in a domain asymptotic to a half-space at infinity) has a threshold eigenvalue, while the half-plane has a $p$-wave resonance. More generally, for a planar wedge of opening $\beta$, zero is a threshold eigenvalue if $\beta<\pi$, a $p$-wave resonance if $\beta=\pi$, and a threshold state of different type if $\beta>\pi$; see Theorems~\ref{thm:exterior_threshold_not_attained}, \ref{thm:exterior_threshold_not_attained_2d}, \ref{thm:threshold-halfspace-subdomain}, \ref{thm:threshold-asymptotic-halfspace} and \ref{thm:planar-wedge-threshold}. We also show that the singular threshold behavior is absent for $\alpha=\infty$ in the model geometries, and hence it is associated to the presence of the point singularity.

\item The low-energy singularity of the resolvent at the critical coupling reduces to the analysis of the Kre\u{\i}n denominator $\alpha-M_y^\Omega(z)$. We compute the resulting expansions for the half-space, half-plane, exterior sphere, exterior disk, and planar wedges. As an application, we study the possible continuation of the eigenvalue branch on the second Riemann sheet through the threshold (here we are inspired by the recent analysis in \cite{ChristiansenDatchevGriffin}, where different examples of this phenomenon are considered). In the two planar models, the exterior-disk eigenvalue disappears, in the sense that the eigenvalue branch does not persist as a resonance when continued to the second Riemann sheet through its $s$-wave threshold; on the contrary, the half-plane eigenvalue persists as a resonance when continued to the second Riemann sheet through its $p$-wave threshold. See Propositions~\ref{prop:rank-one-threshold-reduction}, \ref{prop:low-energy-three-dimensional-models}, \ref{prop:low-energy-two-dimensional-models}, and \ref{prop:persistence-disappearance-models}.
\end{enumerate}

These results isolate two different geometric mechanisms. The local geometry near the Dirichlet boundary controls the leading divergence of the critical coupling, whereas the geometry at infinity controls the nature of the critical threshold state and the singularity of the low-energy resolvent. Exterior balls, flat models, and planar wedges show that these mechanisms are independent: the same local boundary behavior may coexist with different threshold types at infinity. It seems to be noteworthy that several results be sharply quantitative.

 The results obtained in the present study also suggests generalizations in several directions. The most obvious are to consider finitely or infinitely many singularities, and singularities distributed on sub-manifolds; from the geometrical side it is natural to consider different domains, not limited to the quasi-conical class here studied, also including non-compact manifolds. The analysis given here is for Dirichlet boundary conditions; changing boundary conditions could significantly change some of the results, in particular as regards threshold states (see Remark \ref{nbc-halfspace}).
 Finally, we mention time dependent problems, here touched only superficially (see Remark \ref{scattering}), in particular related to the wave dynamics on domains in the presence of point sources (see for the case of the half-space, the resonance expansion given in \cite{NojaRasoStoia25}).  

The paper is organized as follows. Section~\ref{sec:preliminaries} recalls the geometric setting and the existence and properties of the Dirichlet resolvent kernel for the domains considered in the paper. Section~\ref{sec:construction} constructs the one-centre point interactions and derives their domains, actions and Kre\u{\i}n resolvent formulas; several complementary remarks are also given. Section~\ref{sec:spectral} develops the negative-spectrum theory, defines and characterizes the critical coupling, shows domain monotonicity, and discusses the absolute continuity of the point interaction on the positive axis, proving it for several important instances. Section~\ref{sec:model} treats explicit geometries with exact results and several general geometries, giving near-boundary asymptotics, the characterization of critical threshold states, and low-energy resolvent expansions with the analysis of persistence or disappearance of the eigenvalue in the model cases of half-plane and exterior of a disk. Appendix~\ref{app:auxiliary} contains the results concerning existence and properties of the Dirichlet Green kernel and the positivity of its harmonic correction used in the course of the paper.

\section{Preliminaries}\label{sec:preliminaries}

\subsection{Geometric setting}
We work throughout with one of the following classes of domains $\Omega\subset\R^d$:

\begin{itemize}
\item[(a)] a \emph{special Lipschitz domain},
\begin{equation*}
\Omega=\{(x',x_d)\in\R^{d-1}\times\R:x_d>\varphi(x')\},
\qquad
\varphi\in\Lip(\R^{d-1});
\end{equation*}

\item[(b)] an \emph{exterior Lipschitz domain},
\begin{equation*}
\Omega=\R^d\setminus K,
\end{equation*}
where $K\subset\R^d$ is a nonempty compact set with nonempty interior and Lipschitz boundary, and $\Omega$ is connected.
\end{itemize}

\begin{remark}
A special Lipschitz domain has a noncompact boundary given by one global Lipschitz graph, whereas an exterior Lipschitz domain has compact boundary. In particular, the two classes are disjoint.
\end{remark}

Whenever higher boundary regularity is required, we use the standard local graph definition; see, for instance, \cite[Chapter~7]{AdamsFournier}.

\begin{definition}
Let $k\in\N_0$ and $\alpha\in(0,1]$. A domain $\Omega\subset\R^d$ is of class $C^{k,\alpha}$ if, for every $x_0\in\partial\Omega$, there exist $r>0$, a rigid change of coordinates sending $x_0$ to the origin, and a function $\psi\in C^{k,\alpha}(B_r')$ such that, in the new coordinates $x=(x',x_d)\in\R^{d-1}\times\R$,
\begin{equation*}
\Omega\cap B_r=\{(x',x_d)\in B_r:x_d>\psi(x')\},
\end{equation*}
\begin{equation*}
\partial\Omega\cap B_r=\{(x',x_d)\in B_r:x_d=\psi(x')\},
\end{equation*}
where $B_r'\subset\R^{d-1}$ is the ball of radius $r$ centered at the origin.

For $\alpha=1$, the notation $C^{k,1}$ means that the derivatives of order $k$ are locally Lipschitz. In particular, a $C^{1,1}$-boundary is locally represented by $C^1$-graphs with Lipschitz gradient.
\end{definition}

\subsection{Free resolvent}
For $z\in\C\setminus[0,\infty)$, set
\begin{equation*}
\kappa(z):=\sqrt{-z},
\qquad
\Re\kappa(z)>0.
\end{equation*}
The free resolvent kernel of $-\Delta$ on $\R^d$ is denoted by $G_z^0(x,y)=G_z^0(x-y)$ and is given by
\begin{equation}\label{eq:free-resolvent-kernel}
G_z^0(x,y)=
\begin{cases}
\displaystyle \frac{1}{2\pi}K_0\bigl(\kappa(z)|x-y|\bigr), & d=2,\\[0.8em]
\displaystyle \frac{\mathrm e^{-\kappa(z)|x-y|}}{4\pi|x-y|}, & d=3.
\end{cases}
\end{equation}
Here $K_0$ is the modified Bessel function of the second kind, also called the Macdonald function, of order zero; see \cite[(10.25.3)]{DLMF}. Formula~\eqref{eq:free-resolvent-kernel} is standard; compare, for instance, \cite{AGHH} and \cite[Section~1.1]{DerezinskiPoint}.

\subsection{Dirichlet resolvent kernel}
Throughout the paper the $L^2$-inner product is linear in the first argument:
\begin{equation*}
(f,g)_{L^2(\Omega)}=\int_\Omega f(x)\overline{g(x)}\,\dd x.
\end{equation*}
For either class of domains, the Dirichlet Laplacian
\begin{equation*}
A_D:=-\Delta_D^\Omega
\end{equation*}
is the nonnegative self-adjoint operator associated with the closed quadratic form
\begin{equation*}
\mathfrak a_D[u,v]=\int_\Omega \nabla u\cdot\overline{\nabla v}\,\dd x,
\qquad
\Dom(\mathfrak a_D)=H_0^1(\Omega).
\end{equation*}
See, for instance, \cite[Chapter~7]{Kenig1994} and \cite[Section~2]{BehrndtRohlederStadler}. The following theorem collects the properties of its resolvent kernel that will be used below. Its proof is given in Appendix~\ref{app:dirichlet-kernel-proof}.

\begin{theorem}\label{thm:dirichlet-kernel}
Let $d\in\{2,3\}$, let $\Omega\subset\R^d$ be either an exterior Lipschitz domain or a special Lipschitz domain, and let $z\in\C\setminus[0,\infty)$. Then the following assertions hold.
\begin{enumerate}
\item[(i)] For every $y\in\Omega$ there exists a unique function $h_z^\Omega(\cdot,y)\in H^1(\Omega)$ such that
\begin{equation}\label{eq:dirichlet-correction-problem}
\begin{cases}
(-\Delta-z)h_z^\Omega(\cdot,y)=0 & \text{in }\Omega,\\
\tr h_z^\Omega(\cdot,y)=\tr G_z^0(\cdot,y) & \text{on }\partial\Omega.
\end{cases}
\end{equation}

\item[(ii)] The function
\begin{equation}\label{eq:dirichlet-kernel-domain}
G_z^\Omega(x,y):=G_z^0(x,y)-h_z^\Omega(x,y)
\end{equation}
satisfies, for each fixed $y\in\Omega$,
\begin{equation*}
(-\Delta_x-z)G_z^\Omega(x,y)=\delta_y
\end{equation*}
in the sense of distributions in $\Omega$, together with the Dirichlet condition $\tr_xG_z^\Omega(\cdot,y)=0$ on $\partial\Omega$.

\item[(iii)] The function $(x,y)\mapsto G_z^\Omega(x,y)$ is measurable on $\Omega\times\Omega$, smooth off the diagonal, and has the same diagonal singularity as the free kernel~\eqref{eq:free-resolvent-kernel}.

\item[(iv)] For every $f\in C_c^\infty(\Omega)$,
\begin{equation}\label{eq:dirichlet-kernel-representation}
(A_D-z)^{-1}f(x)=\int_\Omega G_z^\Omega(x,y)f(y)\,\dd y
\end{equation}
for almost every $x\in\Omega$. Thus $G_z^\Omega$ is the integral kernel of $(A_D-z)^{-1}$.

\item[(v)] The kernel satisfies
\begin{equation*}
G_z^\Omega(x,y)=G_z^\Omega(y,x),
\qquad
G_{\bar z}^\Omega(x,y)=\overline{G_z^\Omega(x,y)}.
\end{equation*}
\end{enumerate}
\end{theorem}

\section{One-centre point interactions on exterior and special Lipschitz domains}\label{sec:construction}

We construct the point interaction at $y\in\Omega$ as a self-adjoint extension of the Dirichlet Laplacian restricted to functions vanishing at $y$ in dimension $2$ and $3$. The dimension-dependent part enters through the diagonal singularity of the Green function and the corresponding Weyl function in the Kre\u{\i}n formula. The construction is self-consistent but streamlined, avoiding repetitions of minor details and standard procedures.

\subsection{The restricted Dirichlet Laplacian}
Let $d\in\{2,3\}$, let $\Omega\subset\R^d$ be either a special Lipschitz domain or an exterior Lipschitz domain, and fix $y\in\Omega$. We write
\begin{equation*}
A_D=-\Delta_D^\Omega,
\qquad
\Dom(A_D)=\{u\in H_0^1(\Omega):\Delta u\in L^2(\Omega)\}.
\end{equation*}

\begin{proposition}\label{prop:evaluation}
The evaluation map
\begin{equation*}
\tau_y:\Dom(A_D)\longrightarrow\C,
\qquad
\tau_yu=u(y),
\end{equation*}
is well defined and continuous with respect to the graph norm
\begin{equation*}
\|u\|_{\mathrm{gr}(A_D)}^2:=\|u\|_{L^2(\Omega)}^2+\|A_Du\|_{L^2(\Omega)}^2.
\end{equation*}
Consequently,
\begin{equation*}
S_y:=A_D\big|_{\Ker\tau_y}
\end{equation*}
is densely defined, closed, and symmetric.
\end{proposition}

\begin{proof}
Choose $r>0$ such that $\overline{B_{2r}(y)}\subset\Omega$. By interior elliptic regularity,
\begin{equation*}
\|u\|_{H^2(B_r(y))}
\leq C\bigl(\|\Delta u\|_{L^2(B_{2r}(y))}+\|u\|_{L^2(B_{2r}(y))}\bigr)
\leq C'\|u\|_{\mathrm{gr}(A_D)}
\end{equation*}
for every $u\in\Dom(A_D)$. Since $H^2(B_r(y))\hookrightarrow C^0(\overline{B_r(y)})$ for $d\leq3$,
\begin{equation*}
|u(y)|\leq C''\|u\|_{\mathrm{gr}(A_D)}.
\end{equation*}
Thus $\tau_y$ is graph-norm continuous. Its kernel is therefore graph-norm closed in $\Dom(A_D)$, so $S_y$ is closed and symmetric.

Finally, $C_c^\infty(\Omega\setminus\{y\})\subset\Dom(S_y)$ is dense in $L^2(\Omega)$. Indeed, if $\varphi\in C_c^\infty(\Omega)$, multiplying $\varphi$ by a smooth cut-off that vanishes on $B_\varepsilon(y)$ and equals one outside $B_{2\varepsilon}(y)$ gives an approximation in $L^2(\Omega)$. Hence $S_y$ is densely defined.
\end{proof}

For $z\in\rho(A_D)$ we set
\begin{equation*}
G_z^y:=G_z^\Omega(\cdot,y),
\end{equation*}
where $G_z^\Omega$ is the Dirichlet Green kernel introduced in
Theorem~\ref{thm:dirichlet-kernel}.

\begin{theorem}\label{thm:deficiency}
For every $z\in\rho(A_D)$,
\begin{equation*}
\Ran(S_y-z)=\{f\in L^2(\Omega):(f,G_{\bar z}^y)_{L^2(\Omega)}=0\}.
\end{equation*}
If $z\in\rho(A_D)\setminus\R$, then
\begin{equation*}
\Ker(S_y^*-z)=\operatorname{span}\{G_z^y\}.
\end{equation*}
In particular,
\begin{equation*}
n_+(S_y)=n_-(S_y)=1.
\end{equation*}
\end{theorem}

\begin{proof}

By Proposition~\ref{prop:evaluation}, the map
\begin{equation*}
f\mapsto ((A_D-z)^{-1}f)(y)
\end{equation*}
is a bounded linear functional on $L^2(\Omega)$. Hence, by the Riesz
representation theorem, there exists a unique vector $g_{\bar z}^y\in L^2(\Omega)$
such that
\begin{equation*}
((A_D-z)^{-1}f)(y)=(f,g_{\bar z}^y)_{L^2(\Omega)},
\qquad f\in L^2(\Omega).
\end{equation*}
Theorem~\ref{thm:dirichlet-kernel} identifies this Riesz representative with
the Dirichlet Green function. Indeed, for $f\in C_c^\infty(\Omega)$,
\begin{equation*}
((A_D-z)^{-1}f)(y)
=\int_\Omega G_z^\Omega(y,x)f(x)\,\dd x
=(f,G_{\bar z}^\Omega(\cdot,y))_{L^2(\Omega)},
\end{equation*}
and the equality extends to all $f\in L^2(\Omega)$ by density. By uniqueness
of the Riesz representative,
\begin{equation*}
g_{\bar z}^y=G_{\bar z}^\Omega(\cdot,y).
\end{equation*}
Thus
\begin{equation}\label{eq:green-vector-reproducing}
((A_D-z)^{-1}f)(y)=(f,G_{\bar z}^y)_{L^2(\Omega)},
\qquad f\in L^2(\Omega).
\end{equation}

Let $z\in\rho(A_D)$. If $f=(S_y-z)u$ with $u\in\Dom(S_y)$, then $u=(A_D-z)^{-1}f$, and \eqref{eq:green-vector-reproducing} gives
\begin{equation*}
(f,G_{\bar z}^y)_{L^2(\Omega)}=u(y)=0.
\end{equation*}
Conversely, if this scalar product vanishes and $u=(A_D-z)^{-1}f$, then $u(y)=0$, hence $u\in\Dom(S_y)$ and $f=(S_y-z)u$. This proves the range characterization.

For nonreal $z$,
\begin{equation*}
\Ker(S_y^*-z)=\Ran(S_y-\bar z)^\perp=\operatorname{span}\{G_z^y\},
\end{equation*}
and the deficiency indices are therefore $(1,1)$.
\end{proof}

For $u,v\in L^2(\Omega)$, we use the rank-one notation
\begin{equation*}
|u\rangle\langle v|f:=(f,v)_{L^2(\Omega)}u.
\end{equation*}

\subsection{Dimension three}
Assume now that $d=3$. By \eqref{eq:dirichlet-kernel-domain},
\begin{equation*}
G_z^\Omega(x,y)=\frac{\mathrm e^{-\kappa(z)|x-y|}}{4\pi|x-y|}-h_z^\Omega(x,y),
\qquad
z\in\C\setminus[0,\infty).
\end{equation*}
Since
\begin{equation*}
\frac{\mathrm e^{-\kappa(z)|x-y|}}{4\pi|x-y|}
=
\frac{1}{4\pi|x-y|}-\frac{\kappa(z)}{4\pi}+O(|x-y|)
\end{equation*}
as $x\to y$, one has
\begin{equation}\label{AsymptG}
G_z^\Omega(x,y)=\frac{1}{4\pi|x-y|}+M_y^\Omega(z)+O(|x-y|),
\qquad
x\to y,
\end{equation}
where we have introduced the three-dimensional {\em Weyl function} as
\begin{equation}\label{eq:weyl-function-3d}
M_y^\Omega(z):=-\frac{\kappa(z)}{4\pi}-h_z^\Omega(y,y).
\end{equation}
\begin{lemma}\label{lem:weyl-identity-3d}
For the Weyl function defined in \eqref{eq:weyl-function-3d} and for $z,w\in\C\setminus[0,\infty)$, $z\neq w$,
\begin{equation}\label{eq:weyl-identity-3d}
M_y^\Omega(z)-M_y^\Omega(w)=(z-w)(G_z^y,G_{\bar w}^y)_{L^2(\Omega)}.
\end{equation}
In particular,
\begin{equation*}
M_y^\Omega(\bar z)=\overline{M_y^\Omega(z)}.
\end{equation*}
\end{lemma}

\begin{proof}
The resolvent identity gives
\begin{equation*}
G_z^y-G_w^y=(z-w)(A_D-z)^{-1}G_w^y.
\end{equation*}
The singular parts of $G_z^y$ and $G_w^y$ at $y$ coincide. Using \eqref{AsymptG} to take the regularized value at $y$ and exploiting \eqref{eq:green-vector-reproducing}, we obtain
\begin{equation*}
M_y^\Omega(z)-M_y^\Omega(w)=(z-w)(G_w^y,G_{\bar z}^y)_{L^2(\Omega)}.
\end{equation*}
The last scalar product equals $(G_z^y,G_{\bar w}^y)_{L^2(\Omega)}$, which proves \eqref{eq:weyl-identity-3d}. The conjugation property follows from that of the Green kernel.
\end{proof}

\begin{theorem}\label{thm:krein}
For every $\alpha\in\R\cup\{\infty\}$ there exists a self-adjoint extension $A_{\alpha,y}^\Omega$ of $S_y$, with $A_{\infty,y}^\Omega=A_D$, and every self-adjoint extension of $S_y$ is obtained in this way. If $z\in\C\setminus[0,\infty)$ and $\alpha-M_y^\Omega(z)\neq0$, then
\begin{equation}\label{eq:krein-resolvent-3d}
(A_{\alpha,y}^\Omega-z)^{-1}
=
(A_D-z)^{-1}
+
\frac{1}{\alpha-M_y^\Omega(z)}|G_z^y\rangle\langle G_{\bar z}^y|.
\end{equation}
For $\alpha\in\R$,
\begin{equation}\label{eq:operator-domain-3d}
\Dom(A_{\alpha,y}^\Omega)
=
\left\{u=u_z+qG_z^y:u_z\in\Dom(A_D),\ u_z(y)=(\alpha-M_y^\Omega(z))q\right\},
\end{equation}
and
\begin{equation*}
(A_{\alpha,y}^\Omega-z)u=(A_D-z)u_z.
\end{equation*}
\end{theorem}

\begin{proof}
Let $\alpha\in\R$. For nonreal $z$, define
\begin{equation*}
R_{\alpha,y}(z)
:=
(A_D-z)^{-1}
+
\frac{1}{\alpha-M_y^\Omega(z)}|G_z^y\rangle\langle G_{\bar z}^y|.
\end{equation*}
Taking $w=\bar z$ in \eqref{eq:weyl-identity-3d} gives
\begin{equation*}
\Im M_y^\Omega(z)=(\Im z)\|G_z^y\|_{L^2(\Omega)}^2,
\end{equation*}
so $\alpha-M_y^\Omega(z)\neq0$. A direct computation based on \eqref{eq:weyl-identity-3d} yields
\begin{equation*}
R_{\alpha,y}(z)-R_{\alpha,y}(w)
=(z-w)R_{\alpha,y}(z)R_{\alpha,y}(w),
\qquad
R_{\alpha,y}(\bar z)=R_{\alpha,y}(z)^*.
\end{equation*}
Moreover, $R_{\alpha,y}(z)$ is injective. Indeed, if $R_{\alpha,y}(z)f=0$, then
\begin{equation*}
(A_D-z)^{-1}f=-qG_z^y
\end{equation*}
for some $q\in\C$. The left-hand side belongs to $H_0^1(\Omega)$, whereas $G_z^y\notin H^1_{\mathrm{loc}}(\Omega)$ because of its Coulomb singularity. Thus $q=0$ and then $f=0$. The adjoint relation implies that the range is dense. The standard pseudoresolvent criterion therefore gives a self-adjoint operator $A_{\alpha,y}^\Omega$ with resolvent $R_{\alpha,y}$.

If $u=R_{\alpha,y}(z)f$, set
\begin{equation*}
u_z:=(A_D-z)^{-1}f,
\qquad
q:=\frac{(f,G_{\bar z}^y)_{L^2(\Omega)}}{\alpha-M_y^\Omega(z)}.
\end{equation*}
Then $u=u_z+qG_z^y$, and \eqref{eq:green-vector-reproducing} gives
\begin{equation*}
u_z(y)=(f,G_{\bar z}^y)_{L^2(\Omega)}=(\alpha-M_y^\Omega(z))q.
\end{equation*}
Conversely, if $u=u_z+qG_z^y$ satisfies this condition and $f=(A_D-z)u_z$, then $R_{\alpha,y}(z)f=u$. This proves \eqref{eq:operator-domain-3d} and the formula for the action.

The formula for real negative $z$ follows by analytic continuation, or equivalently by the resolvent identity, on every connected component of the set where $\alpha-M_y^\Omega(z)\neq0$.

Finally, let $u\in\Dom(S_y)$. Then the rank-one term in \eqref{eq:krein-resolvent-3d} vanishes on $(S_y-z)u$, so $A_{\alpha,y}^\Omega$ extends $S_y$. Since $S_y$ has deficiency indices $(1,1)$, the parameter $\alpha\in\R\cup\{\infty\}$ exhausts all self-adjoint extensions; see, for instance, \cite{2008OaM,BGP2008}. The case $\alpha=\infty$ gives $A_D$.

\end{proof}

\begin{remark}[Relation with Kre\u{\i}n $\mathcal Q$-functions]
In the formulation of boundary triples, define
\begin{equation*}
\gamma(z)c:=cG_z^y,
\qquad c\in\C.
\end{equation*}
Then $\gamma$ is the Kre\u{\i}n $\Gamma$-field associated with the restriction $S_y$ and the reference extension $A_D$. The identity
\begin{equation*}
M_y^\Omega(z)-M_y^\Omega(w)
=
(z-w)(G_z^y,G_{\bar w}^y)_{L^2(\Omega)}
\end{equation*}
shows that $M_y^\Omega$ is a Kre\u{\i}n $\mathcal Q$-function (see for example \cite[Definition~1.19]{BGP2008}). 
The denominator $\alpha-M_y^\Omega(z)$ used in the Kre\u{\i}n formula is the negative of $\mathcal Q(z)-\alpha$ and up to the irrelevant sign is itself a $\mathcal Q$-function. However, to avoid confusion we will call it throughout the paper the {\em Kre\u{\i}n denominator}.
The same interpretation and denomination applies in dimension two to the logarithmically renormalized Weyl function introduced below.
\end{remark}

\begin{remark}
The condition in \eqref{eq:operator-domain-3d} is independent of the auxiliary point $z$. Equivalently, every $u\in\Dom(A_{\alpha,y}^\Omega)$ has the local expansion
\begin{equation*}
u(x)=\frac{q}{4\pi|x-y|}+\alpha q+o(1),
\qquad
x\to y,
\end{equation*}
for a uniquely determined $q\in\C$.
\end{remark}
\begin{remark}
As it is typical for point interactions, the generic element of the domain is the sum of a singular part and a regular part. The coefficient of the singular part, usually called {\em charge} in this context, is related to the evaluation of the singularity of the regular part (see equation \eqref{eq:operator-domain-3d}).
\end{remark}
\begin{remark}
The simplest example is the half-space. The Weyl function is given in \eqref{eq:halfspace-weyl-3d}. After writing $z=\zeta^2$, $\Im\zeta>0$, its Kre\u{\i}n denominator agrees with $\Gamma_{\alpha,y}^D(\zeta)$ in \cite[Section~2]{NojaRasoStoia25}. A more complete analysis of the model is given in Subsection~\ref{sec:flat-models}.
\end{remark}

\subsection{Dimension two}
Assume now that $d=2$. By \eqref{eq:dirichlet-kernel-domain},
\begin{equation*}
G_z^\Omega(x,y)=\frac{1}{2\pi}K_0\bigl(\kappa(z)|x-y|\bigr)-h_z^\Omega(x,y).
\end{equation*}
By using \cite[(10.31.2)]{DLMF} one has the expansion
\begin{equation*}
K_0(\rho)=-\log\frac{\rho}{2}-\gamma_{\mathrm E}+o(1),
\qquad
\rho\downarrow0,
\end{equation*}
where $\gamma_{\mathrm E}$ is the Euler--Mascheroni constant. Hence one can write
\begin{equation*}
G_z^\Omega(x,y)=-\frac{1}{2\pi}\log|x-y|+M_y^\Omega(z)+o(1),
\qquad
x\to y,
\end{equation*}
with
\begin{equation}\label{eq:weyl-function-2d}
M_y^\Omega(z)
:=
-\frac{1}{2\pi}\left(\log\frac{\kappa(z)}{2}+\gamma_{\mathrm E}\right)-h_z^\Omega(y,y),
\end{equation}
where the principal branch of the logarithm is used. Given these premises, one can build two dimensional point interactions on domains.

\begin{lemma}\label{lem:weyl-identity-2d}
For $z,w\in\C\setminus[0,\infty)$, $z\neq w$,
\begin{equation}\label{eq:weyl-identity-2d}
M_y^\Omega(z)-M_y^\Omega(w)=(z-w)(G_z^y,G_{\bar w}^y)_{L^2(\Omega)}.
\end{equation}
In particular,
\begin{equation*}
M_y^\Omega(\bar z)=\overline{M_y^\Omega(z)},
\qquad
\Im M_y^\Omega(z)=(\Im z)\|G_z^y\|_{L^2(\Omega)}^2.
\end{equation*}
\end{lemma}

\begin{proof}
The proof of Lemma~\ref{lem:weyl-identity-3d} applies verbatim, because the logarithmic singularities of $G_z^y$ and $G_w^y$ cancel in their difference.
\end{proof}

\begin{theorem}\label{thm:krein2}
For every $\alpha\in\R\cup\{\infty\}$ there exists a self-adjoint extension $A_{\alpha,y}^{(2),\Omega}$ of $S_y$, with $A_{\infty,y}^{(2),\Omega}=A_D$, and every self-adjoint extension of $S_y$ is obtained in this way. If $z\in\C\setminus[0,\infty)$ and $\alpha-M_y^\Omega(z)\neq0$, then
\begin{equation}\label{eq:krein-resolvent-2d}
(A_{\alpha,y}^{(2),\Omega}-z)^{-1}
=
(A_D-z)^{-1}
+
\frac{1}{\alpha-M_y^\Omega(z)}|G_z^y\rangle\langle G_{\bar z}^y|.
\end{equation}
For $\alpha\in\R$,
\begin{equation}\label{eq:operator-domain-2d}
\Dom(A_{\alpha,y}^{(2),\Omega})
=
\left\{u=u_z+qG_z^y:u_z\in\Dom(A_D),\ u_z(y)=(\alpha-M_y^\Omega(z))q\right\},
\end{equation}
and
\begin{equation*}
(A_{\alpha,y}^{(2),\Omega}-z)u=(A_D-z)u_z.
\end{equation*}
\end{theorem}

\begin{proof}
The proof is the same as that of Theorem~\ref{thm:krein}, with the logarithmic Weyl function \eqref{eq:weyl-function-2d} and identity \eqref{eq:weyl-identity-2d} in place of their three-dimensional counterparts.
\end{proof}

As in dimension three, \eqref{eq:operator-domain-2d} is independent of $z$. It is equivalent to the local expansion
\begin{equation*}
u(x)=-\frac{q}{2\pi}\log|x-y|+\alpha q+o(1),
\qquad
x\to y.
\end{equation*}

\begin{remark}[Reference length in dimension two]\label{rem:length-scale-2d}
The logarithmic singularity requires a reference length. For $\ell>0$,
\begin{equation*}
\frac{1}{2\pi}K_0\bigl(\kappa(z)|x-y|\bigr)
=
-\frac{1}{2\pi}\log\frac{|x-y|}{\ell}
-
\frac{1}{2\pi}\left(\log\frac{\kappa(z)\ell}{2}+\gamma_{\mathrm E}\right)
+o(1).
\end{equation*}
Accordingly, define
\begin{equation*}
M_y^{\Omega,\ell}(z)
:=
-\frac{1}{2\pi}\left(\log\frac{\kappa(z)\ell}{2}+\gamma_{\mathrm E}\right)-h_z^\Omega(y,y).
\end{equation*}
Then
\begin{equation*}
G_z^\Omega(x,y)
=
-\frac{1}{2\pi}\log\frac{|x-y|}{\ell}+M_y^{\Omega,\ell}(z)+o(1).
\end{equation*}
Throughout the paper we set $\ell=1$ in the chosen units and omit it from the notation. If the reference length is changed from $\ell$ to $\ell'$, then
\begin{equation*}
M_y^{\Omega,\ell'}(z)
=
M_y^{\Omega,\ell}(z)-\frac{1}{2\pi}\log\frac{\ell'}{\ell},
\qquad
\alpha_{\ell'}=\alpha_\ell-\frac{1}{2\pi}\log\frac{\ell'}{\ell}.
\end{equation*}
Thus the difference $\alpha_\ell-M_y^{\Omega,\ell}(z)$, and hence the resolvent and all spectral conclusions, are independent of the choice of $\ell$.
\end{remark}

The explicit half-plane model and its normalization are discussed in Subsection~\ref{sec:flat-models}.

\begin{remark}
Point interactions on domains may alternatively be constructed by quadratic forms, as in the whole-space theory; see \cite{Teta}. This approach is useful for variational questions, but it is not needed here and its development will be given elsewhere.
\end{remark}

\begin{remark}
As for the boundary-less case, the construction above is specific to dimensions $d=2,3$. In higher dimensions the same restriction procedure gives no non-trivial point interaction. More precisely, let $d\geq4$ and let $y\in\Omega$. If one defines the punctured symmetric operator by restricting the Dirichlet Laplacian to smooth functions supported away from $y$, and then takes its closure, the resulting symmetric operator is essentially self-adjoint. The proof is analogous to the standard whole-space case. Actually the boundary affects only the regular part of the Green function and therefore does not change the local singularity, responsible for the restriction to low dimensions (see, for instance, \cite{AGHH, DerezinskiPoint,BGP2005}).
\end{remark}

\section{Spectral theory in the one-centre case}\label{sec:spectral}

In this section we establish the basic spectral properties of the one-centre Dirichlet
point interaction. We first determine the essential and negative spectra, then prove
monotonicity with respect to the ambient domain, and finally we show under which conditions the limiting
absorption principle can be extended from the background Dirichlet Laplacian to the point-interaction
Hamiltonian. When this last property holds, the positive spectrum is absolutely continuous.
\vskip5pt
Whenever a statement applies to both dimensions, we write $A_{\alpha,y}$ for $A_{\alpha,y}^{\Omega}$ in dimension three and for $A_{\alpha,y}^{(2),\Omega}$ in dimension two.

\subsection{Essential spectrum}

\begin{proposition}\label{prop:essential-spectrum}
Let $\Omega\subset\R^d$, $d\in\{2,3\}$, be either an exterior Lipschitz
domain or a special Lipschitz domain, and let $y\in\Omega$. Then
\begin{equation*}
\sigma_{\mathrm{ess}}(A_{\alpha,y})
=
\sigma_{\mathrm{ess}}(A_D)
=
[0,+\infty).
\end{equation*}


\end{proposition}

\begin{proof}
This follows immediately from the Kre\u{\i}n formulas \eqref{eq:krein-resolvent-3d} and \eqref{eq:krein-resolvent-2d}. For the resolvent difference one has
\begin{equation*}
(A_{\alpha,y}-z)^{-1}-(A_D-z)^{-1}
=
\frac{1}{\alpha-M_y^\Omega(z)}
|G_z^y\rangle\langle G_{\bar z}^y|,
\qquad
z\in\rho(A_D)\cap\rho(A_{\alpha,y}).
\end{equation*}
This operator is of rank one, hence compact.  Weyl's theorem then yields the equality of the essential spectra. For the exterior and special Lipschitz domains considered here, the background Dirichlet Laplacian satisfies
$\sigma_{\mathrm{ess}}(A_D)=[0,+\infty)$. This is a well-known fact valid for the more general class of quasi-conical domains (see \cite{EdmundsEvans2018}, Section 6.1). All this proof works identically in the two and three dimensional case.
\end{proof}
\begin{remark}
From the Kre\u{\i}n formula it also follows that at most a single negative eigenvalue exists. However, a more precise statement is proved in the next subsection (see Theorems \ref{thm:critical-coupling-3d} and \ref{thm:critical-coupling-2d}).
\end{remark}

\subsection{Negative spectrum and critical coupling in dimension three}

\begin{theorem}[Dimension three]\label{thm:critical-coupling-3d}
Let $\Omega\subset\R^3$ be either an exterior Lipschitz domain or a special
Lipschitz domain, and let $y\in\Omega$. For $\lambda>0$, set
\begin{equation}\label{eq:negative-weyl-3d}
M_\lambda:=M_y^\Omega(-\lambda)
=-\frac{\sqrt\lambda}{4\pi}-h_{-\lambda}^\Omega(y,y).
\end{equation}
Then the following assertions hold.
\begin{enumerate}
\item[(i)] For every $\lambda>0$, the following conditions are equivalent:
\begin{enumerate}
\item[(a)] $-\lambda\in\sigma_p(A_{\alpha,y}^{\Omega})$;
\item[(b)]
\begin{equation}\label{eq:eigenvalue-equation-3d}
\alpha-M_\lambda=0.
\end{equation}
Equivalently,
\begin{equation*}
\alpha+\frac{\sqrt\lambda}{4\pi}+h_{-\lambda}^\Omega(y,y)=0.
\end{equation*}
\end{enumerate}

\item[(ii)] If these conditions hold, then $-\lambda$ is simple and
\begin{equation*}
\Ker(A_{\alpha,y}^{\Omega}+\lambda)
=\operatorname{span}\{G_{-\lambda}^y\}.
\end{equation*}

\item[(iii)] The function $\lambda\mapsto M_\lambda$ is continuous and
strictly decreasing on $(0,+\infty)$. Moreover,
\begin{equation*}
\lim_{\lambda\to+\infty}M_\lambda=-\infty,
\qquad
M_0:=\lim_{\lambda\downarrow0}M_\lambda\in\R.
\end{equation*}
The finite limit
\begin{equation*}
h_0^\Omega(y,y):=
\lim_{\lambda\downarrow0}h_{-\lambda}^\Omega(y,y)
\end{equation*}
also exists and satisfies
\begin{equation*}
M_0=-h_0^\Omega(y,y).
\end{equation*}

\item[(iv)] Define
\begin{equation*}
\alpha_c^\Omega(y)
:=\sup\bigl\{\alpha\in\R:
\sigma_p(A_{\alpha,y}^{\Omega})\cap(-\infty,0)\neq\varnothing\bigr\}.
\end{equation*}
Then
\begin{equation*}
\alpha_c^\Omega(y)=M_0=-h_0^\Omega(y,y),
\end{equation*}
and
\begin{equation*}
\alpha<\alpha_c^\Omega(y)
\quad\Longleftrightarrow\quad
\sigma_p(A_{\alpha,y}^{\Omega})\cap(-\infty,0)\neq\varnothing.
\end{equation*}
Whenever it exists, the negative eigenvalue is unique and simple.
\end{enumerate}
\end{theorem}

\begin{proof}
Since $A_D\geq0$, one has $-\lambda\in\rho(A_D)$ for every
$\lambda>0$. By the Kre\u{\i}n formula \eqref{eq:krein-resolvent-3d},
$-\lambda$ is an eigenvalue of $A_{\alpha,y}^{\Omega}$ if and only if
\eqref{eq:eigenvalue-equation-3d} holds. The corresponding eigenspace is
spanned by $G_{-\lambda}^y$. This proves \textup{(i)} and \textup{(ii)}.

For $\mu>\lambda>0$, Lemma~\ref{lem:weyl-identity-3d} gives
\begin{equation*}
M_\mu-M_\lambda
=-(\mu-\lambda)
 (G_{-\mu}^y,G_{-\lambda}^y)_{L^2(\Omega)}<0,
\end{equation*}
because the negative-energy Dirichlet Green functions are strictly positive.
Thus $M_\lambda$ is strictly decreasing. Its continuity follows from the
analyticity of $M_y^\Omega$ on $( -\infty,0)$.

By Lemma~\ref{lem:sign-corrector},
\begin{equation*}
h_{-\lambda}^\Omega(y,y)\geq0,
\end{equation*}
and hence, by \eqref{eq:negative-weyl-3d},
\begin{equation*}
M_\lambda\leq-\frac{\sqrt\lambda}{4\pi}.
\end{equation*}
It follows that $M_\lambda\to-\infty$ as $\lambda\to+\infty$. On the
other hand, for any fixed $\lambda_0>0$ and $0<\lambda\leq\lambda_0$,
\begin{equation*}
M_{\lambda_0}\leq M_\lambda\leq0.
\end{equation*}
Therefore the finite limit $M_0$ exists. Formula
\eqref{eq:negative-weyl-3d} then yields
\begin{equation*}
h_{-\lambda}^\Omega(y,y)
=-M_\lambda-\frac{\sqrt\lambda}{4\pi}
\longrightarrow-M_0,
\end{equation*}
which proves \textup{(iii)}.

The range of the continuous strictly decreasing function
$\lambda\mapsto M_\lambda$ is $( -\infty,M_0)$. Hence
\eqref{eq:eigenvalue-equation-3d} has a solution $\lambda>0$ if and only if
$\alpha<M_0$. Assertion \textup{(iv)} follows.
\end{proof}

\subsection{Negative spectrum and critical coupling in dimension two}
The two dimensional case needs a separate analysis because of a more delicate behavior of the terms occurring in the Kre\u{\i}n formula when $\lambda\downarrow 0\ .$ See also subsequent Remark \ref{thresholdweyl}.
\begin{theorem}[Dimension two]\label{thm:critical-coupling-2d}
Let $\Omega\subset\R^2$ be either an exterior Lipschitz domain or a special
Lipschitz domain, and let $y\in\Omega$. For $\lambda>0$, set
\begin{equation}\label{eq:negative-weyl-2d}
M_\lambda^{(2)}:=M_y^\Omega(-\lambda)
=-\frac{1}{2\pi}
 \left(\log\frac{\sqrt\lambda}{2}+\gamma_{\mathrm E}\right)
-h_{-\lambda}^\Omega(y,y).
\end{equation}
Then the following assertions hold.
\begin{enumerate}
\item[(i)] For every $\lambda>0$, the following conditions are equivalent:
\begin{enumerate}
\item[(a)] $-\lambda\in\sigma_p(A_{\alpha,y}^{(2),\Omega})$;
\item[(b)]
\begin{equation}\label{eq:eigenvalue-equation-2d}
\alpha-M_\lambda^{(2)}=0.
\end{equation}
Equivalently,
\begin{equation*}
\alpha+
\frac{1}{2\pi}
 \left(\log\frac{\sqrt\lambda}{2}+\gamma_{\mathrm E}\right)
+h_{-\lambda}^\Omega(y,y)=0.
\end{equation*}
\end{enumerate}

\item[(ii)] If these conditions hold, then $-\lambda$ is simple and
\begin{equation*}
\Ker(A_{\alpha,y}^{(2),\Omega}+\lambda)
=\operatorname{span}\{G_{-\lambda}^y\}.
\end{equation*}

\item[(iii)] The function $\lambda\mapsto M_\lambda^{(2)}$ is continuous
and strictly decreasing on $(0,+\infty)$. Moreover,
\begin{equation*}
\lim_{\lambda\to+\infty}M_\lambda^{(2)}=-\infty,
\qquad
M_0^{(2)}:=\lim_{\lambda\downarrow0}M_\lambda^{(2)}\in\R.
\end{equation*}

\item[(iv)] Define
\begin{equation*}
\alpha_c^{(2),\Omega}(y)
:=\sup\bigl\{\alpha\in\R:
\sigma_p(A_{\alpha,y}^{(2),\Omega})\cap(-\infty,0)
\neq\varnothing\bigr\}.
\end{equation*}
Then
\begin{equation*}
\alpha_c^{(2),\Omega}(y)=M_0^{(2)},
\end{equation*}
and
\begin{equation*}
\alpha<\alpha_c^{(2),\Omega}(y)
\quad\Longleftrightarrow\quad
\sigma_p(A_{\alpha,y}^{(2),\Omega})\cap(-\infty,0)
\neq\varnothing.
\end{equation*}
Whenever it exists, the negative eigenvalue is unique and simple.
\end{enumerate}
\end{theorem}

\begin{proof}
Since $A_D\geq0$, one has $-\lambda\in\rho(A_D)$ for every
$\lambda>0$. The Kre\u{\i}n formula \eqref{eq:krein-resolvent-2d} proves
\textup{(i)} and \textup{(ii)}. For $\mu>\lambda>0$,
Lemma~\ref{lem:weyl-identity-2d} gives
\begin{equation*}
M_\mu^{(2)}-M_\lambda^{(2)}
=-(\mu-\lambda)
 (G_{-\mu}^y,G_{-\lambda}^y)_{L^2(\Omega)}<0.
\end{equation*}
Thus $M_\lambda^{(2)}$ is strictly decreasing, and continuity follows from
the analyticity of the Weyl function on $( -\infty,0)$.

Lemma~\ref{lem:sign-corrector} and \eqref{eq:negative-weyl-2d} imply
\begin{equation*}
M_\lambda^{(2)}
\leq-\frac{1}{2\pi}
 \left(\log\frac{\sqrt\lambda}{2}+\gamma_{\mathrm E}\right),
\end{equation*}
so $M_\lambda^{(2)}\to-\infty$ as $\lambda\to+\infty$.

It remains to control the limit at zero. The complement of $\Omega$ contains
a closed disk $\overline{B_R(a)}$. Set
\begin{equation*}
D:=\R^2\setminus\overline{B_R(a)},
\qquad
\rho:=|y-a|>R,
\end{equation*}
and let
\begin{equation*}
y^*:=a+\frac{R^2}{\rho^2}(y-a).
\end{equation*}
Since $\Omega\subset D$, domain monotonicity and monotonicity with respect
to the spectral parameter give
\begin{equation*}
0<G_{-\lambda}^\Omega(x,y)
\leq G_{-\lambda}^D(x,y)
\leq G_0^D(x,y),
\qquad x\in\Omega,\quad x\neq y.
\end{equation*}
The zero-energy Dirichlet Green function of $D$ is
\begin{equation*}
G_0^D(x,y)
=\frac{1}{2\pi}
 \log\!\left(\frac{\rho\,|x-y^*|}{R\,|x-y|}\right).
\end{equation*}
Its regularized diagonal value is
\begin{equation*}
\lim_{x\to y}
 \left(G_0^D(x,y)+\frac{1}{2\pi}\log|x-y|\right)
=\frac{1}{2\pi}\log\frac{\rho^2-R^2}{R}.
\end{equation*}
Taking the regularized diagonal limit in the preceding comparison yields
\begin{equation*}
M_\lambda^{(2)}
\leq\frac{1}{2\pi}\log\frac{\rho^2-R^2}{R},
\qquad \lambda>0.
\end{equation*}
Since $M_\lambda^{(2)}$ increases as $\lambda\downarrow0$, this uniform
upper bound proves the existence and finiteness of $M_0^{(2)}$. This proves
\textup{(iii)}.

The range of $\lambda\mapsto M_\lambda^{(2)}$ is
$( -\infty,M_0^{(2)})$. Therefore
\eqref{eq:eigenvalue-equation-2d} has a solution $\lambda>0$ if and only if
$\alpha<M_0^{(2)}$, which proves \textup{(iv)}.
\end{proof}

\begin{remark}\label{thresholdweyl}
We stress that, unlike the three-dimensional harmonic correction,
$h_{-\lambda}^\Omega(y,y)$ does not have a finite limit as
$\lambda\downarrow0$. Instead, \eqref{eq:negative-weyl-2d} gives
\begin{equation*}
h_{-\lambda}^\Omega(y,y)
=-\frac{1}{2\pi}
 \left(\log\frac{\sqrt\lambda}{2}+\gamma_{\mathrm E}\right)
-M_0^{(2)}+o(1).
\end{equation*}
Thus the logarithmic divergence of the free term is cancelled by the harmonic
correction, and $M_0^{(2)}$ is the finite renormalized regular part of the Green
function on the diagonal at zero energy.
\end{remark}

\begin{remark}
The whole spaces $\R^3$ and $\R^2$ belong to neither of the two geometric
classes considered here. In dimension three,
\begin{equation*}
M_y^{\R^3}(-\lambda)=-\frac{\sqrt\lambda}{4\pi},
\qquad
\alpha_c^{\R^3}=0.
\end{equation*}
In dimension two,
\begin{equation*}
M_y^{\R^2}(-\lambda)
=-\frac{1}{2\pi}
 \left(\log\frac{\sqrt\lambda}{2}+\gamma_{\mathrm E}\right)
\longrightarrow+\infty
\qquad (\lambda\downarrow0).
\end{equation*}
Hence no finite critical coupling exists in $\R^2$; in the extended sense,
$\alpha_c^{(2),\R^2}=+\infty$, and a one-centre point interaction has one
negative eigenvalue for every $\alpha\in\R$.
\end{remark}

\subsection{Monotonicity with respect to domains}

The construction of Section~\ref{sec:construction} and the preceding 
spectral analysis extend without change to bounded Lipschitz domains. In that
case $0\in\rho(-\Delta_{\Omega,D})$, so the Weyl function has a finite
zero-energy limit. Bounded domains are included below because balls and disks
will be used as comparison domains in Section~\ref{sec:model}.

Within this subsection, to avoid annoying repetitions of statements, $A_{\alpha,y}^{\Omega}$ and
$\alpha_c^\Omega(y)$ denote the point-interaction operator and critical coupling without distinguish the dimension; so, in dimension two they stand for
$A_{\alpha,y}^{(2),\Omega}$ and $\alpha_c^{(2),\Omega}(y)$, respectively.

\begin{theorem}\label{thm:domain_monotonicity_eigenvalue}
Let $\Omega_1\subset\Omega_2\subset\R^d$, $d\in\{2,3\}$, and let
$y\in\Omega_1$. Assume that each $\Omega_j$ is either a bounded Lipschitz
domain, an exterior Lipschitz domain, or a special Lipschitz domain. Whenever it
exists, denote by $\lambda_1(\Omega,\alpha,y)<0$ the unique negative
eigenvalue. Then the following assertions hold.
\begin{enumerate}
\item[(i)] For every $\lambda>0$,
\begin{equation}\label{eq:green-monotonicity-domains}
0<G_{-\lambda}^{\Omega_1}(x,y)
\leq G_{-\lambda}^{\Omega_2}(x,y),
\qquad x\in\Omega_1,\quad x\neq y.
\end{equation}
Consequently,
\begin{equation*}
h_{-\lambda}^{\Omega_1}(y,y)
\geq h_{-\lambda}^{\Omega_2}(y,y),
\end{equation*}
and
\begin{equation}\label{eq:weyl-monotonicity-domains}
M_y^{\Omega_1}(-\lambda)
\leq M_y^{\Omega_2}(-\lambda),
\qquad \lambda>0.
\end{equation}

\item[(ii)] The critical coupling is monotone under inclusion:
\begin{equation*}
\alpha_c^{\Omega_1}(y)\leq\alpha_c^{\Omega_2}(y).
\end{equation*}

\item[(iii)] If $A_{\alpha,y}^{\Omega_1}$ has a negative eigenvalue, then
$A_{\alpha,y}^{\Omega_2}$ also has a negative eigenvalue, and
\begin{equation*}
\lambda_1(\Omega_2,\alpha,y)
\leq\lambda_1(\Omega_1,\alpha,y).
\end{equation*}
\end{enumerate}
\end{theorem}

\begin{proof}
Fix $\lambda>0$, and set
\begin{equation*}
u:=G_{-\lambda}^{\Omega_2}(\cdot,y)
-G_{-\lambda}^{\Omega_1}(\cdot,y)
\qquad\text{in }\Omega_1.
\end{equation*}
The two Green functions have the same singularity at $y$. Hence
\begin{equation*}
u=h_{-\lambda}^{\Omega_1}(\cdot,y)
-h_{-\lambda}^{\Omega_2}(\cdot,y)\big|_{\Omega_1}
\in H^1(\Omega_1),
\end{equation*}
and
\begin{equation*}
(-\Delta+\lambda)u=0
\qquad\text{weakly in }\Omega_1.
\end{equation*}
Moreover,
\begin{equation*}
\operatorname{tr}u
=\operatorname{tr}G_{-\lambda}^{\Omega_2}(\cdot,y)\geq0
\qquad\text{on }\partial\Omega_1.
\end{equation*}
Thus $u^-\in H_0^1(\Omega_1)$. Testing the weak equation with $u^-$
gives
\begin{equation*}
0
=-\int_{\Omega_1}|\nabla u^-|^2\,dx
 -\lambda\int_{\Omega_1}|u^-|^2\,dx.
\end{equation*}
Therefore $u^-=0$, and \eqref{eq:green-monotonicity-domains} follows.
Taking the regularized diagonal limit proves
\eqref{eq:weyl-monotonicity-domains}. This establishes \textup{(i)}.

For each of the three domain classes in the statement,
\begin{equation*}
\alpha_c^\Omega(y)
=\lim_{\lambda\downarrow0}M_y^\Omega(-\lambda).
\end{equation*}
Letting $\lambda\downarrow0$ in
\eqref{eq:weyl-monotonicity-domains} proves \textup{(ii)}.

Assume that $A_{\alpha,y}^{\Omega_1}$ has a negative eigenvalue. Then
\begin{equation*}
\alpha<\alpha_c^{\Omega_1}(y)
\leq\alpha_c^{\Omega_2}(y),
\end{equation*}
so $A_{\alpha,y}^{\Omega_2}$ also has one. Let $\mu_j>0$ be determined by
\begin{equation*}
\alpha=M_y^{\Omega_j}(-\mu_j),
\qquad j=1,2.
\end{equation*}
Since $\mu\mapsto M_y^{\Omega_j}(-\mu)$ is strictly decreasing and
\eqref{eq:weyl-monotonicity-domains} holds for every $\mu>0$, one has
$\mu_2\geq\mu_1$. Consequently,
\begin{equation*}
\lambda_1(\Omega_2,\alpha,y)=-\mu_2
\leq-\mu_1=\lambda_1(\Omega_1,\alpha,y),
\end{equation*}
which proves \textup{(iii)}.
\end{proof}

\begin{remark}\label{rem:domain_monotonicity}
Enlarging a Dirichlet domain increases the critical coupling and lowers the
unique negative eigenvalue whenever it exists. This is the point-interaction
analogue of the usual domain monotonicity for the Dirichlet Laplacian.
\end{remark}

\subsection{Positive spectrum and absolute continuity}

For $s\in\R$, let
\begin{equation*}
L_s^2(\Omega)
:=L^2\bigl(\Omega,\langle x\rangle^{2s}\,dx\bigr),
\qquad
\langle x\rangle:=(1+|x|^2)^{1/2}.
\end{equation*}
The $L^2$-inner product extends continuously to the dual pairs
$L_s^2(\Omega)\times L_{-s}^2(\Omega)$ and
$L_{-s}^2(\Omega)\times L_s^2(\Omega)$. 

The next theorem gives sufficient conditions to pass the limiting absorption principle from the background
Dirichlet Laplacian to the one-centre point interaction.

\begin{theorem}
\label{thm:abstract_transfer_ac}
Let $\Omega\subset\R^d$, $d\in\{2,3\}$, and set
$A_D:=-\Delta^{\Omega}_{D}$. Fix $y\in\Omega$, $\alpha\in\R$, and
$s>1/2$. Set
\begin{equation*}
G_z^y:=G_z^\Omega(\cdot,y),
\qquad
M_y(z):=M_y^\Omega(z).
\end{equation*}
Assume that the limits
\begin{equation*}
R_D^\pm(\lambda)
:=\lim_{\varepsilon\downarrow0}
 (A_D-\lambda\mp i\varepsilon)^{-1}
\end{equation*}
exist in $\mathcal B(L_s^2(\Omega),L_{-s}^2(\Omega))$, locally uniformly
for $\lambda\in(0,+\infty)$. Assume also that
\begin{equation*}
G_{\lambda\pm i0}^y
:=\lim_{\varepsilon\downarrow0}G_{\lambda\pm i\varepsilon}^y
\quad\text{in }L_{-s}^2(\Omega),
\end{equation*}
and
\begin{equation*}
M_y(\lambda\pm i0)
:=\lim_{\varepsilon\downarrow0}M_y(\lambda\pm i\varepsilon)
\end{equation*}
exist locally uniformly for $\lambda\in(0,+\infty)$. Finally, assume that,
for every compact interval $I\Subset(0,+\infty)$,
\begin{equation}\label{eq:scalar-nonvanishing-positive-axis}
\inf_{\lambda\in I}
\left|\alpha-M_y(\lambda+i0)\right|>0.
\end{equation}
Then the limits
\begin{equation*}
R_{\alpha,y}^\pm(\lambda)
:=\lim_{\varepsilon\downarrow0}
 (A_{\alpha,y}-\lambda\mp i\varepsilon)^{-1}
\end{equation*}
exist in $\mathcal B(L_s^2(\Omega),L_{-s}^2(\Omega))$, locally uniformly
for $\lambda\in(0,+\infty)$, and satisfy
\begin{equation}\label{eq:boundary-krein-formula}
R_{\alpha,y}^\pm(\lambda)
=R_D^\pm(\lambda)
+\frac{1}{\alpha-M_y(\lambda\pm i0)}
 (\,\cdot\,,G_{\lambda\mp i0}^y)
 G_{\lambda\pm i0}^y.
\end{equation}
Consequently, the spectrum of $A_{\alpha,y}$ in $(0,+\infty)$ is purely
absolutely continuous. In particular,
\begin{equation*}
\sigma_p(A_{\alpha,y})\cap(0,+\infty)=\varnothing,
\qquad
\sigma_{\mathrm{sc}}(A_{\alpha,y})\cap(0,+\infty)=\varnothing.
\end{equation*}
\end{theorem}

\begin{proof}
Fix a compact interval $I\Subset(0,+\infty)$, and set
\begin{equation*}
c_I:=\inf_{\lambda\in I}
\left|\alpha-M_y(\lambda+i0)\right|>0.
\end{equation*}
By local uniform convergence of the Weyl function, for all sufficiently small
$\varepsilon>0$,
\begin{equation*}
\inf_{\lambda\in I}
\left|\alpha-M_y(\lambda+i\varepsilon)\right|
\geq\frac{c_I}{2}.
\end{equation*}
The analogous estimate for the lower sign follows from
$M_y(\overline z)=\overline{M_y(z)}$ and $\alpha\in\R$.

For $g,h\in L_{-s}^2(\Omega)$, the rank-one operator
$f\mapsto(f,g)h$ belongs to
$\mathcal B(L_s^2(\Omega),L_{-s}^2(\Omega))$, and
\begin{equation*}
\bigl\|(\,\cdot\,,g)h\bigr\|_{
 \mathcal B(L_s^2,L_{-s}^2)}
\leq\|g\|_{L_{-s}^2}\,\|h\|_{L_{-s}^2}.
\end{equation*}
Thus convergence of $g$ and $h$ in $L_{-s}^2(\Omega)$ implies
operator-norm convergence of the corresponding rank-one operators. Passing to
the limit in the Kre\u{\i}n formula therefore gives
\eqref{eq:boundary-krein-formula}, locally uniformly on $I$. Since $I$ is
arbitrary, the limiting absorption principle holds on $(0,+\infty)$.

Let $E_{\alpha,y}$ be the spectral measure of $A_{\alpha,y}$, let
$J\Subset(0,+\infty)$, and let $f\in L_s^2(\Omega)$. Stone's formula,
together with the locally continuous boundary values in
\eqref{eq:boundary-krein-formula}, shows that the spectral measure
\begin{equation*}
B\longmapsto
(E_{\alpha,y}(B)f,f)_{L^2(\Omega)}
\end{equation*}
is absolutely continuous on $J$. Hence
$E_{\alpha,y}(J)L_s^2(\Omega)$ lies in the absolutely continuous subspace.
Since $L_s^2(\Omega)$ is dense in $L^2(\Omega)$, the same is true of
$E_{\alpha,y}(J)L^2(\Omega)$. As $J\Subset(0,+\infty)$ is arbitrary,
the spectrum in $(0,+\infty)$ is purely absolutely continuous.
\end{proof}

\begin{remark}
Theorem~\ref{thm:abstract_transfer_ac} makes no assertion at the threshold
zero; threshold eigenvalues and resonances are treated in Section~\ref{sec:model}.
\end{remark}

In the following, the word {\em outgoing} is used in the standard stationary scattering sense. If $\lambda>0$, the boundary value $(A_D-\lambda-i0)^{-1}$ is obtained as the limit of the resolvent from the physical half-plane and gives the outgoing solution of the Helmholtz equation. In an exterior Euclidean end this means that, outside a compact set, the solution satisfies the Sommerfeld radiation condition
\begin{equation*}
\partial_r u-i\sqrt\lambda\,u=o(r^{-(d-1)/2}),
\qquad r=|x|\to+\infty,
\end{equation*}
in the usual angular sense. Equivalently one has
\begin{equation*}
u(r\omega)
=
r^{-(d-1)/2}e^{i\sqrt\lambda r}
\bigl(a^+(\omega)+o(1)\bigr),
\qquad r\to+\infty,
\end{equation*}
in $L^2$ of the angular variable $\omega$. This is the standard obstacle scattering normalization of the outgoing
amplitude (see, for instance~\cite[Sections~1.3--1.4]{Ramm2017}). The coefficient $a^+$ is the outgoing far-field amplitude. In an asymptotically conical end (a prototype special Lipschitz domain that will be studied next), the same notation denotes the corresponding coefficient in the generalized outgoing expansion associated with the limiting absorption principle and the generalized Fourier transform.
With this normalization, the outgoing amplitude satisfies the standard flux identity: the imaginary radial flux through large spheres equals $\sqrt\lambda$ times the squared norm of the outgoing amplitude. 

\begin{lemma}\label{lem:positive-imaginary-weyl}
Assume that, for every $\lambda>0$, the outgoing Green function
\begin{equation*}
u_\lambda^+:=G_{\lambda+i0}^\Omega(\cdot,y)
\end{equation*}
exists and admits an outgoing amplitude $a_{\lambda,y}^+$. Assume that the
flux identity
\begin{equation}\label{eq:radiation-flux-green-vector}
\lim_{R\to+\infty}
\operatorname{Im}
\int_{\Omega\cap\partial B_R}
\overline{u_\lambda^+}\,\partial_r u_\lambda^+\,d\sigma
=\sqrt\lambda\,\|a_{\lambda,y}^+\|^2
\end{equation}
holds for every $\lambda>0$, and that the Rellich uniqueness
theorem holds. Assume also that $\lambda\mapsto M_y^\Omega(\lambda+i0)$ is
locally continuous on $(0,+\infty)$. Then
\begin{equation}\label{eq:imaginary-weyl-radiation-amplitude}
\operatorname{Im}M_y^\Omega(\lambda+i0)
=\sqrt\lambda\,\|a_{\lambda,y}^+\|^2>0.
\end{equation}
Consequently, \eqref{eq:scalar-nonvanishing-positive-axis} holds for every
$\alpha\in\R$.
\end{lemma}

\begin{proof}
Fix $\lambda>0$. For $R>|y|$ and
$0<\varepsilon<\operatorname{dist}(y,\partial\Omega)$, set
\begin{equation*}
\Omega_{R,\varepsilon}
:=(\Omega\cap B_R)\setminus\overline{B_\varepsilon(y)}.
\end{equation*}
Since \(u_\lambda^+\) solves \((-\Delta-\lambda)u_\lambda^+=0\) in
\(\Omega_{R,\varepsilon}\) and \(\lambda\) is real, Green's second identity applied
to \(u_\lambda^+\) and \(\overline{u_\lambda^+}\) gives
\begin{equation*}
0=
\int_{\partial\Omega_{R,\varepsilon}}
\left(
\overline{u_\lambda^+}\,\partial_\nu u_\lambda^+
-
u_\lambda^+\,\partial_\nu\overline{u_\lambda^+}
\right)d\sigma
=
2i\,\operatorname{Im}
\int_{\partial\Omega_{R,\varepsilon}}
\overline{u_\lambda^+}\,\partial_\nu u_\lambda^+\,d\sigma .
\end{equation*}
The contribution from $\partial\Omega$ vanishes because
$u_\lambda^+$ satisfies the Dirichlet boundary condition. Near $y$,
\begin{equation*}
u_\lambda^+(x)
=
\begin{cases}
\displaystyle
\frac{1}{4\pi|x-y|}+M_y^\Omega(\lambda+i0)+o(1),&d=3,\\[1.2ex]
\displaystyle
-\frac{1}{2\pi}\log|x-y|+M_y^\Omega(\lambda+i0)+o(1),&d=2.
\end{cases}
\end{equation*}
On $\partial B_\varepsilon(y)$, the outer normal to
$\Omega_{R,\varepsilon}$ points towards $y$, and
\begin{equation*}
\partial_\nu u_\lambda^+(x)
=
\begin{cases}
\displaystyle
\frac{1}{4\pi\varepsilon^2}+O(1),&d=3,\\[1.2ex]
\displaystyle
\frac{1}{2\pi\varepsilon}+O(1),&d=2.
\end{cases}
\end{equation*}
Consequently,
\begin{equation*}
\lim_{\varepsilon\downarrow0}
\operatorname{Im}
\int_{\partial B_\varepsilon(y)}
\overline{u_\lambda^+}\,\partial_\nu u_\lambda^+\,d\sigma
=-\operatorname{Im}M_y^\Omega(\lambda+i0).
\end{equation*}
Letting $\varepsilon\downarrow0$ and then $R\to+\infty$, we obtain
\begin{equation*}
\operatorname{Im}M_y^\Omega(\lambda+i0)
=
\lim_{R\to+\infty}
\operatorname{Im}
\int_{\Omega\cap\partial B_R}
\overline{u_\lambda^+}\,\partial_r u_\lambda^+\,d\sigma.
\end{equation*}
Now, using ~\eqref{eq:radiation-flux-green-vector} we obtain 
\eqref{eq:imaginary-weyl-radiation-amplitude}.
If $a_{\lambda,y}^+=0$, Rellich uniqueness implies that
$u_\lambda^+$ vanishes in the scattering end. Unique continuation then gives
$u_\lambda^+=0$ in $\Omega\setminus\{y\}$, contradicting its Coulomb or
logarithmic singularity at $y$. Thus $a_{\lambda,y}^+\neq0$.

Finally, $M_y^\Omega(\lambda+i0)$ is locally continuous in $\lambda$.
Therefore, for every compact $I\Subset(0,+\infty)$,
\begin{equation*}
\inf_{\lambda\in I}
\left|\alpha-M_y^\Omega(\lambda+i0)\right|
\geq
\min_{\lambda\in I}
\operatorname{Im}M_y^\Omega(\lambda+i0)>0.
\end{equation*}
\end{proof}
In the previous Lemma, the hypotheses on the domain are left unspecified, so that the thesis is conditioned to the existence of the outgoing Green function, validity of flux identity ~\eqref{eq:radiation-flux-green-vector} and Rellich uniqueness theorem in the domain $\Omega$. These flat hypotheses have to be verified for the specific domains of interest, that will be done in the following subsections.
\begin{remark}
The quantity $\pi^{-1}\operatorname{Im}M_y^\Omega(\lambda+i0)$
may be interpreted as a measure of the strength with which
a point source at $y$ couples to the continuum at energy
$\lambda$.
\end{remark}
\subsection{Exterior domains}

For exterior domains, the hypotheses of
Lemma~\ref{thm:abstract_transfer_ac} follow from classical obstacle
scattering.

\begin{corollary}\label{cor:exterior_transfer_ac}
Let $\Omega\subset\R^d$, $d\in\{2,3\}$, be an exterior domain with
$C^{1,1}$ boundary. Then, for every $y\in\Omega$ and every
$\alpha\in\R$, the corresponding one-centre Dirichlet point interaction
$A_{\alpha,y}$ satisfies
\begin{equation*}
\sigma_{\mathrm{ac}}(A_{\alpha,y})\cap(0,+\infty)=(0,+\infty),
\end{equation*}
and
\begin{equation*}
\sigma_p(A_{\alpha,y})\cap(0,+\infty)
=
\sigma_{\mathrm{sc}}(A_{\alpha,y})\cap(0,+\infty)
=\varnothing.
\end{equation*}
\end{corollary}

\begin{proof}
For the background Dirichlet Laplacian, classical obstacle scattering gives
the limiting absorption principle, the outgoing far-field expansion, and Rellich
uniqueness on $(0,+\infty)$; see, for instance,
\cite{Jones1953,LaxPhillips1971,Leis1986,Yafaev2000}.

The boundary values of the Green vector follow from the background limiting
absorption principle by a cut-off argument. Choose
$\chi\in C_c^\infty(\Omega)$ with $\chi=1$ near $y$, and let
$\Phi_z$ be the free resolvent kernel with pole at $y$. Then
\begin{equation*}
G_z^\Omega(\cdot,y)
=\chi\Phi_z(\cdot-y)-R_D(z)f_z,
\qquad
f_z:=(-\Delta-z)\bigl(\chi\Phi_z(\cdot-y)\bigr)-\delta_y.
\end{equation*}
The function $f_z$ is smooth and compactly supported away from $y$.
The limiting absorption principle therefore yields the boundary values of the
Green vectors in $L_{-s}^2(\Omega)$, while local elliptic regularity yields the
boundary values of their regular parts at $y$, and hence of the Weyl function.
For $R$ larger than the support of $\chi$, the compactly supported term
$\chi\Phi_{\lambda+i0}(\cdot-y)$ gives no contribution to the flux through
$\Omega\cap\partial B_R$. Thus the outgoing amplitude and the outgoing flux
of $G_{\lambda+i0}^\Omega(\cdot,y)$ are those of the outgoing resolvent applied
to the compactly supported source $-f_{\lambda+i0}$. Classical obstacle scattering
provides the corresponding far-field expansion, flux identity, and Rellich
uniqueness. Hence Lemma~\ref{lem:positive-imaginary-weyl} gives the
non-vanishing condition. Theorem~\ref{thm:abstract_transfer_ac} applies.

Finally, the background Dirichlet Laplacian has spectrum $[0,+\infty)$, and
Proposition~\ref{prop:essential-spectrum} preserves its essential spectrum. Hence
the absolutely continuous positive spectrum fills $(0,+\infty)$.
\end{proof}

\subsection{Special Lipschitz domains}

For a general special Lipschitz domain, a complete scattering theory is
not available. We therefore restrict to the subclasses covered
by the infinite-obstacle scattering results, following especially \cite{Constantin1981, Yafaev2000}, and give references to further special cases treated in the literature in the remarks.
In particular, the following Theorem adapts the geometric hypotheses of
\cite[Chapter~17, Theorem~17.1 and Proposition~17.2]{Yafaev2000} in the
present setting. The essential hypothesis is the so called illumination condition given in equation \eqref{eq:illumination-condition} (see also subsequent Remark \ref{rem:meaning_illumination_cone}).
\begin{theorem}
\label{thm:special_lipschitz_transfer_ac}
Let $d\in\{2,3\}$, and let $\Omega\subset\R^d$ be a $C^2$ special
Lipschitz domain. Assume that there exists $r_0>0$ such that
\begin{equation}\label{eq:illumination-condition}
x\cdot\nu(x)\leq0,
\qquad x\in\partial\Omega,\quad |x|\geq r_0,
\end{equation}
where $\nu(x)$ is the outer unit normal. For $r>r_0$, set
\begin{equation*}
G_r:=\{\omega\in\mathbb S^{d-1}:r\omega\in\Omega\},
\qquad
G:=\bigcup_{r>r_0}G_r,
\end{equation*}
and define the limit cone
$
K:=\{r\omega:r>0,\ \omega\in G\}.
$
Assume that $K$ has nonempty interior and that, for some
$\gamma\in(0,1)$,
\begin{equation}\label{eq:quantitative-asymptotic-cone}
\operatorname{dist}(x,\partial K)=O(|x|^\gamma),
\qquad x\in\partial\Omega,\quad |x|\to+\infty.
\end{equation}
Then, for every $y\in\Omega$ and every $\alpha\in\R$, the corresponding
one-centre Dirichlet point interaction $A_{\alpha,y}$ satisfies
\begin{equation*}
\sigma_{\mathrm{ac}}(A_{\alpha,y})\cap(0,+\infty)=(0,+\infty),
\end{equation*}
and
\begin{equation*}
\sigma_p(A_{\alpha,y})\cap(0,+\infty)
=
\sigma_{\mathrm{sc}}(A_{\alpha,y})\cap(0,+\infty)
=\varnothing.
\end{equation*}
\end{theorem}

\begin{proof}
The infinite-obstacle scattering theory in
\cite[Chapter~17, Theorem~17.1 and Proposition~17.2]{Yafaev2000}
gives, for the background Dirichlet Laplacian, the limiting absorption
principle, the outgoing asymptotic expansion of outgoing solutions, and the
corresponding Rellich uniqueness theorem.

We first verify that these results apply to the Green vector with pole at
$y$. The argument is similar to the one given in Corollary \ref{cor:exterior_transfer_ac} for exterior domains.
Let $\chi\in C_c^\infty(\Omega)$ be equal to one in a neighbourhood of
$y$, and let $\Phi_z$ be the free resolvent kernel of $-\Delta-z$ in
$\R^d$, with pole at the origin. For $z\in\C\setminus[0,+\infty)$ one has
\begin{equation*}
G_z^\Omega(\cdot,y)
=\chi\Phi_z(\cdot-y)-R_D(z)f_z,
\end{equation*}
where $R_D(z)=(A_D-z)^{-1}$ and
\begin{equation*}
f_z:=(-\Delta-z)\bigl(\chi\Phi_z(\cdot-y)\bigr)-\delta_y.
\end{equation*}
The distribution $f_z$ is in fact a smooth compactly supported function,
supported away from $y$. Hence the limiting absorption principle gives, for
every $s>1/2$, the boundary values
\begin{equation*}
G_{\lambda\pm i0}^\Omega(\cdot,y)
\in L_{-s}^2(\Omega),
\qquad \lambda>0,
\end{equation*}
locally uniformly in $\lambda$. Since the first term
$\chi\Phi_{\lambda+i0}(\cdot-y)$ is compactly supported, it vanishes near
infinity and gives no contribution to the far-field coefficient or to the flux
through large spheres. Therefore the outgoing asymptotic, the outgoing
amplitude, and the outgoing flux of $G_{\lambda+i0}^\Omega(\cdot,y)$ are those
of $-R_D(\lambda+i0)f_{\lambda+i0}$.

The results quoted from \cite{Yafaev2000} apply precisely to such compactly supported
sources. They give the outgoing expansion, the far-field amplitude, the flux
identity
\begin{equation*}
\lim_{R\to+\infty}
\operatorname{Im}
\int_{\Omega\cap\partial B_R}
\overline{G_{\lambda+i0}^\Omega(x,y)}\,\partial_rG_{\lambda+i0}^\Omega(x,y)
\,d\sigma(x)
=
\sqrt\lambda\,\|a_{\lambda,y}^+\|^2,
\end{equation*}
and the corresponding Rellich uniqueness theorem. This verifies the scattering
hypotheses of Lemma~\ref{lem:positive-imaginary-weyl} for the Green vector.

It remains only to identify the boundary value of the Weyl function.
The preceding convergence of the Green vectors, together with local elliptic
regularity near the pole and the diagonal expansion defining $M_y^\Omega(z)$,
yields locally continuous boundary values $M_y^\Omega(\lambda\pm i0)$ on
$(0,+\infty)$. Therefore all hypotheses of
Lemma~\ref{lem:positive-imaginary-weyl} are satisfied, and
\begin{equation*}
\operatorname{Im}M_y^\Omega(\lambda+i0)>0,
\qquad \lambda>0.
\end{equation*}
In particular, for real $\alpha$, the Kre\u{\i}n denominator
$\alpha-M_y^\Omega(\lambda+i0)$ does not vanish on the positive axis, and the
non-vanishing condition in Theorem~\ref{thm:abstract_transfer_ac} holds on every
compact subinterval of $(0,+\infty)$.

Theorem~\ref{thm:abstract_transfer_ac} now gives absence of positive eigenvalues
and singular continuous spectrum for $A_{\alpha,y}$. Finally, the generalized
Fourier representation of the background operator gives positive absolutely
continuous spectrum equal to $(0,+\infty)$, and
Proposition~\ref{prop:essential-spectrum} preserves the essential spectrum. Hence
\begin{equation*}
\sigma_{\mathrm{ac}}(A_{\alpha,y})\cap(0,+\infty)=(0,+\infty).
\end{equation*}
\end{proof}

\begin{remark}
\label{rem:meaning_illumination_cone}
Condition \eqref{eq:illumination-condition} means that the continuation of a ray
from the origin through a sufficiently distant boundary point remains in
$\overline\Omega$. The cone $K$ collects the directions along which the
domain extends to infinity. The assumption that $K$ has nonempty interior
excludes degenerate limit cones, while
\eqref{eq:quantitative-asymptotic-cone} requires $\partial\Omega$ to approach
$\partial K$ with sublinear error.
\end{remark}

\begin{remark}
\label{rem:constantin-three-dimensional}
In dimension $d=3$, the quantitative assumption
\eqref{eq:quantitative-asymptotic-cone} is not needed. Constantin proves the
limiting absorption principle, constructs the generalized Fourier transform, and
establishes completeness under the illumination condition
\eqref{eq:illumination-condition} alone; see \cite{Constantin1981} and
\cite[Chapter~17, Section~17.3]{Yafaev2000}. Constantin states the result for
$d\geq3$; in \cite{Yafaev2000} it is noted that this dimensional restriction is probably not
essential, but we are not aware of a complete proof.
\end{remark}

\begin{remark}
\label{rem:geometry-classes-incomparable}
The special Lipschitz condition and the conditions discussed in Constantin or Yafaev papers are
not fully comparable. The former requires the boundary to be a global graph in one fixed
direction, whereas the latter admit domains that are not special Lipschitz, such as
exteriors of paraboloidal obstacles. Conversely, a special Lipschitz domain need
not satisfy the illumination condition. The thin domains studied by Minskii
\cite{Minskii81} form another class and are not used here.
\end{remark}

We explicitly mention two useful subclasses  of special Lipschitz domains where Theorem ~\ref{thm:special_lipschitz_transfer_ac} applies.
\begin{corollary}[Domains conical at infinity]
\label{cor:conical_at_infinity_ac}
Let $d\in\{2,3\}$, and let $\Omega\subset\R^d$ be a $C^2$ special
Lipschitz domain such that, for some $R_0>0$,
\begin{equation*}
\Omega\cap\{|x|>R_0\}
=\Gamma\cap\{|x|>R_0\},
\end{equation*}
where
\begin{equation*}
\Gamma:=\{r\omega:r>0,\ \omega\in\Sigma\}
\end{equation*}
is an open cone with nonempty Lipschitz spherical section
$\Sigma\subset\mathbb S^{d-1}$. Then the conclusion of
Theorem~\ref{thm:special_lipschitz_transfer_ac} holds.
\end{corollary}

\begin{proof}
On $\partial\Gamma$, the radial direction is tangent to the boundary, so
$x\cdot\nu(x)=0$. Moreover, the limit cone is $K=\Gamma$, and
$\operatorname{dist}(x,\partial K)=0$ on the boundary outside a compact set.
Thus the hypotheses of Theorem~\ref{thm:special_lipschitz_transfer_ac} are
satisfied. We also mention that an early treatment of domains conical at infinity is in \cite{Jones1953}.
\end{proof}

\begin{corollary}[Asymptotically half-space and half-plane domains]
\label{cor:asymptotically_flat_ac}
Let $d\in\{2,3\}$, and let
\begin{equation*}
\Omega_\varphi
:=\{(x',x_d)\in\R^{d-1}\times\R:x_d>\varphi(x')\},
\end{equation*}
where $\varphi\in C^2(\R^{d-1})\cap\operatorname{Lip}(\R^{d-1})$.
Assume that, for some $\gamma\in(0,1)$,
\begin{equation}\label{eq:asymptotically-flat-graph}
|\varphi(x')|=O(\langle x'\rangle^\gamma),
\qquad |x'|\to+\infty,
\end{equation}
and that
\begin{equation}\label{eq:graph-illumination}
x'\cdot\nabla\varphi(x')-\varphi(x')\leq0
\qquad\text{for all sufficiently large }|x'|.
\end{equation}
Then the conclusion of Theorem~\ref{thm:special_lipschitz_transfer_ac} holds.
\end{corollary}

\begin{proof}
At a boundary point $x=(x',\varphi(x'))$, the outer unit normal is
\begin{equation*}
\nu(x)
=\frac{(\nabla\varphi(x'),-1)}
 {\sqrt{1+|\nabla\varphi(x')|^2}},
\end{equation*}
and therefore
\begin{equation*}
x\cdot\nu(x)
=\frac{x'\cdot\nabla\varphi(x')-\varphi(x')}
 {\sqrt{1+|\nabla\varphi(x')|^2}}.
\end{equation*}
Thus \eqref{eq:graph-illumination} is precisely the illumination condition.
Moreover, \eqref{eq:asymptotically-flat-graph} implies
$\varphi(x')=o(|x'|)$, so the limit cone is $\R_+^d$. On
$\partial\Omega_\varphi$,
\begin{equation*}
\operatorname{dist}(x,\partial\R_+^d)
=|\varphi(x')|=O(|x|^\gamma),
\qquad |x|\to+\infty.
\end{equation*}
Theorem~\ref{thm:special_lipschitz_transfer_ac} applies.
\end{proof}

\begin{remark}\label{rem:flat-positive-spectrum}
The exact half-space and half-plane correspond to $\varphi\equiv0$.
Corollary~\ref{cor:asymptotically_flat_ac} also includes compactly supported perturbations of half-space or half-plane, but compact support is not required and the result is much more general.
\end{remark}

\begin{remark}
Non localized deformations such as boundaries described by a periodic graph need not satisfy the
illumination condition and require different arguments; compare \cite{Ducomet95} and \cite{Ramm2017}.
\end{remark}
\begin{remark}[Scattering]\label{scattering}
It is natural at this point to formulate the scattering problem for a point interaction on a domain. Indeed, on a fixed
domain, by the Kre\u{\i}n formula the resolvent difference
\begin{equation*}
(A_{\alpha,y}-z)^{-1}-(A_D-z)^{-1}
\end{equation*}
has rank one, so trace-class. Hence the Kato-Rosenblum theorem gives existence and
completeness of the wave operators
\begin{equation*}
W_\pm(A_{\alpha,y},A_D)
:=
s\text{-}\lim_{t\to\pm\infty}
e^{itA_{\alpha,y}}e^{-itA_D}P_{\mathrm{ac}}(A_D).
\end{equation*}

As regards the absolutely continuous subspace, in the one-centre case there is a simple negative eigenvalue to be projected away precisely when
$\alpha<\alpha_c(y)$, while no negative eigenvalue is present for
$\alpha\geq\alpha_c(y)$. At the critical coupling we will see in the next section that a zero-energy state may be a
resonance or, in the geometries where it is square-integrable, a threshold
eigenvalue; in the latter case it is again excluded from the absolutely continuous subspace.

We finally notice that, if a complete geometric scattering theory is available for the background
Dirichlet Laplacian $A_D$ relative to a model operator $A_{\rm mod}$, with
identification $J$, then the chain rule for wave operators gives
\begin{equation*}
W_\pm(A_{\alpha,y},A_{\rm mod};J)
=
W_\pm(A_{\alpha,y},A_D)\,
W_\pm(A_D,A_{\rm mod};J)
\end{equation*}
on the corresponding absolutely continuous subspaces.  For exterior domains this reduces to classical obstacle scattering; for
asymptotically conical domains it fits the generalized Fourier transform and
wave operator framework of \cite{Constantin1981}. Thus the point
interaction adds at most one bound state and one rank-one scattering channel
to the background geometric scattering.
We do not further pursue this topic, because time-dependent problems are not the focus of this work.
\end{remark}

\section{Explicit geometries and threshold behavior}\label{sec:model}
This last section is dedicated to a detailed analysis of the behavior of a point interaction when the coupling is critical. It turns out that an elaborate and interesting picture arises, depending on the geometry of the domain.
We begin with four explicit model geometries: half-space, half-plane and plane sectors, exterior of a ball and exterior of a disk. We then derive the universal
near-boundary asymptotics of the critical coupling and finally we study in detail the nature of the state at the threshold of the continuum when the coupling is critical, $\alpha=\alpha_c^{\Omega}\ .$
\begin{definition}[Critical zero-energy states]
\label{def:critical-zero-energy-state}
Let $d\in\{2,3\}$, let $y\in\Omega$, and assume that the finite critical
coupling and the zero-energy Dirichlet Green function are defined. We write
\begin{equation*}
\alpha_c(y)
:=
\begin{cases}
\alpha_c^{(2),\Omega}(y),&d=2,\\
\alpha_c^\Omega(y),&d=3,
\end{cases}
\qquad
\psi_0:=G_0^\Omega(\cdot,y).
\end{equation*}
We call $\psi_0$ the \emph{critical zero-energy state}. It is harmonic in
$\Omega\setminus\{y\}$, has zero Dirichlet trace on $\partial\Omega$, and
satisfies the local point-interaction condition at the centre $y$ with
coupling $\alpha_c(y)$. When $\psi_0\notin L^2(\Omega)$, this last
condition is meant only locally at $y$, not as membership in the operator
domain. Namely, this means that, as $x\to y$,
\begin{equation*}
\psi_0(x)
=
\begin{cases}
\displaystyle -\frac{1}{2\pi}\log|x-y|+\alpha_c(y)+o(1),&d=2,\\[1.2ex]
\displaystyle \frac{1}{4\pi|x-y|}+\alpha_c(y)+o(1),&d=3.
\end{cases}
\end{equation*}

Set $s_d:=(4-d)/2$. If $\psi_0\in L^2(\Omega)$, then zero is a
zero-energy eigenvalue; when zero also belongs to the essential spectrum of
the corresponding critical Hamiltonian (which is the case we will face), it is a \emph{threshold eigenvalue}. If instead
\begin{equation*}
\psi_0\in
\bigcap_{s>s_d}L_{-s}^2(\Omega)\setminus L^2(\Omega),
\end{equation*}
we call $\psi_0$ a \emph{zero-energy threshold resonance}.

As we will see, in dimension three, the exterior-domain states considered below have at threshold the monopole behavior
\begin{equation*}
\psi_0(x)=\frac{c}{|x|}+O(|x|^{-2}),
\qquad c\neq0, \qquad |x|\to\infty,
\end{equation*}
In dimension two, whenever the end admits the multipole expansion
\begin{equation*}
\psi_0(x)
=
b+\frac{a\cdot x}{|x|^2}+O(|x|^{-2}),
\ \qquad\qquad |x|\to\infty,
\end{equation*}
we use the following terminology: the resonance is of $s$-wave type if
$b\neq0$, and of $p$-wave type if $b=0$ and $a\neq0$. For ends with
a different asymptotic geometry, for example wedges, we retain the general term threshold
resonance and specify the actual far-field decay.
\end{definition}

\begin{remark}
The weighted space
three-dimensional formulation is classical in low energy analysis of the resolvent
\cite{JK79}. The $s$- and $p$-wave terminology in dimension two follows
\cite{CorneanMichelangeliYajima2019}.
\end{remark}

\subsection{Flat model geometries}\label{sec:flat-models}

The Green kernels in the half-space and half-plane follow from the method of
images. In the planar formulas, $K_\nu$ denotes the principal modified Bessel
function of the second kind, also called the Macdonald function. It is the
standard solution of the modified Bessel equation in \cite[(10.25.1),
(10.25.3)]{DLMF}; the small-argument expansion used below is
\cite[(10.31.2)]{DLMF}.

\subsubsection{The three-dimensional half-space}

Let
\begin{equation*}
\Omega=\R^3_+
:=\{x=(x',x_3)\in\R^2\times\R:x_3>0\},
\qquad
 y=(y',y_3)\in\R^3_+,
\end{equation*}
and set $y^*:=(y',-y_3)$. The corresponding point interaction was studied
in detail in \cite{NojaRasoStoia25}.
For $z\in\C\setminus[0,\infty)$, the method of images gives
\begin{equation}\label{eq:halfspace-weyl-3d}
M_y^{\R^3_+}(z)
=
-\frac{\kappa(z)}{4\pi}
-\frac{\mathrm e^{-2y_3\kappa(z)}}{8\pi y_3}.
\end{equation}

\begin{proposition}
\label{prop:halfspace-dirichlet}
For the one-centre Dirichlet point interaction in $\R^3_+$,
\begin{equation*}
\alpha_c^{\R^3_+}(y)=-\frac{1}{8\pi y_3}.
\end{equation*}
If $\alpha<\alpha_c^{\R^3_+}(y)$, there is a unique negative
eigenvalue
\begin{equation*}
-\lambda_\alpha=-a_\alpha^2,
\end{equation*}
where $a_\alpha>0$ is the unique solution of
\begin{equation*}
\alpha+\frac{a}{4\pi}+\frac{e^{-2ay_3}}{8\pi y_3}=0.
\end{equation*}
Moreover,
\begin{equation*}
-16\pi^2\alpha^2< -\lambda_\alpha
\leq -\left(4\pi\alpha+\frac{1}{2y_3}\right)^2.
\end{equation*}
The associated eigenspace is spanned by
\begin{equation*}
G_{-\lambda_\alpha}^{\R^3_+}(x,y)
=
\frac{e^{-\sqrt{\lambda_\alpha}|x-y|}}{4\pi|x-y|}
-
\frac{e^{-\sqrt{\lambda_\alpha}|x-y^*|}}{4\pi|x-y^*|}.
\end{equation*}
At $\alpha=\alpha_c^{\R^3_+}(y)$, zero is a threshold eigenvalue,
with eigenfunction
\begin{equation*}
G_0^{\R^3_+}(x,y)
=
\frac{1}{4\pi|x-y|}-\frac{1}{4\pi|x-y^*|}.
\end{equation*}
\end{proposition}

\begin{proof}
For $a>0$, the image formula gives
\begin{equation*}
G_{-a^2}^{\R^3_+}(x,y)
=
\frac{e^{-a|x-y|}}{4\pi|x-y|}
-
\frac{e^{-a|x-y^*|}}{4\pi|x-y^*|}.
\end{equation*}
Taking $z=-a^2$ in \eqref{eq:halfspace-weyl-3d} gives the regularized diagonal value. The formula for the critical coupling and the characteristic equation follow immediately. Uniqueness of $a_\alpha$ follows from
Theorem~\ref{thm:critical-coupling-3d}. At the root,
\begin{equation*}
\frac{a_\alpha}{4\pi}
=-\alpha-\frac{e^{-2a_\alpha y_3}}{8\pi y_3}.
\end{equation*}
Using $0<e^{-2a_\alpha y_3}<1$ gives
\begin{equation*}
-4\pi\alpha-\frac{1}{2y_3}
\leq a_\alpha<-4\pi\alpha,
\end{equation*}
which is equivalent to the stated eigenvalue bounds.

At zero energy, expansion at infinity gives
\begin{equation*}
G_0^{\R^3_+}(x,y)
=
\frac{y_3x_3}{2\pi|x|^3}+O(|x|^{-3}),
\qquad |x|\to\infty.
\end{equation*}
Thus the monopole term cancels and
$G_0^{\R^3_+}(\cdot,y)\in L^2(\R^3_+)$. To verify the
operator-domain condition, fix $z<0$ and set
$u_z:=G_0^{\R^3_+}(\cdot,y)-G_z^{\R^3_+}(\cdot,y)$. Then
\begin{equation*}
 u_z=(A_D-z)^{-1}\bigl(-zG_0^{\R^3_+}(\cdot,y)\bigr)
 \in\operatorname{dom}(A_D),
\end{equation*}
and the diagonal expansions give
\begin{equation*}
 u_z(y)=\alpha_c^{\R^3_+}(y)-M_y^{\R^3_+}(z).
\end{equation*}
The domain characterization~\eqref{eq:operator-domain-3d}
therefore implies
\begin{equation*}
G_0^{\R^3_+}(\cdot,y)
\in\operatorname{dom}
\bigl(A_{\alpha_c^{\R^3_+}(y),y}^{\R^3_+}\bigr),
\qquad
A_{\alpha_c^{\R^3_+}(y),y}^{\R^3_+}
G_0^{\R^3_+}(\cdot,y)=0.
\end{equation*}
Since the essential spectrum starts at zero, this is a threshold eigenvalue.
\end{proof}

\subsubsection{The two-dimensional half-plane}

Let
\begin{equation*}
\mathbb H:=\R^2_+
=\{x=(x_1,x_2)\in\R^2:x_2>0\},
\qquad
 y=(y_1,y_2)\in\mathbb H,
\end{equation*}
and set $\bar y:=(y_1,-y_2)$.
For $z\in\C\setminus[0,\infty)$, the image formula yields
\begin{equation}\label{eq:halfplane-weyl-2d}
M_y^{\mathbb H}(z)
=
-\frac{1}{2\pi}
\left(\log\frac{\kappa(z)}{2}+\gamma_{\mathrm E}\right)
-\frac{1}{2\pi}K_0\bigl(2\kappa(z)y_2\bigr).
\end{equation}

\begin{proposition}
\label{prop:halfplane-dirichlet}
For the one-centre Dirichlet point interaction in $\mathbb H$,
\begin{equation*}
\alpha_c^{(2),\mathbb H}(y)=\frac{1}{2\pi}\log(2y_2).
\end{equation*}
If $\alpha<\alpha_c^{(2),\mathbb H}(y)$, there is a unique negative
eigenvalue
\begin{equation*}
-\lambda_\alpha=-a_\alpha^2,
\end{equation*}
where $a_\alpha>0$ is the unique solution of
\begin{equation*}
\alpha+\frac{\gamma_{\mathrm E}+\log(a/2)}{2\pi}
+\frac{1}{2\pi}K_0(2ay_2)=0.
\end{equation*}
Its eigenspace is spanned by $G_{-a_\alpha^2}^{\mathbb H}(\cdot,y)$.
At $\alpha=\alpha_c^{(2),\mathbb H}(y)$, the critical state is a
$p$-wave threshold resonance; in particular, zero is not an eigenvalue.
\end{proposition}

\begin{proof}
For $a>0$, the image method gives
\begin{equation*}
G_{-a^2}^{\mathbb H}(x,y)
=
\frac{1}{2\pi}K_0(a|x-y|)
-
\frac{1}{2\pi}K_0(a|x-\bar y|).
\end{equation*}
Taking $z=-a^2$ in \eqref{eq:halfplane-weyl-2d} gives the regularized diagonal value.
The small-argument expansion
\begin{equation*}
K_0(t)=-\log\frac{t}{2}-\gamma_{\mathrm E}
+O\bigl(t^2|\log t|\bigr),
\qquad t\downarrow0,
\end{equation*}
follows from \cite[(10.31.2)]{DLMF}. It gives the stated critical
coupling. The remaining assertions about the negative eigenvalue follow from
Theorem~\ref{thm:critical-coupling-2d}.

At zero energy,
\begin{equation*}
G_0^{\mathbb H}(x,y)
=
\frac{1}{2\pi}\log\frac{|x-\bar y|}{|x-y|}.
\end{equation*}
As $x\to y$ one has
\begin{equation*}
|x-\bar y|=2y_2+o(1),
\end{equation*}
and therefore
\begin{equation*}
G_0^{\mathbb H}(x,y)
=
-\frac{1}{2\pi}\log|x-y|
+
\frac{1}{2\pi}\log(2y_2)
+
o(1).
\end{equation*}
The coefficient of the logarithmic singularity is $q=1$, and the regular constant is
\begin{equation*}
\frac{1}{2\pi}\log(2y_2)
=
\alpha_c^{(2),\mathbb H}(y).
\end{equation*}
Hence $G_0^{\mathbb H}(\cdot,y)$ satisfies the local point-interaction condition at the centre with the critical coupling.
Expanding the logarithms at infinity gives
\begin{equation*}
G_0^{\mathbb H}(x,y)
=
\frac{y_2}{\pi}\frac{x_2}{|x|^2}+O(|x|^{-2}),
\qquad |x|\to\infty.
\end{equation*}
The leading dipole coefficient is nonzero. Hence
\begin{equation*}
G_0^{\mathbb H}(\cdot,y)
\in
\bigcap_{s>0}L^2_{-s}(\mathbb H)\setminus L^2(\mathbb H),
\qquad
G_0^{\mathbb H}(\cdot,y)\in\bigcap_{p>2}L^p(\mathbb H),
\end{equation*}
which is precisely the $p$-wave alternative in
Definition~\ref{def:critical-zero-energy-state}.
\end{proof}

\subsection{Planar sector models}\label{subsec:planar-sector-models}

We next consider the planar wedge, or sector. This is a classical separable
geometry in mathematical physics. It is less elementary than the half-plane,
because the vertex produces a conical singularity and the opening angle enters
all threshold exponents. Nevertheless, the model is still explicit and provides
a useful bridge between flat and conical geometries. For general treatment of conical geometries, see for example \cite{Cessenat96}, Section 5.3.

For $0<\beta<2\pi$, set
\begin{equation*}
\mathcal S_\beta
:=
\{re^{i\theta}:r>0,\ 0<\theta<\beta\}\subset\C,
\qquad
q:=\frac{\pi}{\beta}.
\end{equation*}
After a rigid motion, $\mathcal S_\beta$ is a special Lipschitz domain. The
unperturbed Dirichlet Laplacian on $\mathcal S_\beta$ is always understood as
the Friedrichs realization of $-\Delta$ initially defined on
$C_c^\infty(\mathcal S_\beta)$. Equivalently, it is the self-adjoint operator
associated with the closed quadratic form
\begin{equation*}
\mathfrak a_\beta[u]
=
\int_{\mathcal S_\beta}|\nabla u|^2\,\dd x,
\qquad
\operatorname{dom}\mathfrak a_\beta=H_0^1(\mathcal S_\beta).
\end{equation*}
Thus the Dirichlet condition is imposed on the two sides of the sector, and no
additional boundary condition is imposed at the vertex. This is
important when $\beta>\pi$, where the sector is nonconvex (see \cite{Posilicano2013}).

The unperturbed operator has no discrete spectrum. Indeed, the angular functions
\begin{equation*}
e_n(\theta):=\sqrt{\frac{2}{\beta}}\sin(nq\theta),
\qquad n\in\N,
\end{equation*}
diagonalize the Dirichlet angular operator on $(0,\beta)$, and the corresponding
radial operators are
\begin{equation*}
-
\frac{\dd^2}{\dd r^2}
-
\frac1r\frac{\dd}{\dd r}
+
\frac{n^2q^2}{r^2}
\end{equation*}
in $L^2((0,+\infty),r\,\dd r)$. The Hankel transform of order $nq$,
\begin{equation*}
({\mathcal H}_{nq}f)(\xi)
:=
\int_0^\infty J_{nq}(r\xi)f(r)r\,\dd r,
\end{equation*}
where $J_\nu$ is the Bessel function of the first kind, diagonalizes this
radial operator as multiplication by $\xi^2$. Hence
\begin{equation*}
\sigma(-\Delta_{\mathcal S_\beta,D})
=
\sigma_{\mathrm{ac}}(-\Delta_{\mathcal S_\beta,D})
=[0,+\infty),
\qquad
\sigma_{\mathrm{pp}}(-\Delta_{\mathcal S_\beta,D})
=
\sigma_{\mathrm{sc}}(-\Delta_{\mathcal S_\beta,D})
=\varnothing.
\end{equation*}
For conical domains the purely absolutely continuous spectrum has been established in \cite{Jones1953}; the modern scattering
framework for domains with asymptotic cones is developed in
\cite{Constantin1981,Yafaev2000}. See also \cite{Cessenat96, Ramm2017}.\\
Now suppose that the wedge carries a point interaction at $y$.
Let $y=\rho e^{i\phi}\in\mathcal S_\beta$, with $\rho>0$ and
$0<\phi<\beta$. Then,
\begin{equation*}
F(z):=z^q,
\qquad 0<\arg z<\beta,
\end{equation*}
maps $\mathcal S_\beta$ conformally onto the upper half-plane. By conformal
invariance of the two-dimensional Dirichlet Green function one has
\begin{equation}
\label{eq:sector-zero-green-exact-model}
G_0^{\mathcal S_\beta}(z,y)
=
\frac{1}{2\pi}
\log
\left|
\frac{z^q-\overline{y^q}}
     {z^q-y^q}
\right|.
\end{equation}
Near $z=y$,
\begin{equation*}
z^q-y^q=qy^{q-1}(z-y)+o(|z-y|),
\end{equation*}
whereas
\begin{equation*}
|y^q-\overline{y^q}|
=2\rho^q\sin(q\phi).
\end{equation*}
Therefore
\begin{equation*}
G_0^{\mathcal S_\beta}(z,y)
=
-
\frac{1}{2\pi}\log|z-y|
+
\frac{1}{2\pi}
\log\left(
\frac{2\rho\sin(q\phi)}{q}
\right)
+o(1),
\qquad z\to y.
\end{equation*}
Consequently,
\begin{equation}
\label{eq:sector-critical-coupling-exact-model}
\alpha_c^{(2),\mathcal S_\beta}(y)
=
\frac{1}{2\pi}
\log\left(
\frac{2\rho\sin(q\phi)}{q}
\right)
=
\frac{1}{2\pi}
\log\left(
\frac{2\beta}{\pi}\rho\sin\frac{\pi\phi}{\beta}
\right).
\end{equation}
For $\beta=\pi$, this gives the half-plane value
$\alpha_c^{(2),\mathbb H}(y)=(2\pi)^{-1}\log(2y_2)$.

The negative-energy Green function is obtained by separation of variables. For
$a>0$, put
\begin{equation*}
r_<:=\min\{r,\rho\},
\qquad
r_>:=\max\{r,\rho\}.
\end{equation*}
Then
\begin{equation}
\label{eq:sector-negative-green-expansion}
G_{-a^2}^{\mathcal S_\beta}(r,\theta;\rho,\phi)
=
\frac{2}{\beta}
\sum_{n=1}^{\infty}
\sin(nq\theta)\sin(nq\phi)
I_{nq}(ar_<)K_{nq}(ar_>),
\end{equation}
where $I_\nu$ and $K_\nu$ are the modified Bessel functions of first and
second kind. 

For fixed angular mode $n$, the two independent radial solutions are
$I_{nq}(ar)$ and $K_{nq}(ar)$. The first one is regular at the vertex, whereas
the second one is singular there. More precisely, using the standard small-argument
asymptotics of the modified Bessel functions, see
\cite[(10.30.1)--(10.30.2)]{DLMF},
\begin{equation*}
I_{nq}(ar)\sim C r^{nq},
\qquad
K_{nq}(ar)\sim C r^{-nq},
\qquad r\downarrow0,
\end{equation*}
with $nq>0$. 
At infinity, $K_{nq}$ is selected by its exponential decay; see
\cite[(10.40.1)--(10.40.2)]{DLMF}. This explains the structure
$I_{nq}(ar_<)K_{nq}(ar_>)$ in the Green function expansion. For fixed $a>0$,
the series converges locally uniformly away from the pole, so the local regular
part at $y$ is obtained term by term.

We summarize with the following Proposition.

\begin{proposition}
\label{prop:sector-negative-spectrum}
Let $\mathcal S_\beta$ and $y=\rho e^{i\phi}$ be as above. Then the critical
coupling is given by \eqref{eq:sector-critical-coupling-exact-model}. If
$\alpha<\alpha_c^{(2),\mathcal S_\beta}(y)$, the operator
$A_{\alpha,y}^{(2),\mathcal S_\beta}$ has exactly one negative eigenvalue
$-a_\alpha^2$, where $a_\alpha>0$ is the unique solution of
\begin{equation*}
\alpha=M_y^{\mathcal S_\beta}(-a_\alpha^2).
\end{equation*}
The corresponding eigenspace is spanned by
$G_{-a_\alpha^2}^{\mathcal S_\beta}(\cdot,y)$. If
$\alpha\geq\alpha_c^{(2),\mathcal S_\beta}(y)$, then
$A_{\alpha,y}^{(2),\mathcal S_\beta}$ has no negative eigenvalue.
\end{proposition}

\begin{proof}
The formula for the critical coupling follows from the local expansion of
$G_0^{\mathcal S_\beta}(\cdot,y)$ obtained above. The remaining assertions are
direct consequences of the general two-dimensional one-centre criterion,
Theorem~\ref{thm:critical-coupling-2d}. Indeed, the negative spectrum is governed by
\begin{equation*}
\alpha=M_y^{\mathcal S_\beta}(-a^2),
\qquad a>0,
\end{equation*}
where $a^2\mapsto M_y^{\mathcal S_\beta}(-a^2)$ is continuous, strictly
decreasing, tends to $\alpha_c^{(2),\mathcal S_\beta}(y)$ as $a\downarrow0$,
and tends to $-\infty$ as $a\to+\infty$. Thus there is exactly one negative
eigenvalue precisely for $\alpha<\alpha_c^{(2),\mathcal S_\beta}(y)$, and its
eigenspace is generated by $G_{-a_\alpha^2}^{\mathcal S_\beta}(\cdot,y)$.
\end{proof}

\subsection{Exterior ball models}

We next consider the exterior of a ball in dimensions three and two.  We write $P_n$ for the Legendre
polynomial of degree $n$, normalized by $P_n(1)=1$; see
\cite[Table~18.6.1]{DLMF}. Each addition formula and asymptotic relation used
below is cited at the point where it enters the argument.

The zero-energy Kelvin geometry common to the two models is illustrated in
Figure~\ref{fig:kelvin-exterior-ball}.

\begin{figure}[H]
\centering
\resizebox{0.52\linewidth}{!}{%
\begin{tikzpicture}[
  x=1.15cm,y=1.15cm,
  >=Latex,
  every node/.style={font=\small},
  point/.style={circle,fill=black,inner sep=1.55pt},
  dim/.style={<->,thin}
]
  \def\Rval{1.45}
  \def\rhoval{3.55}
  \pgfmathsetmacro{\ystarval}{\Rval*\Rval/\rhoval}

  \fill[gray!18] (0,0) circle (\Rval);
  \draw[thick] (0,0) circle (\Rval);
  \node[fill=gray!18,inner sep=1pt] at (0,0.72) {$\overline{B_R}$};
  \node[anchor=west,fill=white,inner sep=1.5pt] at (1.72,1.42)
    {$\Omega=\mathbb{R}^d\setminus\overline{B_R}$};

  \draw[thin] (-1.82,0) -- (4.08,0);
  \node[point] (O) at (0,0) {};
  \node[anchor=north east] at (-0.08,-0.08) {$O$};
  \node[point] (Ys) at (\ystarval,0) {};
  \node[anchor=south] at (\ystarval,0.14) {$y^*$};
  \node[point] (Y) at (\rhoval,0) {};
  \node[anchor=south west] at (\rhoval+0.02,0.08) {$y$};

  \draw[->] (0,0) -- node[pos=0.58,left,fill=gray!18,inner sep=1pt] {$R$} (135:\Rval);
  \draw[-{Latex[length=2.3mm]},densely dashed]
    (Y.north west) .. controls (3.08,1.03) and (1.48,1.03) .. (Ys.north east);
  \node[fill=white,inner sep=1.5pt] at (2.15,1.05) {Kelvin inversion};

  \draw[dim] (0,-0.54) -- (\ystarval,-0.54);
  \node[fill=white,inner sep=1pt] at (0.5*\ystarval,-0.54) {$R^2/\rho$};
  \draw[dim] (0,-0.98) -- (\rhoval,-0.98);
  \node[fill=white,inner sep=1pt] at (0.5*\rhoval,-0.98) {$\rho$};
  \node[fill=white,inner sep=1pt] at (1.78,-1.39)
    {$|Oy|\,|Oy^*|=R^2$};
\end{tikzpicture}
}
\caption{Planar cross-section of the Kelvin image construction for the exterior of a ball. The pole $y$ and its image $y^*$ lie on the same radial line and satisfy $|Oy|\,|Oy^*|=R^2$.}
\label{fig:kelvin-exterior-ball}
\end{figure}
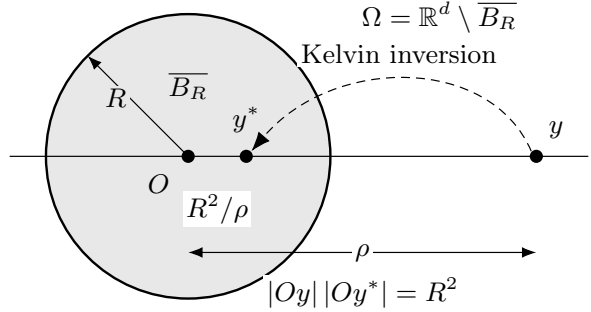

\subsubsection{The exterior of a sphere}

Let
\begin{equation*}
\Omega_R:=\{x\in\R^3:|x|>R\},
\qquad
y=(0,0,\rho),
\qquad
\rho>R.
\end{equation*}
The choice of the positive $x_3$-axis is immaterial by rotational invariance.
Write $x=(r,\theta,\varphi)$ in spherical coordinates, where $\theta$ is
the angle between $x$ and $y$, and set
\begin{equation*}
r_<:=\min\{r,\rho\},
\qquad
r_>:=\max\{r,\rho\}.
\end{equation*}
We use the following normalization of the modified spherical Bessel
functions:
\begin{equation*}
i_n(t):=\sqrt{\frac{\pi}{2t}}\,I_{n+\frac12}(t),
\qquad
k_n(t):=\sqrt{\frac{2}{\pi t}}\,K_{n+\frac12}(t).
\end{equation*}
The first definition agrees with \cite[(10.47.7)]{DLMF}, whereas the second
is $2/\pi$ times the normalization in \cite[(10.47.9)]{DLMF}. This choice
removes an otherwise recurring factor $2/\pi$ from the Yukawa addition
formula.

\begin{proposition}
\label{prop:exterior-sphere}
For the one-centre Dirichlet point interaction in $\Omega_R$,
\begin{equation*}
\alpha_c^{\Omega_R}(y)=-\frac{R}{4\pi(\rho^2-R^2)}.
\end{equation*}
If $\alpha<\alpha_c^{\Omega_R}(y)$, there is a unique negative eigenvalue
\begin{equation*}
-\lambda_\alpha=-a_\alpha^2,
\end{equation*}
where $a_\alpha>0$ is the unique solution of
\begin{equation*}
\alpha+\frac{a}{4\pi}
+
\frac{a}{4\pi}
\sum_{n=0}^{\infty}(2n+1)
\frac{i_n(aR)}{k_n(aR)}k_n(a\rho)^2
=0.
\end{equation*}
Its eigenspace is spanned by $G_{-a_\alpha^2}^{\Omega_R}(\cdot,y)$.
At $\alpha=\alpha_c^{\Omega_R}(y)$, the critical state is a
three-dimensional monopole threshold resonance; in particular, zero is not an
eigenvalue.
\end{proposition}

\begin{proof}
Because the pole lies on the symmetry axis, the harmonic correction is expanded
in zonal spherical harmonics,
\begin{equation*}
h(r,\theta)=\sum_{n=0}^{\infty}c_n(r)P_n(\cos\theta).
\end{equation*}
Substitution into $( -\Delta+a^2)h=0$ gives, for each $n$,
\begin{equation*}
r^2c_n''(r)+2rc_n'(r)-\bigl(a^2r^2+n(n+1)\bigr)c_n(r)=0.
\end{equation*}
Its independent solutions are $i_n(ar)$ and $k_n(ar)$. The large-argument
asymptotics \cite[(10.40.1)--(10.40.2)]{DLMF}, together with the definitions
above, show that only $k_n(ar)$ decays as $r\to\infty$.

With our normalization, the Yukawa addition theorem obtained from
\cite[(10.60.3)]{DLMF} reads
\begin{equation*}
\frac{e^{-a|x-y|}}{4\pi|x-y|}
=
\frac{a}{4\pi}
\sum_{n=0}^{\infty}(2n+1)P_n(\cos\theta)
 i_n(ar_<)k_n(ar_>),
\qquad a>0.
\end{equation*}
For fixed $a>0$ and $R<\rho$, the series converges locally uniformly for
$r>R$, away from the pole. Hence the boundary matching at $r=R$ and
the evaluation of the regular part at $x=y$ are justified term by term.
For $R<\rho$, matching the trace of each angular mode at $r=R$
therefore fixes the coefficient of the decaying solution and gives
\begin{equation*}
h_{-a^2}^{\Omega_R}(x,y)
=
\frac{a}{4\pi}
\sum_{n=0}^{\infty}(2n+1)P_n(\cos\theta)
\frac{i_n(aR)}{k_n(aR)}k_n(ar)k_n(a\rho).
\end{equation*}
Therefore
\begin{equation*}
G_{-a^2}^{\Omega_R}(x,y)
=
\frac{e^{-a|x-y|}}{4\pi|x-y|}
-
\frac{a}{4\pi}
\sum_{n=0}^{\infty}(2n+1)P_n(\cos\theta)
\frac{i_n(aR)}{k_n(aR)}k_n(ar)k_n(a\rho).
\end{equation*}
Since the correction is regular at the pole and $P_n(1)=1$
\cite[Table~18.6.1]{DLMF},
\begin{equation*}
h_{-a^2}^{\Omega_R}(y,y)
=
\frac{a}{4\pi}
\sum_{n=0}^{\infty}(2n+1)
\frac{i_n(aR)}{k_n(aR)}k_n(a\rho)^2.
\end{equation*}
The characteristic equation in the statement now follows from
$\alpha=M_y^{\Omega_R}(-a^2)$, and uniqueness follows from
Theorem~\ref{thm:critical-coupling-3d}.

At zero energy, set
\begin{equation*}
y^*:=\frac{R^2}{\rho^2}y.
\end{equation*}
For $|x|=R$, one has $|x-y|=(\rho/R)|x-y^*|$. Kelvin inversion therefore
gives
\begin{equation*}
G_0^{\Omega_R}(x,y)
=
\frac{1}{4\pi|x-y|}
-
\frac{R}{\rho}\frac{1}{4\pi|x-y^*|}.
\end{equation*}
Since $|y-y^*|=(\rho^2-R^2)/\rho$, its expansion at the pole is
\begin{equation*}
G_0^{\Omega_R}(x,y)
=
\frac{1}{4\pi|x-y|}
-
\frac{R}{4\pi(\rho^2-R^2)}
+O(|x-y|),
\end{equation*}
which gives the stated critical coupling. At infinity,
\begin{equation*}
G_0^{\Omega_R}(x,y)
=
\frac{1-R/\rho}{4\pi|x|}+O(|x|^{-2}).
\end{equation*}
The coefficient is nonzero; hence
\begin{equation*}
G_0^{\Omega_R}(\cdot,y)
\in
\bigcap_{s>1/2}L^2_{-s}(\Omega_R)\setminus L^2(\Omega_R).
\end{equation*}
Thus the critical state is a monopole threshold resonance.
\end{proof}

\subsubsection{The exterior of a disk}

Let
\begin{equation*}
\Omega_R:=\{x\in\R^2:|x|>R\},
\qquad
y=(\rho,\phi),
\qquad
\rho>R.
\end{equation*}
Write $x=(r,\theta)$ in polar coordinates and set
\begin{equation*}
r_<:=\min\{r,\rho\},
\qquad
r_>:=\max\{r,\rho\}.
\end{equation*}

\begin{proposition}
\label{prop:exterior-disk}
For the one-centre Dirichlet point interaction in $\Omega_R$,
\begin{equation*}
\alpha_c^{(2),\Omega_R}(y)
=
\frac{1}{2\pi}\log\frac{\rho^2-R^2}{R}.
\end{equation*}
If $\alpha<\alpha_c^{(2),\Omega_R}(y)$, there is a unique negative
eigenvalue
\begin{equation*}
-\lambda_\alpha=-a_\alpha^2,
\end{equation*}
where $a_\alpha>0$ is the unique solution of
\begin{equation*}
\alpha+\frac{\gamma_{\mathrm E}+\log(a/2)}{2\pi}
+
\frac{1}{2\pi}
\sum_{n\in\mathbb Z}
\frac{I_{|n|}(aR)}{K_{|n|}(aR)}K_{|n|}(a\rho)^2
=0.
\end{equation*}
Its eigenspace is spanned by $G_{-a_\alpha^2}^{\Omega_R}(\cdot,y)$.
At $\alpha=\alpha_c^{(2),\Omega_R}(y)$, the critical state is an
$s$-wave threshold resonance; in particular, zero is not an eigenvalue.
\end{proposition}

\begin{proof}
Expand the harmonic correction in angular Fourier modes,
\begin{equation*}
h(r,\theta)=\sum_{n\in\mathbb Z}c_n(r)e^{in(\theta-\phi)}.
\end{equation*}
The equation $( -\Delta+a^2)h=0$ reduces to
\begin{equation*}
r^2c_n''(r)+rc_n'(r)-\bigl(a^2r^2+n^2\bigr)c_n(r)=0.
\end{equation*}
Its independent radial solutions are $I_{|n|}(ar)$ and
$K_{|n|}(ar)$. By \cite[(10.40.1)--(10.40.2)]{DLMF}, the exterior decay
condition selects $K_{|n|}(ar)$.

Neumann's addition theorem \cite[(10.44.3)]{DLMF}, specialized to order zero,
gives
\begin{equation*}
K_0(a|x-y|)
=
\sum_{n\in\mathbb Z}e^{in(\theta-\phi)}
I_{|n|}(ar_<)K_{|n|}(ar_>),
\qquad a>0.
\end{equation*}
For fixed $a>0$ and $R<\rho$, the series converges locally uniformly for
$r>R$, away from the pole. Thus the Fourier modes may be matched on
$\partial\Omega_R$, and the regular part may be evaluated at $x=y$ term
by term.
Matching each Fourier coefficient at $r=R<\rho$ gives
\begin{equation*}
h_{-a^2}^{\Omega_R}(x,y)
=
\frac{1}{2\pi}
\sum_{n\in\mathbb Z}e^{in(\theta-\phi)}
\frac{I_{|n|}(aR)}{K_{|n|}(aR)}
K_{|n|}(ar)K_{|n|}(a\rho).
\end{equation*}
Thus
\begin{equation*}
G_{-a^2}^{\Omega_R}(x,y)
=
\frac{1}{2\pi}K_0(a|x-y|)
-
\frac{1}{2\pi}
\sum_{n\in\mathbb Z}e^{in(\theta-\phi)}
\frac{I_{|n|}(aR)}{K_{|n|}(aR)}
K_{|n|}(ar)K_{|n|}(a\rho),
\end{equation*}
and hence
\begin{equation*}
h_{-a^2}^{\Omega_R}(y,y)
=
\frac{1}{2\pi}
\sum_{n\in\mathbb Z}
\frac{I_{|n|}(aR)}{K_{|n|}(aR)}K_{|n|}(a\rho)^2.
\end{equation*}
This gives the characteristic equation in the statement.

We show explicitly the cancellation that produces the finite threshold value.
Set
\begin{equation*}
p:=\frac{R}{\rho}\in(0,1),
\qquad
L_a:=-\log(a/2)-\gamma_{\mathrm E}.
\end{equation*}
For fixed $t>0$, the power series \cite[(10.25.2)]{DLMF} and
\cite[(10.31.2)]{DLMF} give
\begin{equation*}
I_0(at)=1+O(a^2),
\qquad
K_0(at)=L_a-\log t+O\bigl(a^2|\log a|\bigr).
\end{equation*}
Consequently, the zero angular mode satisfies
\begin{equation*}
\frac{I_0(aR)}{K_0(aR)}K_0(a\rho)^2
=
L_a+\log R-2\log\rho+o(1).
\end{equation*}
For every fixed $n\geq1$, the limiting forms
\cite[(10.30.1)--(10.30.2)]{DLMF} yield
\begin{equation*}
I_n(t)\sim\frac{t^n}{2^n n!},
\qquad
K_n(t)\sim2^{n-1}(n-1)!\,t^{-n},
\qquad t\downarrow0,
\end{equation*}
so that
\begin{equation*}
\frac{I_n(aR)}{K_n(aR)}K_n(a\rho)^2
\longrightarrow
\frac{p^{2n}}{2n}.
\end{equation*}
The large-order asymptotic forms \cite[(10.41.1)--(10.41.2)]{DLMF}
yield, uniformly for $a$ in a sufficiently small bounded interval, a summable
majorant of the form $Cp^{2n}/n$. Dominated convergence
therefore gives
\begin{equation*}
\sum_{n\in\mathbb Z\setminus\{0\}}
\frac{I_{|n|}(aR)}{K_{|n|}(aR)}K_{|n|}(a\rho)^2
\longrightarrow
\sum_{n=1}^{\infty}\frac{p^{2n}}{n}
=-\log(1-p^2).
\end{equation*}
It follows that
\begin{equation*}
h_{-a^2}^{\Omega_R}(y,y)
=
\frac{1}{2\pi}
\bigl(L_a+\log R-2\log\rho-\log(1-p^2)\bigr)+o(1).
\end{equation*}
Since
\begin{equation*}
M_y^{\Omega_R}(-a^2)
=\frac{L_a}{2\pi}-h_{-a^2}^{\Omega_R}(y,y),
\end{equation*}
the divergent terms cancel and
\begin{equation*}
\lim_{a\downarrow0}M_y^{\Omega_R}(-a^2)
=
\frac{1}{2\pi}\log\frac{\rho^2-R^2}{R}.
\end{equation*}

The zero-energy Green function gives the threshold state directly. With
\begin{equation*}
y^*:=\frac{R^2}{\rho^2}y,
\end{equation*}
one has $|x-y|=(\rho/R)|x-y^*|$ for $|x|=R$, and hence
\begin{equation*}
G_0^{\Omega_R}(x,y)
=
\frac{1}{2\pi}\log\frac{\rho|x-y^*|}{R|x-y|}.
\end{equation*}
Its expansion at the pole is
\begin{equation*}
G_0^{\Omega_R}(x,y)
=
-\frac{1}{2\pi}\log|x-y|
+
\frac{1}{2\pi}\log\frac{\rho^2-R^2}{R}
+O(|x-y|),
\end{equation*}
and its expansion at infinity is
\begin{equation*}
G_0^{\Omega_R}(x,y)
=
b_y+\frac{d_y\cdot x}{|x|^2}+O(|x|^{-2}),
\qquad |x|\to\infty,
\end{equation*}
where
\begin{equation*}
b_y:=\frac{1}{2\pi}\log\frac{\rho}{R}>0,
\qquad
d_y:=\frac{y-y^*}{2\pi}.
\end{equation*}
Therefore
\begin{equation*}
G_0^{\Omega_R}(\cdot,y)
\in
\bigcap_{s>1}L^2_{-s}(\Omega_R)\setminus L^2(\Omega_R),
\end{equation*}
and the nonzero constant term identifies an $s$-wave threshold resonance.
\end{proof}

\begin{table}[t!]
\centering
\small
\begin{tabular}{c|c}
model & critical coupling \\
\hline\\
$\R^3_+$ &
$-\dfrac{1}{8\pi y_3}$\\[1.2ex]
$\R^2_+$ &
$\dfrac{1}{2\pi}\log(2y_2)$\\[1.2ex]
$\mathcal S_\beta$, $y=\rho e^{i\phi}$ &
$\dfrac{1}{2\pi}\log\!\left(\dfrac{2\rho\sin(q\phi)}{q}\right)$\\[1.2ex]
$\R^3\setminus\overline{B_R}$, $\rho=|y|$ &
$-\dfrac{R}{4\pi(\rho^2-R^2)}$\\[1.2ex]
$\R^2\setminus\overline{B_R}$, $\rho=|y|$ &
$\dfrac{1}{2\pi}\log\!\left(\dfrac{\rho^2-R^2}{R}\right)$
\end{tabular}
\vskip5pt
\caption{Critical coupling in the exact models. In the sector row, $q=\pi/\beta$.}
\label{tab:exact-model-negative-spectrum}
\end{table}
The Table 1 summarizes the critical couplings for the various exact models treated.
\subsection{Near-boundary behavior of the critical coupling}

The explicit Green functions of balls and disks, combined with the domain
monotonicity of Theorem~\ref{thm:domain_monotonicity_eigenvalue}, yield sharp
near-boundary estimates for the critical coupling.

\begin{lemma}
\label{lem:critical_ball_disk}
Let $\rho>0$.
\begin{enumerate}
\item[(i)] If $B_\rho\subset\R^3$, $y\in B_\rho\setminus\{0\}$,
$r:=|y|$, and $d_\rho:=\rho-r$, then
\begin{equation*}
G_0^{B_\rho}(x,y)
=
\frac{1}{4\pi|x-y|}
-
\frac{\rho}{|y|}\frac{1}{4\pi|x-y^*|},
\qquad
y^*:=\frac{\rho^2}{|y|^2}y,
\end{equation*}
and
\begin{equation*}
\alpha_c^{B_\rho}(y)
=
-\frac{\rho}{4\pi(\rho^2-r^2)}
=
-\frac{\rho}{4\pi d_\rho(2\rho-d_\rho)}.
\end{equation*}

\item[(ii)] If $D_\rho\subset\R^2$, $y\in D_\rho\setminus\{0\}$,
$r:=|y|$, and $d_\rho:=\rho-r$, then
\begin{equation*}
G_0^{D_\rho}(x,y)
=
-\frac{1}{2\pi}\log|x-y|
+
\frac{1}{2\pi}\log\!\left(\frac{|y|}{\rho}|x-y^*|\right),
\qquad
y^*:=\frac{\rho^2}{|y|^2}y,
\end{equation*}
and
\begin{equation*}
\alpha_c^{(2),D_\rho}(y)
=
\frac{1}{2\pi}\log\!\left(\frac{\rho^2-r^2}{\rho}\right)
=
\frac{1}{2\pi}\log\!\left(\frac{d_\rho(2\rho-d_\rho)}{\rho}\right).
\end{equation*}
\end{enumerate}
The formulas at $y=0$ are obtained by continuity.
\end{lemma}

\begin{proof}
These are the classical Kelvin-image formulas; see, for example,
\cite{Kellogg1929}. Taking the regular part at \(x=y\) in the above
decompositions gives the stated critical couplings.
\end{proof}

The comparison configuration used in the next theorem is shown in
Figure~\ref{fig:tangent-ball-comparison}.

\begin{figure}[!t]
\centering
\resizebox{0.52\linewidth}{!}{%
\begin{tikzpicture}[
  x=1.10cm,y=1.10cm,
  >=Latex,
  every node/.style={font=\small},
  point/.style={circle,fill=black,inner sep=1.55pt}
]
  \def\rad{1.35}
  \def\dy{0.62}

  \fill[gray!7]
    (-2.70,0.5832) -- (-2.70,2.95) -- (2.70,2.95) -- (2.70,0.5832)
    plot[smooth,domain=2.70:-2.70,samples=100] (\x,{0.08*\x*\x}) -- cycle;
  \draw[very thick,smooth,domain=-2.70:2.70,samples=120]
    plot (\x,{0.08*\x*\x});
  \node[anchor=east,fill=gray!7,inner sep=1pt] at (2.48,2.55) {$\Omega$};
  \node[anchor=west,fill=white,inner sep=1pt] at (2.10,0.48) {$\partial\Omega$};

  \fill[gray!20] (0,\rad) circle (\rad);
  \draw[thick] (0,\rad) circle (\rad);
  \fill[white] (0,-\rad) circle (\rad);
  \draw[thick,dashed] (0,-\rad) circle (\rad);

  \draw[densely dashed,gray!70] (0,-2.82) -- (0,2.82);
  \node[point] (X) at (0,0) {};
  \node[anchor=north west] at (0.10,-0.10) {$x_y$};
  \node[point] (Y) at (0,\dy) {};
  \node[anchor=west,fill=gray!20,inner sep=1pt] at (0.15,\dy+0.02) {$y$};

  \node[anchor=west,fill=gray!20,inner sep=1.2pt] at (0.52,1.82)
    {$B^{\mathrm{int}}_{\rho_*}$};
  \node[anchor=west,fill=white,inner sep=1.2pt] at (0.52,-1.82)
    {$B^{\mathrm{ext}}_{\rho_*}$};

  \draw[<->] (0.52,0.05) -- (0.52,\dy-0.05);
  \node[anchor=west,fill=white,inner sep=1pt] at (0.60,0.5*\dy) {$d(y)$};

  \draw[<->] (-0.58,0.06) -- (-0.58,\rad-0.06);
  \node[anchor=east,fill=gray!20,inner sep=1pt] at (-0.67,0.5*\rad) {$\rho_*$};
  \draw[<->] (-0.58,-0.06) -- (-0.58,-\rad+0.06);
  \node[anchor=east,fill=white,inner sep=1pt] at (-0.67,-0.5*\rad) {$\rho_*$};
\end{tikzpicture}
}
\caption{Planar cross-section of the tangent-ball comparison at a nearest boundary point $x_y$. The interior and exterior balls have radius $\rho_*$, while $d(y)=\operatorname{dist}(y,\partial\Omega)$.}
\label{fig:tangent-ball-comparison}
\end{figure}
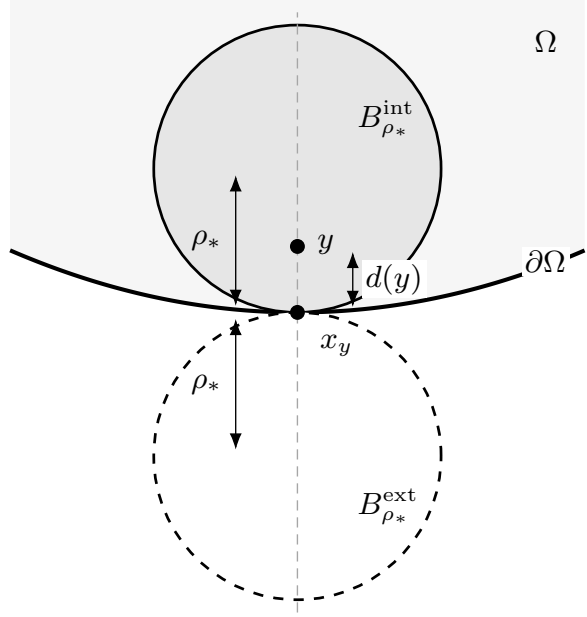

\begin{theorem}[Near-boundary estimates]
\label{thm:near-boundary-alpha}
Let $d\in\{2,3\}$, and let $\Omega\subset\R^d$ be either an
exterior $C^{1,1}$-domain or a special Lipschitz domain whose defining function is of class $C^{1,1}$
with globally Lipschitz gradient. Then there exists $\rho_*>0$
such that $\Omega$ satisfies the uniform interior and exterior ball condition
with radius $\rho_*$. Set
\begin{equation*}
d(y):=\operatorname{dist}(y,\partial\Omega).
\end{equation*}
For $0<d(y)<\rho_*$, the following estimates hold.

If $n=3$, then
\begin{equation}
\label{eq:near-boundary-estimate-3d}
-\frac{\rho_*}{4\pi d(y)(2\rho_*-d(y))}
\leq
\alpha_c^\Omega(y)
\leq
-\frac{\rho_*}{4\pi d(y)(2\rho_*+d(y))},
\end{equation}
and consequently
\begin{equation*}
\alpha_c^\Omega(y)
=
-\frac{1}{8\pi d(y)}+O(1)
\qquad\text{as }d(y)\downarrow0.
\end{equation*}

If $n=2$, then
\begin{equation}
\label{eq:near-boundary-estimate-2d}
\frac{1}{2\pi}\log\!\left(
\frac{d(y)(2\rho_*-d(y))}{\rho_*}
\right)
\leq
\alpha_c^{(2),\Omega}(y)
\leq
\frac{1}{2\pi}\log\!\left(
\frac{d(y)(2\rho_*+d(y))}{\rho_*}
\right),
\end{equation}
and consequently
\begin{equation*}
\alpha_c^{(2),\Omega}(y)
=
\frac{1}{2\pi}\log(2d(y))+O(d(y))
=
\frac{1}{2\pi}\log d(y)+O(1)
\qquad\text{as }d(y)\downarrow0.
\end{equation*}
\end{theorem}

\begin{proof}
For an exterior \(C^{1,1}\)-domain, the uniform interior and exterior ball
condition follows from the standard tubular-neighbourhood property and
compactness of the boundary. More explicitly, there exists \(\rho_*>0\) such
that the normal map
\begin{equation*}
\Phi:\partial\Omega\times(-\rho_*,\rho_*)\to\R^d,
\qquad
\Phi(x,t)=x+t n(x),
\end{equation*}
where \(n(x)\) is the inward unit normal, is one-to-one onto a neighbourhood
of \(\partial\Omega\). In particular, if \(d(y)<\rho_*\), then the nearest
boundary point \(x_y\) is unique and
\begin{equation*}
y=x_y+d(y)n(x_y).
\end{equation*}
For a global \(C^{1,1}\)-graph, the same conclusion follows from the global
Lipschitz bound on the defining function and the uniform Lipschitz bound on
its gradient, which give a uniform tubular neighbourhood and hence uniform
interior and exterior tangent balls.

Fix $y\in\Omega$ with $0<\ell:=d(y)<\rho_*$. Let $x_y\in\partial\Omega$ be the nearest boundary point and write $n_y:=n(x_y)$. Then
\begin{equation*}
y=x_y+\ell n_y.
\end{equation*}
Let
\begin{equation*}
c_{\mathrm{int}}:=x_y+\rho_*n_y,
\qquad
c_{\mathrm{ext}}:=x_y-\rho_*n_y.
\end{equation*}
The tangent balls
\begin{equation*}
B_{\rho_*}(c_{\mathrm{int}})\subset\Omega,
\qquad
B_{\rho_*}(c_{\mathrm{ext}})\subset\Omega^c
\end{equation*}
give the inclusions
\begin{equation*}
B_{\rho_*}(c_{\mathrm{int}})
\subset
\Omega
\subset
\R^d\setminus\overline{B_{\rho_*}(c_{\mathrm{ext}})}.
\end{equation*}
Furthermore,
\begin{equation*}
|y-c_{\mathrm{int}}|=\rho_*-\ell,
\qquad
|y-c_{\mathrm{ext}}|=\rho_*+\ell.
\end{equation*}

Assume first that $n=3$. Domain monotonicity, Lemma~\ref{lem:critical_ball_disk}, and the exterior-ball
formulas yield
\begin{align*}
-\frac{\rho_*}
 {4\pi\bigl(\rho_*^2-(\rho_*-\ell)^2\bigr)}
&\leq \alpha_c^\Omega(y)\leq
-\frac{\rho_*}
 {4\pi\bigl((\rho_*+\ell)^2-\rho_*^2\bigr)}.
\end{align*}
After expanding the two differences of squares, this is exactly
\eqref{eq:near-boundary-estimate-3d}. The identities
\begin{equation*}
\frac{\rho_*}{\ell(2\rho_*-\ell)}
=
\frac{1}{2\ell}+\frac{1}{2(2\rho_*-\ell)},
\qquad
\frac{\rho_*}{\ell(2\rho_*+\ell)}
=
\frac{1}{2\ell}-\frac{1}{2(2\rho_*+\ell)}
\end{equation*}
then give the stated three-dimensional asymptotic.

If $n=2$, the same comparison, now using Lemma~\ref{lem:critical_ball_disk} and the exterior-disk
formulas, gives
\begin{align*}
\frac{1}{2\pi}
\log\!\left(
\frac{\rho_*^2-(\rho_*-\ell)^2}{\rho_*}
\right)
\leq \alpha_c^{(2),\Omega}(y)
\leq
\frac{1}{2\pi}
\log\!\left(
\frac{(\rho_*+\ell)^2-\rho_*^2}{\rho_*}
\right),
\end{align*}
which is \eqref{eq:near-boundary-estimate-2d}. Factoring
$2\ell$ in both logarithms gives
\begin{equation*}
\alpha_c^{(2),\Omega}(y)
-
\frac{1}{2\pi}\log(2\ell)
=
O(\ell),
\end{equation*}
and completes the proof.
\end{proof}

\begin{corollary}
\label{cor:no_eigenvalue_near_boundary}
Under the assumptions of Theorem~\ref{thm:near-boundary-alpha}, fix
$\alpha\in\R$. There exists $d_\alpha>0$ such that, whenever
$d(y)<d_\alpha$, the corresponding
one-centre point-interaction Hamiltonian has no nonpositive eigenvalue.
\end{corollary}

\begin{proof}
Theorem~\ref{thm:near-boundary-alpha} gives
$\alpha_c^\Omega(y)\to-\infty$ in dimension three and
$\alpha_c^{(2),\Omega}(y)\to-\infty$ in dimension two as
$y$ approaches the boundary. Thus $\alpha$ is strictly larger than the
critical coupling for all sufficiently small boundary distances. The 
characteristic equation then excludes both a negative eigenvalue and a zero-energy
eigenstate.
\end{proof}

\begin{remark}
The leading asymptotics are exactly those of the flat models:
\begin{equation*}
\alpha_c^\Omega(y)
\sim
-\frac{1}{8\pi\operatorname{dist}(y,\partial\Omega)}
\qquad(n=3),
\end{equation*}
and
\begin{equation*}
\alpha_c^{(2),\Omega}(y)
=
\frac{1}{2\pi}\log\!\bigl(2\operatorname{dist}(y,\partial\Omega)\bigr)
+o(1)
\qquad(n=2).
\end{equation*}
Thus the first-order behavior is local and independent of the global geometry.
\end{remark}

\begin{remark}
Puri and Vainberg \cite{PuriVainberg2020} study critical couplings for regular,
compactly supported potentials in exterior elliptic problems. Their analysis is
based on a Birman--Schwinger operator, whereas the present one-centre problem is
reduced to the characteristic equation
\begin{equation*}
\alpha-M_y^\Omega(-\lambda)=0.
\end{equation*}
The contrast is particularly clear near a flat Dirichlet boundary. In the regular
potential problem the critical parameter may remain bounded in dimensions
$d\geq3$; see \cite[Theorem~3.1]{PuriVainberg2020}. For a point interaction,
Theorem~\ref{thm:near-boundary-alpha} shows instead that the critical coupling
tends to $-\infty$ in both dimensions considered here. This reflects the genuinely
singular nature of the extension-theoretic perturbation.
\end{remark}

\subsection{Critical threshold states}\label{subsec:critical-threshold-states}

At the critical coupling, the zero-energy state is the Green function
$G_0^\Omega(\cdot,y)$, and its far-field behavior determines whether zero is an
eigenvalue or a resonance. Actually, a variety of behavior appears already in the model domains. Exterior domains produce a monopole resonance in
three dimensions and an $s$-wave resonance in two dimensions. A
three-dimensional special Lipschitz domain contained in a half-space has instead
a threshold eigenvalue. Planar wedges exhibit the wider range of behaviors allowed
by conical ends.

\subsubsection{Exterior domains}

To state the following Theorem we need a premise. Let $\Omega=\R^3\setminus\overline K$, where
$K\Subset\R^3$ is a bounded $C^{1,1}$-domain. Let $p_K$ be the capacitary potential of $K$, namely the unique harmonic function in $\Omega$ such that
\begin{equation*}
\Delta p_K=0\quad\text{in }\Omega,
\qquad
p_K=1\quad\text{on }\partial\Omega,
\qquad
p_K(x)\longrightarrow0\quad\text{as }|x|\to\infty.
\end{equation*}
Such unique solution exists, for example as a consequence of the classical Perron method.
Set $\omega_\infty:=1-p_K$. Equivalently, $\omega_\infty$ is the unique solution of
\begin{equation*}
\Delta\omega_\infty=0\quad\text{in }\Omega,
\qquad
\omega_\infty=0\quad\text{on }\partial\Omega,
\qquad
\omega_\infty(x)\longrightarrow1\quad\text{as }|x|\to\infty.
\end{equation*}

\begin{theorem}
\label{thm:exterior_threshold_not_attained}
Let $\Omega=\R^3\setminus\overline K$, where
$K\Subset\R^3$ is a bounded $C^{1,1}$-domain, and let
$y\in\Omega$. Then
\begin{equation}
\label{eq:exterior-threshold-tail-3d}
G_0^\Omega(x,y)
=
\frac{\omega_\infty(y)}{4\pi|x|}+O(|x|^{-2})
\qquad\text{as }|x|\to\infty,
\end{equation}
with $0<\omega_\infty(y)<1$. Consequently,
\begin{equation*}
G_0^\Omega(\cdot,y)
\in
\bigcap_{s>1/2}L^2_{-s}(\Omega)\setminus L^2(\Omega),
\end{equation*}
and the critical state is a monopole threshold resonance. In particular,
\begin{equation*}
0\notin\sigma_p\bigl(A_{\alpha_c^\Omega(y),y}^\Omega\bigr).
\end{equation*}
\end{theorem}

\begin{proof}
Set $\psi_0:=G_0^\Omega(\cdot,y)$. Since $\Omega\subset\R^3$, domain
monotonicity gives
\begin{equation*}
0<\psi_0(x)
\leq
G_0^{\R^3}(x,y)
=
\frac{1}{4\pi|x-y|},
\qquad x\in\Omega\setminus\{y\}.
\end{equation*}
In particular, $\psi_0(x)=O(|x|^{-1})$. Choose $R_0$ so that
$\overline K\cup\{y\}\subset B_{R_0}$. For $|x|>R_0$, $\psi_0$ is
harmonic. 
Since \(\psi_0(x)=O(|x|^{-1})\), the Kelvin transform
\[
\widetilde\psi_0(\xi)
:=
|\xi|^{-1}\psi_0\!\left(\frac{\xi}{|\xi|^2}\right)
\]
is bounded in a punctured neighbourhood of the origin. It is harmonic there by
the Kelvin transform, hence the isolated singularity at \(0\) is removable. Thus
\(\widetilde\psi_0\) is harmonic, and therefore real analytic, near \(0\). Writing
\[
\widetilde\psi_0(\xi)
=
A(y)+b(y)\cdot \xi+O(|\xi|^2),
\qquad \xi\to0,
\]
and substituting \(\xi=x/|x|^2=\omega/r\), \(r=|x|\), \(\omega=x/|x|\), gives
\[
\psi_0(r\omega)
=
\frac{A(y)}{r}
+
\frac{b(y)\cdot\omega}{r^2}
+
O(r^{-3}).
\]
Equivalently, with \(A(y)=c(y)/(4\pi)\),
\[
\psi_0(x)
=
\frac{c(y)}{4\pi |x|}
+
O(|x|^{-2}),
\qquad
\partial_r\psi_0(x)
=
-\frac{c(y)}{4\pi |x|^2}
+
O(|x|^{-3}).
\]
as $|x|\to\infty$, for some $c(y)\in\R$. The same argument applied to
$1-\omega_\infty$ gives
\begin{equation*}
\omega_\infty(x)=1+O(|x|^{-1}),
\qquad
\partial_r\omega_\infty(x)=O(|x|^{-2}).
\end{equation*}

We identify $c(y)$ by Green's second identity. Let
\begin{equation*}
D_{R,\varepsilon}
:=
\Omega\cap B_R\setminus\overline{B_\varepsilon(y)},
\qquad R>R_0.
\end{equation*}
Since both functions vanish on $\partial\Omega$,
\begin{equation*}
\int_{\partial D_{R,\varepsilon}}
\bigl(
\psi_0\partial_n\omega_\infty
-
\omega_\infty\partial_n\psi_0
\bigr)\,d\sigma
=0.
\end{equation*}
On $\partial B_\varepsilon(y)$, the outward normal of
$D_{R,\varepsilon}$ is directed towards $y$, and
\begin{equation*}
\psi_0(x)=\frac{1}{4\pi\varepsilon}+O(1),
\qquad
\partial_n\psi_0(x)=\frac{1}{4\pi\varepsilon^2}+O(1).
\end{equation*}
Since $\omega_\infty(x)=\omega_\infty(y)+O(\varepsilon)$, the inner
boundary integral converges to $-\omega_\infty(y)$. Consequently,
\begin{equation*}
\int_{\partial B_R}
\bigl(
\psi_0\partial_r\omega_\infty
-
\omega_\infty\partial_r\psi_0
\bigr)\,d\sigma
=
\omega_\infty(y).
\end{equation*}
The first term on the left is $O(R^{-1})$, while the second equals
$c(y)+O(R^{-1})$. Letting $R\to\infty$ yields
\begin{equation*}
c(y)=\omega_\infty(y).
\end{equation*}
The maximum principle applied to $\omega_\infty$ and
$1-\omega_\infty$, followed by the strong maximum principle, gives
$0<\omega_\infty<1$ in $\Omega$. This proves
\eqref{eq:exterior-threshold-tail-3d} with a nonzero coefficient.

The local Coulomb singularity belongs to $L^2_{\mathrm{loc}}$. At infinity,
the leading term has size $|x|^{-1}$, and hence
\begin{equation*}
\int_{|x|>R}
\langle x\rangle^{-2s}|G_0^\Omega(x,y)|^2\,dx
<\infty
\quad\Longleftrightarrow\quad
s>\frac12.
\end{equation*}
Thus the critical state is a monopole threshold resonance and cannot be a zero
eigenfunction.
\end{proof}

\begin{theorem}
\label{thm:exterior_threshold_not_attained_2d}
Let $\Omega=\R^2\setminus\overline K$, where
$K\Subset\R^2$ is a bounded $C^{1,1}$-domain, and let
$y\in\Omega$. There exist $c_\Omega(y)>0$ and
$a_\Omega(y)\in\R^2$ such that
\begin{equation}
\label{eq:exterior-threshold-tail-2d}
G_0^\Omega(x,y)
=
c_\Omega(y)
+
\frac{a_\Omega(y)\cdot x}{|x|^2}
+
O(|x|^{-2})
\qquad\text{as }|x|\to\infty.
\end{equation}
Consequently,
\begin{equation*}
G_0^\Omega(\cdot,y)
\in
\bigcap_{s>1}L^2_{-s}(\Omega)\setminus L^2(\Omega),
\end{equation*}
and the critical state is an $s$-wave threshold resonance. In particular,
\begin{equation*}
0\notin
\sigma_p\bigl(A_{\alpha_c^{(2),\Omega}(y),y}^{(2),\Omega}\bigr).
\end{equation*}
\end{theorem}

\begin{proof}

After a harmless translation of coordinates, we may assume that
\(0\in K^\circ\). Let
\[
\mathcal I(x):=\frac{x}{|x|^2}.
\]
Since \(K\) contains a neighbourhood of the origin and is bounded, the set
\[
\widetilde\Omega:=\mathcal I(\Omega)\cup\{0\}
\]
is a bounded \(C^{1,1}\)-domain, and \(0\) is an interior point of
\(\widetilde\Omega\).

The map \(\mathcal I\) is anti-conformal (conformal up to a reverse of the orientation) in
\(\mathbb R^2\setminus\{0\}\). By conformal invariance of the two-dimensional
Dirichlet Green function, with the present logarithmic normalization,
\[
G_0^\Omega(x,y)
=
G_0^{\widetilde\Omega}\bigl(\mathcal I(x),\mathcal I(y)\bigr).
\]
Indeed, the right-hand side has the same logarithmic singularity at \(x=y\),
vanishes on \(\partial\Omega\), and is bounded at infinity; uniqueness gives
the identity.

Since \(\mathcal I(y)\neq0\), the function
$
X\longmapsto G_0^{\widetilde\Omega}\bigl(X,\mathcal I(y)\bigr)
$
is harmonic in a neighbourhood of \(X=0\). Hence it is real analytic there, and
\begin{equation*}
G_0^{\widetilde\Omega}\bigl(X,\mathcal I(y)\bigr)
=
c_\Omega(y)+a_\Omega(y)\cdot X+O(|X|^2),
\qquad X\to0,
\end{equation*}
where
\begin{equation*}
c_\Omega(y):=
G_0^{\widetilde\Omega}\bigl(0,\mathcal I(y)\bigr),
\qquad
a_\Omega(y):=
\nabla_XG_0^{\widetilde\Omega}\bigl(0,\mathcal I(y)\bigr).
\end{equation*}

Since \(0\) is an interior point of \(\widetilde\Omega\) and
\(\mathcal I(y)\neq0\), positivity of the Dirichlet Green function and the
strong maximum principle give
\[
c_\Omega(y)=G_0^{\widetilde\Omega}\bigl(0,\mathcal I(y)\bigr)>0.
\]
The logarithmic singularity at $y$ is locally square-integrable. On the other
hand, the positive constant limit at infinity excludes square integrability, and
\begin{equation*}
\int_{|x|>R}
\langle x\rangle^{-2s}|G_0^\Omega(x,y)|^2\,dx
<\infty
\quad\Longleftrightarrow\quad
s>1.
\end{equation*}
Thus the critical state is an $s$-wave threshold resonance.
\end{proof}

\subsubsection{Three-dimensional domains contained in a half-space and generalizations}
\label{subsec:threshold-halfspace-subdomains}

\begin{theorem}
\label{thm:threshold-halfspace-subdomain}
Let $\Omega\subset\R^3_+$ be a special Lipschitz domain, and let
$y\in\Omega$. Then
\begin{equation*}
G_0^\Omega(\cdot,y)\in L^2(\Omega),
\end{equation*}
and
\begin{equation*}
\ker A_{\alpha_c^\Omega(y),y}^\Omega
=
\operatorname{span}\{G_0^\Omega(\cdot,y)\}.
\end{equation*}
Consequently, zero belongs to the essential spectrum and it is a simple threshold eigenvalue.
\end{theorem}

\begin{proof}
After a rigid
motion, assume that
\begin{equation*}
\Omega\subset\mathbb R^3_+:=\{x_3>0\},
\qquad
y\in\Omega,
\qquad
\bar y:=(y_1,y_2,-y_3).
\end{equation*}
For \(\lambda>0\), let \(G_{-\lambda}^\Omega(\cdot,y)\) denote the Dirichlet
Green kernel of \(A_D^\Omega+\lambda\), where \(A_D^\Omega=-\Delta_D^\Omega\).

First, if \(0<\lambda_1<\lambda_2\), the resolvent identity gives
\begin{equation*}
(A_D^\Omega+\lambda_1)^{-1}
-
(A_D^\Omega+\lambda_2)^{-1}
=
(\lambda_2-\lambda_1)
(A_D^\Omega+\lambda_1)^{-1}
(A_D^\Omega+\lambda_2)^{-1}.
\end{equation*}
Since the Dirichlet resolvent is positivity preserving, the right-hand side is
positivity preserving. Hence, at the level of kernels,
\begin{equation*}
G_{-\lambda_1}^\Omega(x,y)
\geq
G_{-\lambda_2}^\Omega(x,y),
\qquad x\neq y.
\end{equation*}
Thus \(G_{-\lambda}^\Omega(x,y)\) is monotone increasing as
\(\lambda\downarrow0\).

Second, Dirichlet domain monotonicity yields, for every \(\lambda>0\),
\begin{equation*}
0<G_{-\lambda}^\Omega(x,y)
\leq
G_{-\lambda}^{\mathbb R^3_+}(x,y),
\qquad x\in\Omega\setminus\{y\}.
\end{equation*}
Indeed, the two kernels have the same Coulomb singularity at \(y\). Hence the
difference
\begin{equation*}
w(x):=
G_{-\lambda}^{\mathbb R^3_+}(x,y)-G_{-\lambda}^{\Omega}(x,y)
\end{equation*}
has a removable singularity at \(y\) and satisfies
\begin{equation*}
(-\Delta+\lambda)w=0
\qquad\text{in }\Omega.
\end{equation*}
Moreover, the Dirichlet trace of \(G_{-\lambda}^{\Omega}(\cdot,y)\) vanishes on
\(\partial\Omega\), while
$
G_{-\lambda}^{\mathbb R^3_+}(\cdot,y)\geq0\ \ \text{in }\mathbb R^3_+.
$
Therefore
\begin{equation*}
w|_{\partial\Omega}
=
G_{-\lambda}^{\mathbb R^3_+}(\cdot,y)\big|_{\partial\Omega}
\geq0.
\end{equation*}
Finally, since \(\lambda>0\), both kernels decay exponentially at infinity. Now we use exhaustion: applying
the maximum principle to \(w\) on bounded truncations of \(\Omega\) and letting the
truncation radius tend to infinity gives
\begin{equation*}
w\geq0\quad\text{in }\Omega.
\end{equation*}
Thus
\begin{equation*}
G_{-\lambda}^{\Omega}(x,y)\leq
G_{-\lambda}^{\mathbb R^3_+}(x,y).
\end{equation*}
Consequently the monotone limit
\begin{equation*}
G_0^\Omega(x,y):=
\lim_{\lambda\downarrow0}G_{-\lambda}^\Omega(x,y)
\end{equation*}
exists and satisfies
\begin{equation*}
0<G_0^\Omega(x,y)
\leq
\lim_{\lambda\downarrow0}G_{-\lambda}^{\mathbb R^3_+}(x,y)
=
G_0^{\mathbb R^3_+}(x,y).
\end{equation*}
By the image formula in the half-space,
\begin{equation*}
G_0^{\mathbb R^3_+}(x,y)
=
\frac{1}{4\pi|x-y|}
-
\frac{1}{4\pi|x-\bar y|}.
\end{equation*}
Therefore
\begin{equation*}
0<G_0^\Omega(x,y)
\leq
\frac{1}{4\pi|x-y|}
-
\frac{1}{4\pi|x-\bar y|}.
\end{equation*}
Now, the right hand-side is square integrable. For \(|x|\geq2|y|\), the mean-value
theorem applied to the Coulomb kernel gives
\begin{equation*}
\left|
\frac{1}{|x-y|}
-
\frac{1}{|x-\bar y|}
\right|
\leq
|y-\bar y|
\sup_{\eta\in[y,\bar y]}
\left|\nabla_\eta\frac{1}{|x-\eta|}\right|
\leq
C_y |x|^{-2}.
\end{equation*}
Hence
\begin{equation*}
G_0^\Omega(x,y)=O(|x|^{-2}),
\qquad |x|\to\infty.
\end{equation*}
This decay is square-integrable at infinity in dimension three. Near the pole,
\(G_0^\Omega(x,y)\) has the usual Coulomb singularity
\begin{equation*}
G_0^\Omega(x,y)=\frac{1}{4\pi|x-y|}+O(1),
\qquad x\to y,
\end{equation*}
which is locally square-integrable in dimension three. Thus
$
G_0^\Omega(\cdot,y)\in L^2(\Omega).
$
We next verify the operator-domain condition. Fix $z<0$ and set
\begin{equation*}
u_z:=G_0^\Omega(\cdot,y)-G_z^\Omega(\cdot,y).
\end{equation*}
The singularities at $y$ cancel, and in the distributional sense
\begin{equation*}
(A_D-z)u_z=-zG_0^\Omega(\cdot,y)\in L^2(\Omega).
\end{equation*}
Therefore
\begin{equation*}
u_z=(A_D-z)^{-1}\bigl(-zG_0^\Omega(\cdot,y)\bigr)
\in\operatorname{dom}(A_D).
\end{equation*}
Taking the regular part of $G_0^\Omega(\cdot,y)-G_z^\Omega(\cdot,y)$ at the pole gives
\begin{equation*}
u_z(y)=M_y^\Omega(0)-M_y^\Omega(z)
=
\alpha_c^\Omega(y)-M_y^\Omega(z).
\end{equation*}
The domain formula in Theorem~\ref{thm:krein}, with charge $q=1$, now shows
that $G_0^\Omega(\cdot,y)$ lies in the domain of
$A_{\alpha_c^\Omega(y),y}^\Omega$ and that
\begin{equation*}
A_{\alpha_c^\Omega(y),y}^\Omega G_0^\Omega(\cdot,y)=0.
\end{equation*}
Taking \(\alpha=\alpha_c^\Omega(y)\), the zero eigenvalue found above lies at
the bottom of the essential spectrum. Therefore it is a threshold eigenvalue.
Its simplicity follows from the one-centre extension theory.
\end{proof}

The half-space assumption is used only to obtain the comparison estimate
\begin{equation*}
G_0^\Omega(x,y)\leq G_0^{\R^3_+}(x,y),
\end{equation*}
which gives the square-integrable tail of the critical state. Thus the same
conclusion would hold in any class of domains for which the zero-energy Green
function has an $L^2$ tail at infinity, for instance an $O(|x|^{-2})$ decay in
three dimensions, and for which the bottom of the essential spectrum of the
Dirichlet Laplacian is zero. For ends asymptotic to a half-space such behavior
is natural, but it requires an independent zero-energy asymptotic analysis
of the Green function, given in the next result.

\begin{theorem}
\label{thm:threshold-asymptotic-halfspace}
Let \(\Omega\subset\mathbb R^3\) be a special Lipschitz domain. Assume that,
after a rigid motion,
\begin{equation*}
\Omega=\{(x',x_3)\in\mathbb R^2\times\mathbb R:x_3>\varphi(x')\},
\end{equation*}
where \(\varphi:\mathbb R^2\to\mathbb R\) is globally Lipschitz and
\begin{equation*}
\varphi(x')\longrightarrow0
\qquad\text{as } |x'|\to\infty .
\end{equation*}
Let \(y\in\Omega\). Then
\begin{equation*}
G_0^\Omega(\cdot,y)\in L^2(\Omega),
\end{equation*}
and
\begin{equation*}
\ker A_{\alpha_c^\Omega(y),y}^\Omega
=
\operatorname{span}\{G_0^\Omega(\cdot,y)\}.
\end{equation*}
In particular, zero is a threshold eigenvalue of
\(A_{\alpha_c^\Omega(y),y}^\Omega\).
\end{theorem}

\begin{proof}
We prove the only non immediate point, namely the square integrability of the zero-energy Green
function at infinity.
Choose \(\varepsilon>0\) so large that
\begin{equation*}
y\in H_\varepsilon:=\{x_3>-\varepsilon\}.
\end{equation*}
Since \(\varphi(x')\to0\) as \(|x'|\to\infty\), there exists \(R_0>0\) such that
\begin{equation*}
\varphi(x')>-\frac{\varepsilon}{2}
\qquad\text{for } |x'|>R_0 .
\end{equation*}
Increasing \(R_0\), if necessary, we may assume that \(y\in B_{R_0}\) and
\begin{equation*}
\Omega\cap\{|x|>R_0\}
\subset
\{x_3>-\varepsilon/2\}
\Subset H_\varepsilon .
\end{equation*}

For \(\lambda>0\), let
\begin{equation*}
u_\lambda(x):=G_{-\lambda}^\Omega(x,y),
\end{equation*}
and let
\begin{equation*}
v(x):=G_0^{H_\varepsilon}(x,y)
\end{equation*}
be the zero-energy Dirichlet Green function of the half-space
\(H_\varepsilon\). Then \(v>0\) in \(H_\varepsilon\), and \(v\) is harmonic in
\(\Omega\cap\{|x|>R_0\}\), because \(y\in B_{R_0}\). Moreover,
\(u_\lambda\) solves
\begin{equation*}
(-\Delta+\lambda)u_\lambda=0
\qquad\text{in } \Omega\cap\{|x|>R_0\}.
\end{equation*}

We claim that there is a constant \(C>0\), independent of \(\lambda\), such that
\begin{equation*}
u_\lambda(x)\leq C v(x),
\qquad x\in \Omega\cap\{|x|>R_0\}.
\end{equation*}
Indeed, by comparison with the free-space Yukawa kernel,
\begin{equation*}
0<u_\lambda(x)
\leq
\frac{e^{-\sqrt{\lambda}|x-y|}}{4\pi |x-y|}
\leq
\frac{1}{4\pi |x-y|}.
\end{equation*}
On the compact set
\begin{equation*}
\Gamma_{R_0}:=\overline{\Omega\cap\partial B_{R_0}},
\end{equation*}
the function \(v\) is strictly positive. Hence, choosing \(C>0\) large enough,
we have
\begin{equation*}
C v\geq u_\lambda
\qquad\text{on } \Gamma_{R_0},
\end{equation*}
uniformly for all \(\lambda>0\).

Set
\begin{equation*}
w_\lambda:=C v-u_\lambda .
\end{equation*}
In \(\Omega\cap\{|x|>R_0\}\) one has
\begin{equation*}
(-\Delta+\lambda)w_\lambda
=
C\lambda v
\geq0.
\end{equation*}
On the genuine boundary part
\(\partial\Omega\cap\{|x|>R_0\}\), the Dirichlet trace of \(u_\lambda\)
vanishes and \(v\geq0\). Hence
\begin{equation*}
w_\lambda\geq0
\qquad\text{on } \partial\Omega\cap\{|x|>R_0\}.
\end{equation*}
On the artificial boundary \(\Gamma_{R_0}\), \(w_\lambda\geq0\) by the choice
of \(C\).

It remains only to control the boundary at infinity when applying the maximum
principle on bounded truncations. Since
\begin{equation*}
\Omega\cap\{|x|>R_0\}\subset\{x_3>-\varepsilon/2\},
\end{equation*}
the explicit half-space formula gives, after increasing \(R_0\) if necessary,
\begin{equation*}
v(x)\geq c |x|^{-3}
\qquad
\text{for } x\in\Omega,\ |x|>R_0,
\end{equation*}
with some \(c>0\). On the other hand,
\begin{equation*}
u_\lambda(x)
\leq
\frac{e^{-\sqrt{\lambda}(|x|-|y|)}}{4\pi(|x|-|y|)}.
\end{equation*}
Thus, for each fixed \(\lambda>0\), \(w_\lambda\geq0\) on
\(\Omega\cap\partial B_R\) for all sufficiently large \(R\). Applying the maximum
principle to \(w_\lambda\) on
\begin{equation*}
\Omega\cap\{R_0<|x|<R\}
\end{equation*}
and then letting \(R\to\infty\), we obtain
\begin{equation*}
u_\lambda(x)\leq C v(x),
\qquad x\in\Omega\cap\{|x|>R_0\}.
\end{equation*}
Letting \(\lambda\downarrow0\) yields
\begin{equation*}
G_0^\Omega(x,y)\leq C G_0^{H_\varepsilon}(x,y),
\qquad x\in\Omega\cap\{|x|>R_0\}.
\end{equation*}

The half-space Green function satisfies
\begin{equation*}
G_0^{H_\varepsilon}(x,y)=O(|x|^{-2})
\qquad\text{as } |x|\to\infty .
\end{equation*}
Therefore
\begin{equation*}
G_0^\Omega(x,y)=O(|x|^{-2})
\qquad\text{as } |x|\to\infty .
\end{equation*}
This decay is square-integrable at infinity in dimension three. Near the pole,
\(G_0^\Omega(\cdot,y)\) has the usual Coulomb singularity
\begin{equation*}
G_0^\Omega(x,y)
=
\frac{1}{4\pi|x-y|}
+
O(1),
\qquad x\to y,
\end{equation*}
which is locally square-integrable in dimension three. Hence
$
G_0^\Omega(\cdot,y)\in L^2(\Omega).
$
Verification that this \(L^2\)-function is in the domain and is the critical eigenfunction is exactly as in Theorem \ref{thm:threshold-halfspace-subdomain}\(z<0\) and it is not repeated.
\end{proof}

\begin{remark}[Dependence on the boundary condition]\label{nbc-halfspace}
The preceding comparison is specific to Dirichlet boundary conditions. For the
Neumann half-space the image charge has the same sign,
\begin{equation*}
G_0^{\R^3_+,N}(x,y)
=
\frac{1}{4\pi|x-y|}
+
\frac{1}{4\pi|x-\bar y|},
\end{equation*}
so the monopole tail survives. In addition, Neumann Green functions do not enjoy
the same pointwise monotonicity under domain inclusion.
\end{remark}

\subsubsection{Planar wedges}

For $0<\beta<2\pi$, set
\begin{equation*}
\mathcal S_\beta
:=
\{z=re^{i\theta}:r>0,\ 0<\theta<\beta\}\subset\C,
\qquad
q:=\frac{\pi}{\beta}.
\end{equation*}
After a rigid motion, $\mathcal S_\beta$ is a special Lipschitz domain, and the branch
of $z\mapsto z^q$ determined by $0<\arg z<\beta$ maps it conformally onto
the upper half-plane.

\begin{theorem}
\label{thm:planar-wedge-threshold}
Let $y=\rho e^{i\varphi}\in \mathcal S_\beta$. The critical coupling is the
quantity in \eqref{eq:sector-critical-coupling-exact-model}. Moreover, uniformly
for $0\leq\theta\leq\beta$,
\begin{equation}
\label{eq:wedge-threshold-tail}
G_0^{\mathcal S_\beta}(re^{i\theta},y)
=
\frac{\rho^q\sin(q\varphi)}{\pi}
 r^{-q}\sin(q\theta)
+
O(r^{-2q})
\qquad\text{as }r\to\infty.
\end{equation}
Consequently:
\begin{enumerate}
\item[(i)] if $0<\beta<\pi$, then
$G_0^{\mathcal S_\beta}(\cdot,y)\in L^2(\mathcal S_\beta)$, and zero is a simple
threshold eigenvalue;

\item[(ii)] if $\beta=\pi$, then
\begin{equation*}
G_0^{\mathcal S_\pi}(\cdot,y)
\in
\bigcap_{s>0}L^2_{-s}(\mathcal S_\pi)\setminus L^2(\mathcal S_\pi),
\end{equation*}
and the critical state is the $p$-wave resonance of the half-plane;

\item[(iii)] if $\pi<\beta<2\pi$, then
\begin{equation*}
G_0^{\mathcal S_\beta}(\cdot,y)
\in
\bigcap_{s>1-\pi/\beta}L^2_{-s}(\mathcal S_\beta)\setminus L^2(\mathcal S_\beta),
\end{equation*}
and the critical state is a threshold resonance with decay
$|x|^{-\pi/\beta}$. It is neither of standard $s$-wave nor of standard
$p$-wave type.
\end{enumerate}
\end{theorem}

\begin{proof}

The zero-energy Green function and the critical coupling were computed in
Subsection~\ref{subsec:planar-sector-models}. We only extract here the
large-distance behavior and the corresponding threshold classification.

Let
\begin{equation*}
z=re^{i\theta},
\qquad
y=\rho e^{i\varphi},
\qquad
Z=z^q,
\qquad
Y=y^q.
\end{equation*}
Then, by \eqref{eq:sector-zero-green-exact-model},
\begin{equation*}
G_0^{\mathcal S_\beta}(z,y)
=
\frac{1}{2\pi}
\log\left|
\frac{1-\overline Y/Z}{1-Y/Z}
\right|.
\end{equation*}
Since
\begin{equation*}
\log|1-\zeta|
=
-\operatorname{Re}\zeta
+
O(|\zeta|^2),
\qquad
\zeta\to0,
\end{equation*}
uniformly for \(\zeta\) in bounded angular sectors, we obtain
\begin{equation*}
G_0^{\mathcal S_\beta}(z,y)
=
\frac{1}{2\pi}
\operatorname{Re}
\left(
\frac{Y-\overline Y}{Z}
\right)
+
O(|Z|^{-2}).
\end{equation*}
Now
\begin{equation*}
Y-\overline Y
=
2i\rho^q\sin(q\varphi),
\qquad
Z=r^q e^{iq\theta}.
\end{equation*}
Hence
\begin{equation*}
\operatorname{Re}
\left(
\frac{Y-\overline Y}{Z}
\right)
=
2\rho^q r^{-q}\sin(q\varphi)\sin(q\theta).
\end{equation*}
This gives
\begin{equation*}
G_0^{\mathcal S_\beta}(re^{i\theta},y)
=
\frac{\rho^q\sin(q\varphi)}{\pi}
r^{-q}\sin(q\theta)
+
O(r^{-2q}),
\end{equation*}
uniformly for \(0\leq\theta\leq\beta\), which is
\eqref{eq:wedge-threshold-tail}.

The logarithmic singularity at \(y\) is locally square-integrable. At infinity,
the leading angular profile is nonzero, because \(0<\varphi<\beta\) implies
\(\sin(q\varphi)>0\), and
\begin{equation*}
\int_0^\beta \sin^2(q\theta)\,\dd\theta
=
\frac{\beta}{2}.
\end{equation*}
Therefore the tail is square-integrable precisely when
\begin{equation*}
\int_R^\infty r^{1-2q}\,\dd r<\infty,
\end{equation*}
that is, precisely when \(q>1\), equivalently \(0<\beta<\pi\).

More generally, for the weighted spaces,
\begin{equation*}
\int_{\mathcal S_\beta\cap\{|x|>R\}}
\langle x\rangle^{-2s}
|G_0^{\mathcal S_\beta}(x,y)|^2\,\dd x
<\infty
\end{equation*}
holds precisely when
\begin{equation*}
\int_R^\infty r^{1-2q-2s}\,\dd r<\infty,
\end{equation*}
that is, precisely when
$
s>1-q.
$
This proves the stated \(L^2\) and weighted \(L^2\) alternatives. In particular,
for \(\beta=\pi\) one has \(q=1\), and the leading term is the usual dipole
tail \(r^{-1}\sin\theta\); this is the \(p\)-wave resonance of the half-plane.
For \(\pi<\beta<2\pi\), one has \(q<1\), so the critical state decays as
\(|x|^{-q}=|x|^{-\pi/\beta}\). This is a threshold resonance of conical type,
neither of the standard planar \(s\)-wave nor of the standard planar \(p\)-wave
type.

It remains to verify that when \(0<\beta<\pi\) the critical state is in the operator domain and solves the eigenvalue equation. This is checked again as in the proof of \ref{thm:threshold-halfspace-subdomain} and it is omitted.

\end{proof}

\begin{remark}
The three alternatives are visible in the following examples.
For $\beta=\pi/2$,
\begin{equation*}
\mathcal S_{\pi/2}=\{(x_1,x_2):x_2>|x_1|\},
\end{equation*}
and the critical state decays as $|x|^{-2}$, hence is an $L^2$ eigenfunction. For $\beta=\pi$, the domain is a half-plane and
the $|x|^{-1}$ dipole tail gives a $p$-wave resonance. For
$\beta=3\pi/2$,
\begin{equation*}
\mathcal S_{3\pi/2}=\{(x_1,x_2):x_2>-|x_1|\},
\end{equation*}
and the critical state decays as $|x|^{-2/3}$, giving a conical
resonance.
\end{remark}

Figure~\ref{fig:wedge-threshold-regimes} summarizes these three representative
apertures.

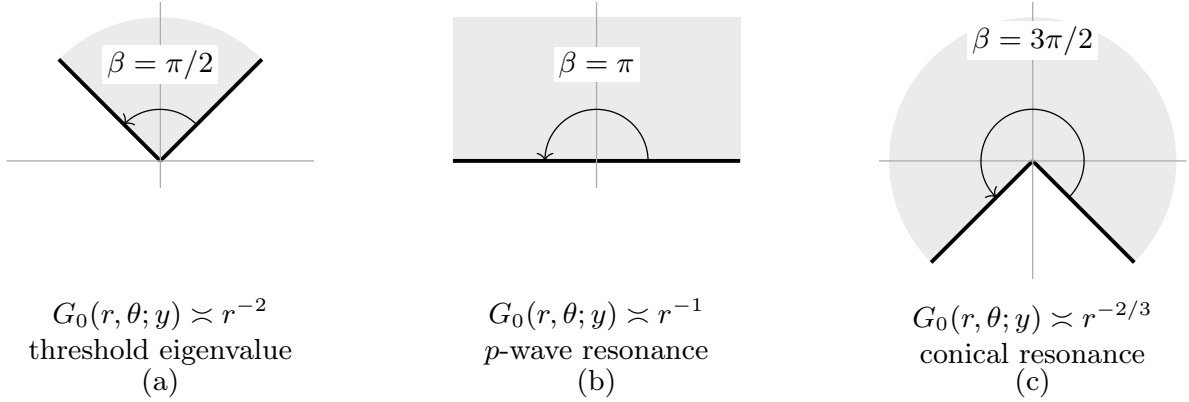
\begin{figure}[H]
\centering
\resizebox{\linewidth}{!}{%
\begin{tikzpicture}[
  x=0.92cm,y=0.92cm,
  every node/.style={font=\small},
  ray/.style={very thick},
  axis/.style={thin,gray!65},
  anglelabel/.style={fill=white,inner sep=1.4pt}
]
  \def\L{1.70}
  \path[use as bounding box] (-1.90,-2.80) rectangle (11.40,2.10);

  \begin{scope}[xshift=0cm]
    \fill[gray!16] (0,0) -- (45:\L) arc[start angle=45,end angle=135,radius=\L] -- cycle;
    \draw[ray] (0,0) -- (45:\L);
    \draw[ray] (0,0) -- (135:\L);
    \draw[axis] (-1.82,0)--(1.82,0);
    \draw[axis] (0,-0.32)--(0,1.88);
    \draw[->] (0.43,0.43) arc[start angle=45,end angle=135,radius=0.61];
    \node[anglelabel] at (0,1.14) {$\beta=\pi/2$};
    \node[align=center] at (0,-2.02)
      {$G_0(r,\theta;y)\asymp r^{-2}$\\
       threshold eigenvalue};
    \node at (0,-2.62) {\textnormal{(a)}};
  \end{scope}

  \begin{scope}[xshift=4.75cm]
    \fill[gray!16] (-\L,0) rectangle (\L,1.70);
    \draw[ray] (-\L,0)--(\L,0);
    \draw[axis] (0,-0.32)--(0,1.88);
    \draw[->] (0.61,0) arc[start angle=0,end angle=180,radius=0.61];
    \node[anglelabel] at (0,1.14) {$\beta=\pi$};
    \node[align=center] at (0,-2.02)
      {$G_0(r,\theta;y)\asymp r^{-1}$\\
       $p$-wave resonance};
    \node at (0,-2.62) {\textnormal{(b)}};
  \end{scope}

  \begin{scope}[xshift=9.50cm]
    \fill[gray!16] (0,0) -- (-45:\L) arc[start angle=-45,end angle=225,radius=\L] -- cycle;
    \draw[ray] (0,0)--(-45:\L);
    \draw[ray] (0,0)--(225:\L);
    \draw[axis] (-1.82,0)--(1.82,0);
    \draw[axis] (0,-1.40)--(0,1.88);
    \draw[->] (0.43,-0.43) arc[start angle=-45,end angle=225,radius=0.61];
    \node[anglelabel] at (0,1.42) {$\beta=3\pi/2$};
    \node[align=center] at (0,-2.02)
      {$G_0(r,\theta;y)\asymp r^{-2/3}$\\
       conical resonance};
    \node at (0,-2.62) {\textnormal{(c)}};
  \end{scope}
\end{tikzpicture}
}
\caption{Threshold behavior for three representative planar wedges. The decay exponent is $\pi/\beta$: a narrow wedge gives an eigenvalue, the half-plane gives the borderline $p$-wave resonance and a wide wedge gives a conical resonance.}
\label{fig:wedge-threshold-regimes}
\end{figure}

\begin{remark}
The corner is not responsible for the classification. Let $\Omega\subset\R^2$
be a $C^2$ special Lipschitz domain that agrees with $\mathcal S_\beta$ outside a
ball. After increasing $R$ if necessary, assume that $y\in\Omega\cap B_R$
and that
\begin{equation*}
\Omega\cap\{r>R\}=\mathcal S_\beta\cap\{r>R\}.
\end{equation*}
In the conical end, $u(r,\theta):=G_0^\Omega(re^{i\theta},y)$ after separation of variables has the expansion
\begin{equation*}
u(r,\theta)
=
\sum_{n=1}^\infty b_n(y)
\left(\frac{R}{r}\right)^{n\pi/\beta}
\sin\!\left(\frac{n\pi\theta}{\beta}\right),
\end{equation*}
where
\begin{equation*}
b_n(y)
=
\frac{2}{\beta}
\int_0^\beta
u(R,\vartheta)
\sin\!\left(\frac{n\pi\vartheta}{\beta}\right)
\,d\vartheta.
\end{equation*}
Positivity gives $b_1(y)>0$, and therefore
\begin{equation*}
G_0^\Omega(re^{i\theta},y)
=
c_\Omega(y)r^{-\pi/\beta}
\sin\!\left(\frac{\pi\theta}{\beta}\right)
+
O(r^{-2\pi/\beta}),
\qquad c_\Omega(y)>0.
\end{equation*}
Thus compact smoothing preserves the leading exponent and the threshold
classification.
\end{remark}

\subsubsection{The limiting extension \(\alpha=\infty\) at threshold}
The unperturbed Dirichlet Laplacian is the limiting member of the one-centre
extension family:
\begin{equation*}
A_{\infty,y}^\Omega=A_D^\Omega .
\end{equation*}
Equivalently, in the decomposition
$
u=u_z+qG_z^\Omega(\cdot,y),
\quad
u_z\in\operatorname{dom}(A_D^\Omega),
$
the condition \(\alpha=\infty\) forces the charge \(q\) to vanish. Hence the
threshold states of \(A_D^\Omega\), if any, must be regular harmonic functions
in \(\Omega\), with zero Dirichlet trace, and with the appropriate 
behavior at infinity. 
In the model geometries considered above, the threshold behavior of
\(A_D^\Omega\) is as follows.
\vskip10pt
\begin{center}
\begin{tabular}{c|c|c}
\(\Omega\) & threshold state for \(A_D^\Omega\) & threshold type \\
\hline
\(\mathbb R^3_+\) & none & regular threshold \\
\(\mathbb R^2_+\) & none & regular threshold \\
three-dimensional exterior domain & none in the usual \(3d\) class &
regular threshold \\
two-dimensional exterior domain & logarithmic harmonic state &
logarithmic resonance \\
planar sector \(\mathcal S_\beta\) & none &
regular threshold
\end{tabular}
\end{center}
\vskip10pt

Let us explain the entries in the table. First, zero is not an
\(L^2\)-eigenvalue of the unperturbed Dirichlet Laplacian in any of these
unbounded model domains. Indeed, if
\begin{equation*}
A_D^\Omega u=0,
\qquad
u\in\operatorname{dom}(A_D^\Omega),
\end{equation*}
then
\begin{equation*}
0=(A_D^\Omega u,u)_{L^2(\Omega)}
=\int_\Omega |\nabla u|^2\,dx,
\end{equation*}
and hence \(u\) is constant. The Dirichlet condition and square integrability
force \(u=0\).

In the half-space and half-plane, there is no threshold resonance either. In
these models, a threshold state would be a harmonic
function with zero Dirichlet trace on the boundary plane and with the relevant behavior at infinity. 
The standard $L^2$-weighted spaces exclude all non-zero harmonic polynomials compatible
with the odd (Dirichlet b.c.) symmetry. Thus zero is a regular threshold for the unperturbed Dirichlet
Laplacian in both \(\mathbb R^3_+\) and \(\mathbb R^2_+\). 
For the exterior sphere
$
\Omega=\mathbb R^3\setminus\overline{B_R},
$
the regular radial harmonic function with zero boundary trace is
\begin{equation*}
h_\Omega(x)=1-\frac{R}{|x|}.
\end{equation*}
It cannot be a resonance in the usual three-dimensional threshold class
$
\bigcap_{s>1/2}L^2_{-s}(\Omega),
$
because it tends to a non-zero constant at infinity, and a constant tail belongs
to \(L^2_{-s}(\Omega)\) in dimension three only for \(s>3/2\). Thus zero is a
regular threshold for the unperturbed exterior Dirichlet Laplacian in dimension
three.
The exterior disk is exceptional. For
$
\Omega=\mathbb R^2\setminus\overline{B_R},
$
the regular radial harmonic function with zero boundary trace is
\begin{equation*}
h_\Omega(x)=\log\frac{|x|}{R}.
\end{equation*}
It is not square-integrable, but it satisfies
$
h_\Omega\in L^2_{-s}(\Omega)
\quad\text{for every }s>1.
$
Thus zero is a logarithmic threshold resonance of the unperturbed
Dirichlet Laplacian in the exterior disk. This is the only threshold
singularity which survives in the limiting extension \(\alpha=\infty\) among the
standard models listed above.
The sector case is slightly different. Let
\begin{equation*}
\mathcal S_\beta
=
\{(r,\theta):r>0,\ 0<\theta<\beta\},
\qquad
q=\frac{\pi}{\beta},
\end{equation*}
with \(0<\beta<2\pi\). The separated zero-energy Dirichlet harmonic modes in
the first angular channel are generated by
\begin{equation*}
r^q\sin(q\theta)
\qquad\text{and}\qquad
r^{-q}\sin(q\theta).
\end{equation*}
For the Friedrichs extension of the Dirichlet Laplacian, neither mode gives an admissible
threshold state. The growing mode \(r^q\sin(q\theta)\) is regular at the
vertex, but is too large at infinity; indeed
\begin{equation*}
r^q\sin(q\theta)
\in L^2_{-s}(\mathcal S_\beta\cap\{r>1\})
\quad\Longleftrightarrow\quad
s>1+q.
\end{equation*}
Thus it does not belong to 
\(\bigcap_{s>1}L^2_{-s}\). The decaying mode
\(r^{-q}\sin(q\theta)\) has the correct conical tail at infinity, since
\begin{equation*}
r^{-q}\sin(q\theta)
\in L^2_{-s}(\mathcal S_\beta\cap\{r>1\})
\quad\Longleftrightarrow\quad
s>1-q,
\end{equation*}
but it is singular at the vertex. More precisely,
$
\nabla\bigl(r^{-q}\sin(q\theta)\bigr)
\sim r^{-q-1}
\quad\text{as } r\downarrow0,
$
and hence it is not in the Friedrichs form domain near the vertex. This remains
true in the non-convex case \(\beta>\pi\): then \(q<1\), so
\(r^{-q}\sin(q\theta)\) is locally square-integrable at the vertex, but its
gradient is still not locally square-integrable. In conclusion, for every
\(0<\beta<2\pi\), zero is a regular threshold for the
unperturbed Dirichlet Laplacian on \(\mathcal S_\beta\).

After the above classification we can conclude that the singular behavior of the threshold in nontrivial point interactions is associated to the presence of the singularity.

\subsection{Low-energy resolvent expansions}
\label{subsec:low-energy-resolvent}

The Kre\u{\i}n formula gives a direct operator-theoretic counterpart of
the classification in the preceding subsection.  We use the negative-energy
wave number $\kappa>0$, so that the spectral parameter is $-\kappa^2$.
This avoids branch conventions; the outgoing expansion is recovered by the
substitution $\kappa=-ik$, $\Im k>0$.

In this subsection $A_{\alpha,y}$ denotes the one-centre point-interaction
Hamiltonian in the relevant dimension, namely $A_{\alpha,y}^{\Omega}$ in
dimension three and $A_{\alpha,y}^{(2),\Omega}$ in dimension two.  Similarly,
$\alpha_c(y)$ denotes the corresponding critical coupling:
$\alpha_c^\Omega(y)$ in dimension three and
$\alpha_c^{(2),\Omega}(y)$ in dimension two.  We set
\begin{equation*}
R_D(z):=(A_D-z)^{-1},
\qquad
G_z^y:=G_z^\Omega(\cdot,y),
\qquad
\psi_0:=G_0^\Omega(\cdot,y),
\end{equation*}
and write
\begin{equation*}
A_{c,y}:=A_{\alpha_c(y),y},
\qquad
D_y(\kappa):=\alpha_c(y)-M_y^\Omega(-\kappa^2).
\end{equation*}
For $\kappa>0$, the Kre\u{\i}n formula is the exact identity
\begin{equation}
\label{eq:critical-krein-negative-energy}
(A_{c,y}+\kappa^2)^{-1}-R_D(-\kappa^2)
=
\frac{1}{D_y(\kappa)}
|G_{-\kappa^2}^y\rangle\langle G_{-\kappa^2}^y|.
\end{equation}
Thus the additional threshold singularity created by the point interaction is
entirely determined by the denominator $D_y(\kappa)$. Throughout this
subsection the displayed expansions concern the resolvent difference in
\eqref{eq:critical-krein-negative-energy}; the background Dirichlet resolvent
may have its own threshold behavior, which is not included in the displayed
rank-one singular terms.

\begin{proposition}
\label{prop:rank-one-threshold-reduction}
Let $d\in\{2,3\}$, and assume that the finite critical coupling and the
zero-energy Green function are defined.  Then, for every $\kappa>0$,
\begin{equation}
\label{eq:threshold-denominator-weyl-identity}
D_y(\kappa)
=
\kappa^2
\int_\Omega
G_{-\kappa^2}^\Omega(x,y)G_0^\Omega(x,y)\,dx
>0.
\end{equation}
Moreover,
\begin{equation*}
0<G_{-\kappa^2}^\Omega(x,y)\leq G_0^\Omega(x,y),
\qquad x\neq y,
\end{equation*}
and, whenever $\psi_0\in L_{-s}^2(\Omega)$,
\begin{equation}
\label{eq:green-vector-weighted-convergence}
G_{-\kappa^2}^y\longrightarrow\psi_0
\quad\text{in }L_{-s}^2(\Omega)
\qquad(\kappa\downarrow0).
\end{equation}
Consequently, if
\begin{equation}
\label{eq:abstract-denominator-asymptotic}
D_y(\kappa)=d_y(\kappa)(1+o(1)),
\qquad d_y(\kappa)>0,\qquad d_y(\kappa)\longrightarrow0,
\end{equation}
then
\begin{equation}
\label{eq:abstract-rank-one-threshold-expansion}
(A_{c,y}+\kappa^2)^{-1}-R_D(-\kappa^2)
=
\frac{1}{d_y(\kappa)}
|\psi_0\rangle\langle\psi_0|
+o\bigl(d_y(\kappa)^{-1}\bigr)
\end{equation}
in $\mathcal B(L_s^2(\Omega),L_{-s}^2(\Omega))$.

If $\psi_0\in L^2(\Omega)$, then
\begin{equation}
\label{eq:eigenvalue-denominator-general}
D_y(\kappa)
=
\kappa^2\|\psi_0\|_{L^2(\Omega)}^2+o(\kappa^2),
\end{equation}
and
\begin{equation}
\label{eq:eigenvalue-pole-general}
(A_{c,y}+\kappa^2)^{-1}-R_D(-\kappa^2)
=
\frac{1}{\kappa^2}P_0+o(\kappa^{-2}),
\qquad
P_0:=
\frac{|\psi_0\rangle\langle\psi_0|}{\|\psi_0\|_{L^2(\Omega)}^2}.
\end{equation}
The expansion \eqref{eq:eigenvalue-pole-general} holds in
$\mathcal B(L^2(\Omega))$.
\end{proposition}

\begin{proof}
For $\kappa,\mu>0$, the identity on the Weyl functions gives
\begin{equation*}
M_y^\Omega(-\kappa^2)-M_y^\Omega(-\mu^2)
=
-(\kappa^2-\mu^2)
\bigl(G_{-\kappa^2}^y,G_{-\mu^2}^y\bigr)_{L^2(\Omega)}.
\end{equation*}
Letting $\mu\downarrow0$ yields
\eqref{eq:threshold-denominator-weyl-identity}.  The integral is understood
as a monotone limit; its finiteness follows from the identity itself.

The resolvent identity and positivity of the Dirichlet kernels imply that
$G_{-\kappa^2}^\Omega(x,y)$ increases pointwise to
$G_0^\Omega(x,y)$ as $\kappa\downarrow0$.  Hence
\eqref{eq:green-vector-weighted-convergence} follows from dominated
convergence.  The rank-one estimate
\begin{equation*}
\bigl\||u\rangle\langle u|-|v\rangle\langle v|\bigr\|
\leq
(\|u\|+\|v\|)\|u-v\|
\end{equation*}
in the corresponding weighted spaces, together with
\eqref{eq:critical-krein-negative-energy}, proves
\eqref{eq:abstract-rank-one-threshold-expansion}.

If $\psi_0\in L^2(\Omega)$, the same monotone convergence holds in
$L^2(\Omega)$.  Formula \eqref{eq:threshold-denominator-weyl-identity}
then gives \eqref{eq:eigenvalue-denominator-general}, and
\eqref{eq:eigenvalue-pole-general} follows from the Kre\u{\i}n formula.
\end{proof}

We now compute the denominator in the model geometries of this section.

\begin{proposition}
\label{prop:low-energy-three-dimensional-models}
Let $s>1/2$.
\begin{enumerate}
\item[(i)] For the Dirichlet half-space $\Omega=\R^3_+$, write
$y=(y',y_3)$.  Then
\begin{align}
D_y(\kappa)
&=
-\frac{1}{8\pi y_3}
+\frac{\kappa}{4\pi}
+\frac{e^{-2\kappa y_3}}{8\pi y_3}
\notag\\
&=
\frac{y_3}{4\pi}\kappa^2+O(\kappa^3).
\label{eq:half-space-critical-denominator-negative}
\end{align}
In particular,
\begin{equation*}
\|G_0^{\R^3_+}(\cdot,y)\|_{L^2(\R^3_+)}^2
=\frac{y_3}{4\pi},
\end{equation*}
and
\begin{equation}
\label{eq:half-space-resolvent-eigenpole}
(A_{c,y}+\kappa^2)^{-1}-R_D(-\kappa^2)
=
\frac{1}{\kappa^2}P_0+o(\kappa^{-2}).
\end{equation}
Thus the critical half-space state produces an eigenvalue pole.

\item[(ii)] For the exterior of the sphere $\Omega_R=\R^3\setminus\overline{B_R}$, let
$|y|=\rho>R$, and $p:=R/\rho$.  Then
\begin{equation}
\label{eq:exterior-sphere-critical-denominator-negative}
D_y(\kappa)
=
\frac{(1-p)^2}{4\pi}\,\kappa+O(\kappa^2).
\end{equation}
Consequently,
\begin{equation*}
(A_{c,y}+\kappa^2)^{-1}-R_D(-\kappa^2)
=
\frac{4\pi}{(1-p)^2\kappa}
|G_0^{\Omega_R}(\cdot,y)\rangle
\langle G_0^{\Omega_R}(\cdot,y)|
+o(\kappa^{-1})
\end{equation*}
in $\mathcal B(L_s^2(\Omega_R),L_{-s}^2(\Omega_R))$.
Thus the monopole threshold resonance produces a simple pole in the wave
number.
\end{enumerate}
\end{proposition}

\begin{proof}
For the half-space, formula~\eqref{eq:halfspace-weyl-3d} at
$z=-\kappa^2$ gives the first line of
\eqref{eq:half-space-critical-denominator-negative}; expansion of the
exponential gives the second.  Proposition~\ref{prop:rank-one-threshold-reduction}
then yields the norm and \eqref{eq:half-space-resolvent-eigenpole}.

For the exterior of the sphere, the partial-wave computation in
Subsection~4.2.1 gives
\begin{equation*}
h_{-\kappa^2}^{\Omega_R}(y,y)
=
h_0^{\Omega_R}(y,y)
-\frac{2p-p^2}{4\pi}\kappa+O(\kappa^2).
\end{equation*}
Since
\begin{equation*}
M_y^{\Omega_R}(-\kappa^2)
=-\frac{\kappa}{4\pi}-h_{-\kappa^2}^{\Omega_R}(y,y),
\end{equation*}
formula \eqref{eq:exterior-sphere-critical-denominator-negative} follows.
The resolvent expansion is then an immediate consequence of
Proposition~\ref{prop:rank-one-threshold-reduction}.
\end{proof}

\begin{proposition}
\label{prop:low-energy-two-dimensional-models}
Let $s>1$.
\begin{enumerate}
\item[(i)] Let $H=\R^2_+$, $y=(y_1,y_2)\in H$, and set
$\psi_0=G_0^{H,D}(\cdot,y)$.  Then
\begin{equation}
\label{eq:half-plane-critical-denominator-negative}
D_y(\kappa)
=
\frac{y_2^2\kappa^2}{2\pi}
\bigl(1-\gamma_E-\log(\kappa y_2)\bigr)
+O(\kappa^4|\log\kappa|).
\end{equation}
Consequently,
\begin{align}
&(A_{c,y}+\kappa^2)^{-1}-R_D(-\kappa^2)
\notag\\
&\qquad=
\frac{2\pi}{y_2^2\kappa^2
\bigl(1-\gamma_E-\log(\kappa y_2)\bigr)}
|\psi_0\rangle\langle\psi_0|
+o\bigl((\kappa^2|\log\kappa|)^{-1}\bigr).
\label{eq:half-plane-p-wave-resolvent}
\end{align}
This is the characteristic $p$-wave singularity.

\item[(ii)] Let $\Omega_R=\R^2\setminus\overline{B_R}$,
$|y|=\rho>R$, and set
\begin{equation*}
c:=\log\frac{\rho}{R},
\qquad
L_R(\kappa):=-\log\frac{\kappa R}{2}-\gamma_E.
\end{equation*}
Then
\begin{equation}
\label{eq:exterior-disk-critical-denominator-negative}
D_y(\kappa)
=
\frac{c^2}{2\pi L_R(\kappa)}(1+o(1)).
\end{equation}
Consequently, with $\psi_0=G_0^{\Omega_R}(\cdot,y)$,
\begin{equation}
\label{eq:exterior-disk-s-wave-resolvent}
(A_{c,y}+\kappa^2)^{-1}-R_D(-\kappa^2)
=
\frac{2\pi L_R(\kappa)}{c^2}
|\psi_0\rangle\langle\psi_0|
+o(|\log\kappa|).
\end{equation}
This is the characteristic logarithmic $s$-wave singularity.  If
$\phi_s:=(2\pi/c)\psi_0$, so that $\phi_s(x)\to1$ as
$|x|\to\infty$, then the leading term in
\eqref{eq:exterior-disk-s-wave-resolvent} is
\begin{equation*}
\frac{L_R(\kappa)}{2\pi}|\phi_s\rangle\langle\phi_s|.
\end{equation*}
\end{enumerate}
\end{proposition}

\begin{proof}
For the half-plane, formula~\eqref{eq:halfplane-weyl-2d} and the
small-argument expansion of the Macdonald function $K_0$, see
\cite[(10.31.2)]{DLMF},
\begin{equation*}
K_0(t)
=
-\left(\log\frac{t}{2}+\gamma_E\right)
+\frac{t^2}{4}
\left(1-\gamma_E-\log\frac{t}{2}\right)
+O(t^4|\log t|)
\end{equation*}
give \eqref{eq:half-plane-critical-denominator-negative}.  Formula
\eqref{eq:half-plane-p-wave-resolvent} follows from
Proposition~\ref{prop:rank-one-threshold-reduction}.

For the exterior disk, only the zero angular mode contributes at order
$L_R(\kappa)^{-1}$.  Indeed,
\begin{equation*}
K_0(\kappa\rho)
=L_R(\kappa)-c+o(1),
\qquad
\frac{I_0(\kappa R)}{K_0(\kappa R)}
=\frac{1}{L_R(\kappa)}(1+o(1)),
\end{equation*}
and therefore
\begin{align*}
\frac{I_0(\kappa R)}{K_0(\kappa R)}K_0(\kappa\rho)^2
&=
\frac{(L_R(\kappa)-c)^2}{L_R(\kappa)}
+o(L_R(\kappa)^{-1})
\\
&=
L_R(\kappa)-2c+
\frac{c^2}{L_R(\kappa)}
+o(L_R(\kappa)^{-1}).
\end{align*}
The nonzero angular modes have finite limits and errors
$o(L_R(\kappa)^{-1})$.  Combining this with the free logarithmic term in
the Weyl function and with the value of $\alpha_c^{(2),\Omega_R}(y)$
gives \eqref{eq:exterior-disk-critical-denominator-negative}.  The
resolvent expansion follows again from
Proposition~\ref{prop:rank-one-threshold-reduction}.
\end{proof}

\begin{remark}
The preceding two propositions give the standard orders in
three and two dimensions; compare \cite{JK79,CorneanMichelangeliYajima2019}.
Planar sectors show that conical ends produce further aperture-dependent
orders, which are computed next.
\end{remark}
\begin{proposition}
\label{prop:low-energy-planar-wedges}
Let $\mathcal S_\beta$, $y=\rho e^{i\varphi}$, and
$q=\pi/\beta$ be as in
Theorem~\ref{thm:planar-wedge-threshold}.  Set
\begin{equation*}
D_{\beta,y}(\kappa)
:=
\alpha_c^{(2),\mathcal S_\beta}(y)
-M_y^{\mathcal S_\beta}(-\kappa^2).
\end{equation*}
Then:
\begin{enumerate}
\item[(i)] if $q>1$,
\begin{equation*}
D_{\beta,y}(\kappa)
=
\kappa^2\|G_0^{\mathcal S_\beta}(\cdot,y)\|_{L^2(\mathcal S_\beta)}^2
+o(\kappa^2),
\end{equation*}
and the critical resolvent has the eigenvalue term
$\kappa^{-2}P_0$;

\item[(ii)] if $q=1$, one recovers the half-plane expansion
\eqref{eq:half-plane-critical-denominator-negative} and the
$p$-wave singularity $(\kappa^2|\log\kappa|)^{-1}$;

\item[(iii)] if $0<q<1$,
\begin{equation}
\label{eq:wide-wedge-denominator}
D_{\beta,y}(\kappa)
=
C_{\beta,y}\kappa^{2q}(1+o(1)),
\end{equation}
where
\begin{equation*}
C_{\beta,y}
=
\frac{\rho^{2q}\sin^2(q\varphi)}
{\pi 2^{2q}}
\frac{\Gamma(1-q)}{\Gamma(1+q)}
>0.
\end{equation*}
Consequently, for every $s>1-q$,
\begin{align}
(A_{c,y}+\kappa^2)^{-1}-R_D(-\kappa^2)=
\frac{1}{C_{\beta,y}\kappa^{2q}}
|G_0^{\mathcal S_\beta}(\cdot,y)\rangle
\langle G_0^{\mathcal S_\beta}(\cdot,y)|
+o(\kappa^{-2q})
\label{eq:wide-wedge-resolvent-singularity}
\end{align}
in $\mathcal B(L_s^2(\mathcal S_\beta),L_{-s}^2(\mathcal S_\beta))$.
\end{enumerate}
\end{proposition}
\begin{remark}
Notice that the resolvent singularity varies continuously with the opening angle of the wedge:
for $\beta>\pi$ its order is $\kappa^{-2\pi/\beta}$.
\end{remark}
\begin{proof}
Part~\textup{(i)} follows from
Proposition~\ref{prop:rank-one-threshold-reduction}, and part~\textup{(ii)}
is the half-plane case.  Assume $0<q<1$.  The negative-energy Green
kernel in the wedge has the separation-of-variables expansion
\begin{equation*}
G_{-\kappa^2}^{\mathcal S_\beta}(r,\theta;\rho,\varphi)
=
\frac{2}{\beta}
\sum_{n=1}^\infty
\sin(nq\theta)\sin(nq\varphi)
I_{nq}(\kappa r_<)K_{nq}(\kappa r_>).
\end{equation*}
For $r>\rho$, the first angular mode of the zero-energy Green function is
\begin{equation*}
\frac{\rho^q\sin(q\varphi)}{\pi}
\,r^{-q}\sin(q\theta).
\end{equation*}
By angular orthogonality, the divergent part of the integral in
\eqref{eq:threshold-denominator-weyl-identity} is therefore
\begin{equation*}
\frac{\rho^q\sin^2(q\varphi)}{\pi}
I_q(\kappa\rho)
\int_\rho^\infty K_q(\kappa r)r^{1-q}\,dr.
\end{equation*}
Using the small-argument expansion of the modified Bessel function
$I_q$ and the standard integral of the modified Bessel function $K_q$, see
\cite[(10.30.1), (10.43.19)]{DLMF},
\begin{equation*}
I_q(t)=\frac{(t/2)^q}{\Gamma(1+q)}(1+o(1)),
\end{equation*}
and
\begin{equation*}
\int_0^\infty t^{1-q}K_q(t)\,dt
=
2^{-q}\Gamma(1-q),
\qquad 0<q<1,
\end{equation*}
one obtains
\begin{equation*}
\int_{\mathcal S_\beta}
G_{-\kappa^2}^{\mathcal S_\beta}(x,y)G_0^{\mathcal S_\beta}(x,y)\,dx
=
C_{\beta,y}\kappa^{2q-2}(1+o(1)).
\end{equation*}
The contribution of bounded radial regions is $O(1)$ in the integral in
\eqref{eq:threshold-denominator-weyl-identity}, hence it gives
$O(\kappa^2)$ after multiplication by $\kappa^2$, which is lower order than
$\kappa^{2q}$ because $q<1$. The $n$-th angular mode gives
$O(\kappa^{2nq})$ when $nq<1$, $O(\kappa^2|\log\kappa|)$ when
$nq=1$, and $O(\kappa^2)$ when $nq>1$. In all cases, for $n\geq2$, this
is of lower order than the first mode contribution $C_{\beta,y}\kappa^{2q}$.
Formula \eqref{eq:wide-wedge-denominator} now follows from
\eqref{eq:threshold-denominator-weyl-identity}, and
\eqref{eq:wide-wedge-resolvent-singularity} follows from
Proposition~\ref{prop:rank-one-threshold-reduction}.
\end{proof}

\begin{table}[H]
\centering
\small
\resizebox{\textwidth}{!}{%
\begin{tabular}{c|c|c|c}
model or end & critical state & $D_y(\kappa)=\alpha_c-M_y(-\kappa^2)$ & threshold type\\
\hline\\
3D exterior &
$|x|^{-1}$ tail &
$c_y\kappa$ &
monopole resonance\\[0.6ex]
3D half-space type &
$|x|^{-2}$ tail, $L^2$ &
$C_y\kappa^2$ &
threshold eigenvalue\\[0.6ex]
2D exterior disk &
nonzero constant tail &
$C_y/\log(1/\kappa)$ &
$s$-wave resonance, disappearing\\[0.6ex]
2D half-plane &
dipole $|x|^{-1}$ &
$C_y\kappa^2\log(1/\kappa)$ &
$p$-wave resonance, persistent\\[0.6ex]
2D sector, $0<\beta<\pi$ &
$r^{-q}$ tail, $L^2$ &
$C_y\kappa^2$ &
threshold eigenvalue\\[0.6ex]
2D sector, $\pi<\beta<2\pi$ &
$r^{-q}$ tail, not $L^2$ &
$C_{\beta,y}\kappa^{2q}$ &
conical resonance
\end{tabular}}
\caption{Threshold type and leading denominator in the model geometries. Here $q=\pi/\beta$, and the constants in the third column are positive.}
\label{tab:model-threshold-denominators}
\end{table}

\subsection{Persistence and disappearance of eigenvalues in two dimensions}
\label{subsec:persistence-disappearance-models}

Christiansen, Datchev, and Griffin in \cite{ChristiansenDatchevGriffin} distinguish negative eigenvalues that
approach zero and continue across the threshold as resonances from those
that disappear at the threshold; the former are called \emph{persistent}
and the latter \emph{disappearing} \cite{ChristiansenDatchevGriffin}.
Their results concern smooth families of compactly supported potentials on
$\R^2$, and therefore do not apply directly to the present family of
singular self-adjoint extensions on domains with boundary.  Nevertheless,
the same distinction can be formulated exactly in the two model geometries
above, exploiting the fact that the poles of the resolvent are the zeros of a denominator in the Kre\u{\i}n formula.

Fix one of the two-dimensional model domains below and set
\begin{equation*}
\varepsilon:=\alpha-\alpha_c^{(2),\Omega}(y).
\end{equation*}
For $\varepsilon<0$, let
\begin{equation*}
E_\varepsilon=-\kappa_\varepsilon^2<0
\end{equation*}
be the unique eigenvalue.  We use the same symbol
$M_y^\Omega(k^2)$ for the outgoing continuation of the Weyl function to
the logarithmic sheet
\begin{equation*}
-\frac{\pi}{2}\leq \arg k<\frac{3\pi}{2}.
\end{equation*}
The pole equation is
\begin{equation}
\label{eq:continued-pole-equation}
\alpha-M_y^\Omega(k^2)
=
\varepsilon+\alpha_c^{(2),\Omega}(y)-M_y^\Omega(k^2)=0.
\end{equation}
For $\varepsilon<0$, the zero $k=i\kappa_\varepsilon$ of
\eqref{eq:continued-pole-equation} is the negative eigenvalue.  In the
present rank-one setting, we say that this eigenvalue branch is persistent
if it extends continuously through $\varepsilon=0$ to zeros
$k_\varepsilon\to0$ of \eqref{eq:continued-pole-equation} for
$\varepsilon>0$; the corresponding outgoing Green vectors are then
resonant states.  Otherwise the branch is called disappearing.

In the following proposition we use the Lambert function in the discussion of persistence of
planar eigenvalue branches. The Lambert function is the multivalued inverse of
\begin{equation*}
w\longmapsto we^w;
\end{equation*}
its branches are denoted by $W_m$, $m\in\mathbb Z$, and satisfy
\begin{equation*}
W_m(z)e^{W_m(z)}=z.
\end{equation*}
The principal branch is $W_0$. The branch $W_{-1}$ is real-valued on
$(-e^{-1},0)$, with values in $(-\infty,-1]$. In the continuation formulas
below, $W_{-1}$ is analytically continued around the origin through the lower
half-plane. Thus, for $t>0$ small, this continued branch has the asymptotic
form
\begin{equation*}
W_{-1}(t)
=
\log t-2\pi i
-
\log(\log t-2\pi i)
+o(1),
\qquad t\downarrow0.
\end{equation*}
This is the branch relevant to the outgoing logarithmic sheet used in the
two-dimensional half-plane model; see \cite[Section~4.13]{DLMF}.

\begin{proposition}
\label{prop:persistence-disappearance-models}
The critical eigenvalue branch of the exterior disk is disappearing, while
the critical eigenvalue branch of the half-plane is persistent.  More
precisely, the following assertions hold.

\begin{enumerate}
\item[(i)] Let
$\Omega_R=\R^2\setminus\overline{B_R}$, let $|y|=\rho>R$, and set
\begin{equation*}
c:=\log\frac{\rho}{R}.
\end{equation*}
As $\varepsilon\to0^-$,
\begin{equation}
\label{eq:disk-eigenvalue-disappearance-asymptotic}
|\varepsilon|\log\frac{1}{|E_\varepsilon|}
\longrightarrow \frac{c^2}{\pi}.
\end{equation}
For $\varepsilon>0$ sufficiently small,
\eqref{eq:continued-pole-equation} has no root tending to zero on the
chosen logarithmic sheet.  Hence the eigenvalue disappears at the
$s$-wave threshold.

\item[(ii)] Let $H=\R^2_+$, let $y=(y_1,y_2)\in H$, and write
$h:=y_2$.  As $\varepsilon\to0^-$,
\begin{equation}
\label{eq:half-plane-eigenvalue-persistence-asymptotic}
|E_\varepsilon|
=
\frac{4\pi|\varepsilon|}
{h^2|\log|\varepsilon||}
\bigl(1+o(1)\bigr).
\end{equation}
Moreover, the zero $k=i\kappa_\varepsilon$ continues through
$\varepsilon=0$ to a resonance pole $k_\varepsilon\to0$ for
$\varepsilon>0$.  At leading order,
\begin{equation}
\label{eq:half-plane-persistent-lambert}
k_\varepsilon^2
=
-\frac{4\pi\varepsilon}
{h^2 W_{-1}\!\left(4\pi e^{2(\gamma_E-1)}\varepsilon\right)}
\bigl(1+o(1)\bigr).
\end{equation}
Here, for \(\varepsilon>0\), \(W_{-1}\) denotes the continuation of the real
branch through the lower half-plane described above. The pole
\(k_\varepsilon\) is a resonance.  Hence the eigenvalue persists through
the $p$-wave threshold.
\end{enumerate}
\end{proposition}

\begin{proof}
For the exterior disk, the Bessel expansion leading to
\eqref{eq:exterior-disk-critical-denominator-negative} continues to the
chosen logarithmic sheet and gives
\begin{equation}
\label{eq:disk-outgoing-denominator}
\alpha_c^{(2),\Omega_R}(y)-M_y^{\Omega_R}(k^2)
=
\frac{c^2}{2\pi L_R(-ik)}\bigl(1+o(1)\bigr),
\qquad
L_R(-ik):=-\log\frac{-ikR}{2}-\gamma_E.
\end{equation}
For $\varepsilon<0$, inserting $k=i\kappa_\varepsilon$ in
\eqref{eq:continued-pole-equation} yields
\begin{equation*}
|\varepsilon|
=
\frac{c^2}{2\pi L_R(\kappa_\varepsilon)}\bigl(1+o(1)\bigr).
\end{equation*}
Since $E_\varepsilon=-\kappa_\varepsilon^2$, this is equivalent to
\eqref{eq:disk-eigenvalue-disappearance-asymptotic}.  On the fixed
logarithmic sheet, as $k\to0$,
\begin{equation*}
\Re L_R(-ik)\longrightarrow+\infty,
\qquad
\Im L_R(-ik)=O(1).
\end{equation*}
Consequently the leading term in \eqref{eq:disk-outgoing-denominator} has
positive real part for $k$ sufficiently small.  It cannot cancel a
positive $\varepsilon$ in \eqref{eq:continued-pole-equation}.  Thus no
resonance pole tends to zero when $\varepsilon\downarrow0$, which proves
part~\textup{(i)}.

For the half-plane, analytic continuation of
\eqref{eq:half-plane-critical-denominator-negative} gives
\begin{equation}
\label{eq:half-plane-outgoing-denominator}
\alpha_c^{(2),H}(y)-M_y^H(k^2)
=
\frac{h^2k^2}{2\pi}
\bigl(\log(-ikh)+\gamma_E-1\bigr)
+O(k^4|\log k|).
\end{equation}
At $k=i\kappa_\varepsilon$, the eigenvalue equation and
\eqref{eq:half-plane-outgoing-denominator} imply
\eqref{eq:half-plane-eigenvalue-persistence-asymptotic}.

To continue the pole, first omit the error in
\eqref{eq:half-plane-outgoing-denominator} and set $w=-ikh$.  The
resulting equation is
\begin{equation*}
w^2\bigl(\log w+\gamma_E-1\bigr)=2\pi\varepsilon.
\end{equation*}
Writing $t=w^2$ reduces it to
\begin{equation*}
t\bigl(\log t+2(\gamma_E-1)\bigr)=4\pi\varepsilon,
\end{equation*}
which is solved by
\begin{equation*}
t=
\frac{4\pi\varepsilon}
{W_{-1}\!\left(4\pi e^{2(\gamma_E-1)}\varepsilon\right)}.
\end{equation*}
Since $k^2=-t/h^2$, this gives the leading term in
\eqref{eq:half-plane-persistent-lambert}.  For $\varepsilon<0$, the
real branch $W_{-1}$ produces $k=i\kappa_\varepsilon$.  Continuing
this branch below the origin gives, for $\varepsilon>0$, a value of \(t\)
such that \(k^2=-t/h^2\) lies in the fourth quadrant; choosing the
corresponding square root gives a pole \(k_\varepsilon\to0\) on the
outgoing sheet.  At this scale
\begin{equation*}
k^4|\log k|=o(|\varepsilon|),
\end{equation*}
and the error term is smaller than the leading denominator on a
small circle around the zero of the truncated equation. Rouché's theorem
therefore shows that the exact denominator has a zero in this circle. This
zero has the asymptotic \eqref{eq:half-plane-persistent-lambert}.
This proves part~\textup{(ii)}.
\end{proof}

Thus the two point-interaction models reproduce the correspondence found
for regular planar potentials in \cite{ChristiansenDatchevGriffin}.
Notice however that the agreement is not an application of the results of
\cite{ChristiansenDatchevGriffin}: here the perturbation parameter changes
a singular self-adjoint extension, and the conclusion follows from the
explicit consideration of the Weyl functions appearing in the Kre\u{\i}n formula.  Extending this analysis to general
exterior or asymptotically flat domains would require a corresponding
outgoing low-energy expansion and is left for future work.

\begin{remark}
The planar wedges show that
the persistent behavior is not exhausted by the $p$-wave
half-plane case.  Set $q=\pi/\beta$.  If $0<\beta<\pi$, then $q>1$ and the critical
state is a threshold eigenfunction.  In this case the denominator has
the  form
\begin{equation*}
\alpha_c^{(2),\mathcal S_\beta}(y)-M_y^{\mathcal S_\beta}(-\kappa^2)
=
\kappa^2\|G_0^{\mathcal S_\beta}(\cdot,y)\|_{L^2(\mathcal S_\beta)}^2
+o(\kappa^2),
\qquad \kappa\downarrow0.
\end{equation*}
Thus, at the level of the leading Kre\u{\i}n denominator, the eigenvalue
branch has the persistent behavior.  If
$\beta=\pi$, one recovers the half-plane and hence the $p$-wave persistent
branch of Proposition~\ref{prop:persistence-disappearance-models}.  If
$\pi<\beta<2\pi$, then $0<q<1$ and Proposition~\ref{prop:low-energy-planar-wedges}
gives the law
\begin{equation*}
\alpha_c^{(2),\mathcal S_\beta}(y)-M_y^{\mathcal S_\beta}(-\kappa^2)
=
C_{\beta,y}\kappa^{2q}(1+o(1)).
\end{equation*}

After outgoing continuation, obtained by setting $\kappa=-ik$ and fixing a
branch of $\log(-ik)$, the continued denominator has the leading form
\begin{equation*}
\alpha-\alpha_c
+
C_{\beta,y}(-ik)^{2q}
+
o(k^{2q}).
\end{equation*}
Thus the corresponding pole branch is governed by the same exponent $2q$.
The branch choice is needed because the sectors produce fractional
powers of the spectral parameter.
\end{remark}
\begin{remark}
In dimension three, the model examples considered here show persistence of
eigenvalues. In the exterior-sphere model, after outgoing continuation, a
persistent virtual state, also called an anti-bound state, appears and is
associated with the monopole threshold resonance; see
\eqref{eq:exterior-sphere-critical-denominator-negative}. In the half-space
model, a persistent complex resonance pair occurs and is associated with a
threshold eigenvalue; see Remark~3.8 in \cite{NojaRasoStoia25}. The
disappearing mechanism observed in the two-dimensional exterior disk is tied
to the logarithmic low-energy structure of dimension two and seems to have no
direct analogue in these three-dimensional one-centre models.
\end{remark}


\appendix

\section{Auxiliary analytic results}
\label{app:auxiliary}
The material in this appendix is standard, or follows by straightforward adaptations of standard arguments. Since we do not know a reference treating exactly the two classes of domains considered here, we include the details for completeness.

\subsection{Proof of the Dirichlet resolvent kernel theorem}
\label{app:dirichlet-kernel-proof}
\begin{proof}[Proof of Theorem~\ref{thm:dirichlet-kernel}]
\medskip
Here we adapt the proof of Lemma 3.1 in \cite{BehrndtRohleder}.
Fix $y\in\Omega$. Since $y$ is an interior point, there exists $\varepsilon_y>0$ such that $\dist(y,\pa\Omega)>2\varepsilon_y$. Hence $G^0_\la(\cdot,y)$ is smooth on the $\varepsilon_y$-neighborhood of $\pa\Omega$.
If $\Omega$ is an exterior Lipschitz domain, then $\pa\Omega$ is compact. Therefore $\tr G^0_\la(\cdot,y)\in H^s(\pa\Omega)$ for every $s\in\R$, in particular $\tr G^0_\la(\cdot,y)\in H^{1/2}(\pa\Omega)$; compare the standard trace theory in \cite[Chapter~7]{AdamsFournier}.
Assume next that $\Omega$ is a special Lipschitz domain,
\begin{equation*}
\Omega=\{(x',x_d): x_d>\varphi(x')\},
\qquad \varphi\in W^{1,\infty}(\R^{d-1}).
\end{equation*}
Set $g_y(x'):=G^0_\la((x',\varphi(x')),y)$ for $x'\in\R^{d-1}$. Because $\Re\kap(\la)>0$ and $y$ has positive distance from the boundary, the explicit free kernels in \eqref{eq:free-resolvent-kernel}, together with their first derivatives, have exponential decay along the boundary graph. Since $\varphi$ is globally Lipschitz, the chain rule gives
\begin{equation*}
g_y\in H^1(\R^{d-1}).
\end{equation*}
By the standard chart definition of Sobolev spaces on Lipschitz hypersurfaces, this implies $\tr G^0_\la(\cdot,y)\in H^{1/2}(\pa\Omega)$. This is also consistent with the sharp trace theorem on special Lipschitz domains proved in \cite[Theorem~4.19]{Gaudin}. Thus, in both cases, the boundary datum in \eqref{eq:dirichlet-correction-problem} belongs to $H^{1/2}(\pa\Omega)$.
\medskip
Since $\la\in\rho(A_D)$, the Dirichlet boundary value problem
\begin{equation*}
(-\Delta-\la)u=0\text{ in }\Omega,
\qquad
\tr u=\phi\text{ on }\pa\Omega,
\end{equation*}
has a unique solution $u\in H^1(\Omega)$ for every $\phi\in H^{1/2}(\pa\Omega)$; this is precisely \cite[Lemma~3.1]{BehrndtRohleder}, and the associated Dirichlet solution map is encoded by the Poisson operator from \cite[Definition~3.2 and Lemma~3.3]{BehrndtRohleder}. Applying that result with $\phi=\tr G^0_\la(\cdot,y)$ produces the unique $h_\la^\Omega(\cdot,y)\in H^1(\Omega)$ solving \eqref{eq:dirichlet-correction-problem}. This proves part \textup{(1)}.
\medskip
Define now $G_\la^\Omega$ by \eqref{eq:dirichlet-kernel-domain}. Since $h_\la^\Omega(\cdot,y)$ satisfies the homogeneous equation, we have
\begin{equation*}
(-\Delta_x-\la)G_\la^\Omega(x,y)=(-\Delta_x-\la)G^0_\la(x,y)-(-\Delta_x-\la)h_\la^\Omega(x,y)=\delta_y.
\end{equation*}
Moreover, $\tr_x G_\la^\Omega(x,y)=\tr G^0_\la(\cdot,y)-\tr h_\la^\Omega(\cdot,y)=0$. Hence part \textup{(2)} follows.
\medskip
The correction term $h_\la^\Omega(\cdot,y)$ solves the homogeneous equation $(-\Delta-\la)h_\la^\Omega(\cdot,y)=0$ in all of $\Omega$. Hence, by standard interior elliptic regularity, see for instance \cite[Chapter~8]{GilbargTrudinger}, one has $h_\la^\Omega(\cdot,y)\in C^\infty_{\mathrm{loc}}(\Omega)$. Thus the only diagonal singularity of $G_\la^\Omega(x,y)=G^0_\la(x,y)-h_\la^\Omega(x,y)$ is the free one: in dimension three it is Coulombic and in dimension two it is logarithmic.
\medskip
Finally, let $f\in C_c^\infty(\Omega)$ and extend it by zero to $\widetilde f\in C_c^\infty(\R^d)$. Define
\begin{equation*}
u_0(x):=(-\Delta_{\R^d}-\la)^{-1}\widetilde f(x)
=
\int_\Omega G^0_\la(x,\xi)f(\xi)\,\dd \xi.
\end{equation*}
Then $u_0\in H^2_{\mathrm loc}(\R^d)$ and $(-\Delta-\la)u_0=f$ in $\Omega$. Because $\Re\kap(\la)>0$ and $f$ is compactly supported in the interior of $\Omega$, the function $u_0$ and its first derivatives decay exponentially at infinity. Therefore $\tr u_0\in H^{1/2}(\pa\Omega)$ both in the exterior and in the special Lipschitz case. Let $w\in H^1(\Omega)$ be the unique solution of $(-\Delta-\la)w=0$ in $\Omega$ with $\tr w=\tr u_0$ on $\pa\Omega$. Then $u:=u_0-w$ satisfies
\begin{equation*}
(-\Delta-\la)u=f\quad\text{in }\Omega,\qquad \tr u=0\quad\text{on }\pa\Omega,
\end{equation*}
hence $u=(A_D-\la)^{-1}f$ by uniqueness in the resolvent set. It remains to identify $w$ as the integral of the correction terms. Set $\phi_\xi:=\tr G^0_\la(\cdot,\xi)\in H^{1/2}(\pa\Omega)$. Since $\supp f$ is compact and contained in the interior of $\Omega$, the map $\xi\mapsto \phi_\xi$ is continuous from $\supp f$ into $H^{1/2}(\pa\Omega)$ and therefore Bochner integrable. By linearity and boundedness of the Dirichlet solution operator,
\begin{equation*}
w=\int_\Omega h_\la^\Omega(\cdot,\xi)\,f(\xi)\,\dd \xi.
\end{equation*}
Consequently,
\begin{equation*}
u(x)
=
\int_\Omega \bigl(G^0_\la(x,\xi)-h_\la^\Omega(x,\xi)\bigr)f(\xi)\,\dd \xi
=
\int_\Omega G_\la^\Omega(x,\xi)f(\xi)\,\dd \xi,
\end{equation*}
which is exactly \eqref{eq:dirichlet-kernel-representation}.

Finally, the Dirichlet Laplacian commutes with complex conjugation. Hence
\begin{equation*}
\overline{(A_D-\lambda)^{-1}f}
=
(A_D-\overline\lambda)^{-1}\overline f,
\end{equation*}
which gives, at the level of kernels,
\begin{equation*}
G_{\overline\lambda}^\Omega(x,y)
=
\overline{G_\lambda^\Omega(x,y)}.
\end{equation*}
On the other hand, self-adjointness gives
\begin{equation*}
(A_D-\lambda)^{-1*}
=
(A_D-\overline\lambda)^{-1}.
\end{equation*}
Using the kernel representation, this implies
\begin{equation*}
G_{\overline\lambda}^\Omega(x,y)
=
\overline{G_\lambda^\Omega(y,x)}.
\end{equation*}
Combining the two identities yields
\begin{equation*}
G_\lambda^\Omega(x,y)=G_\lambda^\Omega(y,x).
\end{equation*}
The equalities first hold in the distributional kernel sense and then pointwise away from the diagonal by interior smoothness. The theorem is proved.
\end{proof}

\subsection{Positivity of the regular part of the Green kernel at negative energies}
\label{app:sign-corrector}

The following observation is the only sign property of the regular part of the Green kernel needed in the paper.

\begin{lemma}\label{lem:sign-corrector}
Let $d\in\{2,3\}$, let $\Omega\subset\R^d$ be either an exterior Lipschitz domain or a special Lipschitz domain, let $y\in\Omega$, and let $\lambda>0$. Let $h_{-\lambda}^\Omega(\cdot,y)$ be the regular part of the Dirichlet Green function, that is,
\begin{equation*}
(-\Delta+\lambda)h_{-\lambda}^\Omega(\cdot,y)=0\quad\text{in }\Omega,
\qquad
\tr h_{-\lambda}^\Omega(\cdot,y)=\tr G^0_{-\lambda}(\cdot,y)\quad\text{on }\pa\Omega.
\end{equation*}
Then
\begin{equation*}
h_{-\lambda}^\Omega(x,y)\ge0\quad\text{for all }x\in\Omega,
\qquad
h_{-\lambda}^\Omega(y,y)\ge0.
\end{equation*}
\end{lemma}

\begin{proof}
The boundary datum $G^0_{-\lambda}(\cdot,y)|_{\partial\Omega}$ is real and nonnegative. By uniqueness of the Dirichlet problem, the corresponding solution $h_{-\lambda}^\Omega(\cdot,y)$ is real-valued. Set
\begin{equation*}
w:=\bigl(h_{-\lambda}^\Omega(\cdot,y)\bigr)^-.
\end{equation*}
Since $h_{-\lambda}^\Omega(\cdot,y)\in H^1(\Omega)$ and its trace is nonnegative, one has
$
w\in H_0^1(\Omega).
$
Testing the weak equation for $h_{-\lambda}^\Omega(\cdot,y)$ against $w$ gives
\begin{equation*}
0
=
\int_\Omega \nabla h_{-\lambda}^\Omega(x,y)\cdot \nabla w(x)\,dx
+
\lambda\int_\Omega h_{-\lambda}^\Omega(x,y)\,w(x)\,dx.
\end{equation*}
On the set $\{w>0\}$ one has $w=-h_{-\lambda}^\Omega(\cdot,y)$ and $\nabla w=-\nabla h_{-\lambda}^\Omega(\cdot,y)$, hence
\begin{equation*}
0
=
-\int_\Omega |\nabla w(x)|^2\,dx
-
\lambda\int_\Omega |w(x)|^2\,dx.
\end{equation*}
Therefore $w=0$, and thus $h_{-\lambda}^\Omega(x,y)\ge0$ almost everywhere in $\Omega$. Since $h_{-\lambda}^\Omega(\cdot,y)$ is a weak solution of
\begin{equation*}
(-\Delta+\lambda)u=0 \quad \text{in }\Omega,
\end{equation*}
standard interior elliptic regularity yields
\begin{equation*}
h_{-\lambda}^\Omega(\cdot,y)\in C^\infty(\Omega).
\end{equation*}
Hence the almost everywhere inequality upgrades to
$
h_{-\lambda}^\Omega(x,y)\ge 0 \quad \text{for all }x\in\Omega.
$
Evaluating at $x=y$ gives
\begin{equation*}
h_{-\lambda}^\Omega(y,y)\ge 0.
\end{equation*}
\end{proof}

\section*{Acknowledgments.}
\noindent D.N acknowledges the support of the Next Generation EU - Prin 2022 project "Singular Interactions and Effective Models in Mathematical Physics- 2022CHELC7" and of the INdAM-GNFM.  The first author is also grateful to Prof.Tanya Christiansen for elucidation about results in \cite{ChristiansenDatchevGriffin}, giving origin to the results in Subsection~\ref{subsec:persistence-disappearance-models}.

\end{document}